\documentclass[letterpaper,USenglish,cleveref,numberwithinsect,thm-restate]{no-lipics-v2022}
\usepackage[utf8]{inputenc}
\usepackage{color}
\usepackage{csquotes}
\usepackage{xspace}
\usepackage{amsmath}
\usepackage{amssymb}
\usepackage{mathtools}
\usepackage{bbm}
\usepackage{bm}
\usepackage[linesnumbered, noend]{algorithm2e}
\usepackage[draft]{fixme}
\usepackage{tikz}
\usetikzlibrary{arrows,arrows.meta,decorations.pathreplacing,decorations.pathmorphing,shapes,calc,patterns,shapes,matrix,math,quotes}

\nolinenumbers
\fxsetup{mode=multiuser,theme=color, layout=inline}
\FXRegisterAuthor{eg}{eg}{\color{red} {\bf  Egor}}
\FXRegisterAuthor{tk}{atk}{\color{blue} {\bf Tomasz}}

\newcommand{\Oh}{\ensuremath{\mathcal{O}}\xspace}
\newcommand{\Ohtilde}{\ensuremath{\smash{\rlap{\raisebox{-0.2ex}{$\widetilde{\phantom{\Oh}}$}}\Oh}}\xspace}
\newcommand{\Thtilde}{\ensuremath{\smash{\rlap{\raisebox{-0.2ex}{$\widetilde{\phantom{\Theta}}$}}\Theta}}\xspace}
\newcommand{\Ohtilda}{\Ohtilde}

\newcommand{\RR}{\mathbb{R}\xspace}
\newcommand{\R}{\RR}
\newcommand{\Real}{\RR}
\newcommand{\Rz}{\R_{\ge 0}}

\newcommand{\ZZ}{\mathbb{Z}\xspace}

\newcommand{\G}{\mathcal{G}}
\newcommand{\N}{\mathcal{N}}
\newcommand{\Tr}{\mathcal{T}}
\newcommand{\symb}{\mathsf{symb}}
\renewcommand{\S}{\mathcal{S}}

\renewcommand{\emptyset}{\varnothing}

\newcommand{\on}[1]{\operatorname{#1}}
\DeclareMathOperator*{\argmax}{arg\,max}

\DeclareMathOperator{\rhs}{rhs}

\DeclareMathOperator{\AG}{AG}

\DeclareMathOperator{\dist}{dist}
\newcommand{\len}{\mathsf{len}}
\newcommand{\plen}{\mathsf{plen}}
\newcommand{\pre}{\mathsf{par}}
\newcommand{\dep}{\mathsf{dep}}
\DeclareMathOperator{\BM}{BM}
\DeclareMathOperator{\poly}{poly}
\DeclareMathOperator{\ssum}{sum}
\newcommand{\tOh}{\Ohtilda}
\newcommand{\AGw}{\AG^w}

\newcommand{\oAGw}{\overline{\AG}^w}
\newcommand{\mds}{\mathsf{CMO}}
\newcommand{\meq}[1]{\stackrel{#1}{=}}
\newcommand{\perim}{\mathcal{P}}
\newcommand{\dd}{.\,.}
\newcommand{\core}[1]{\on{core}(#1)}
\newcommand{\coresiz}{\delta}
\newcommand{\dens}[1]{#1^{\square}}

\newcommand{\inp}{in}

\newcommand{\outp}{out}

\newcommand{\boxds}{\mathsf{D}^w}
\newcommand{\boxdsk}{\boxds_k}
\newcommand{\boxdsslp}{\hat{\mathsf{D}}^w}
\newcommand{\boxdskslp}{\boxdsslp_k}
\newcommand{\slp}{\mathcal{G}}

\newcommand{\distag}{\dist_{\oAGw(X, Y)}}

\newcommand{\BMw}{\BM^w}
\def\emptystring{\ensuremath\varepsilon}
\def\w#1#2{\ensuremath w(#1, #2)}
\def\myw#1#2#3{\ensuremath #1(#2, #3)}
\def\wi#1#2#3{\ensuremath w_{#1}(#2, #3)}
\def\ipmOpName{{\tt IPM}\xspace}
\def\accOpName{{\tt Access}\xspace}
\def\extractOpName{{\tt Extract}\xspace}
\def\lenOpName{{\tt Length}\xspace}

\newcommand{\modelname}{\texttt{PILLAR}\xspace}
\def\lceOp#1#2{{\tt LCP}(#1, #2)}
\def\lcbOp#1#2{{\tt LCP}^R(#1, #2)}

\def\accOp#1#2{#1\position{#2}}

\newcommand{\cA}{\mathcal{A}}
\newcommand{\cB}{\mathcal{B}}
\newcommand{\cC}{\mathcal{C}}
\newcommand{\cD}{\mathcal{D}}
\newcommand{\cO}{\mathcal{O}}
\newcommand{\Als}{\mathbf{A}}
\newcommand{\bA}{\mathbf{A}}
\newcommand{\bB}{\mathbf{B}}
\newcommand{\Ta}{\mathbb{T}}
\newcommand{\Ra}{\mathbb{S}}

\def\twoheadleadsto{\tikz[baseline=(a.base)]{\draw[%
    decorate,decoration={zigzag,segment length=4, amplitude=.9},%
    ] (0,0) -- (.25, 0);%
    \draw[%
    -{Classical TikZ Rightarrow}.{Classical TikZ Rightarrow},%
    ] (.25, 0) -- (.4, 0);%
    \node (a) at (.4/2,-.03) {\phantom{\(\leadsto\)}};%
}}

\newcommand{\onto}{\twoheadleadsto}
\def\aonto#1{\onto}

\newcommand{\fragmentco}[2]{[#1\dd #2)}
\newcommand{\fragmentoc}[2]{(#1\dd #2]}
\newcommand{\fragmentoo}[2]{(#1\dd #2)}
\newcommand{\fragmentcc}[2]{[#1\dd #2]}
\newcommand{\position}[1]{[#1]}

\newcommand{\Esigma}{\overline{\Sigma}}
\newcommand{\ed}{\mathsf{ed}}
\newcommand{\wed}{\ed^w}
\newcommand\sed{\mathsf{self}\text{-}\ed}

\newcommand{\floor}[1]{\lfloor{#1}\rfloor}
\newcommand{\ceil}[1]{\lceil{#1}\rceil}
\newcommand{\set}[1]{\{#1\}}
\newcommand{\setmax}[1]{\max\set{#1}}

\newcommand{\Zp}{\mathbb{Z}_{>0}}

\newcommand{\bC}{\bar{C}}
\newcommand{\bj}{\bar{\jmath}}
\newcommand{\bx}{\bar{x}}
\newcommand{\by}{\bar{y}}
\newcommand{\bk}{\bar{k}}
\newcommand{\hk}{\hat{k}}
\newcommand{\hn}{\hat{n}}
\newcommand{\hx}{\hat{x}}
\newcommand{\hy}{\hat{y}}
\newcommand{\bX}{\bar{X}}
\newcommand{\bY}{\bar{Y}}
\newcommand{\hX}{\hat{X}}
\newcommand{\hY}{\hat{Y}}
\newcommand{\hcA}{\hat{\cA}}
\newcommand{\kt}{\tilde{k}}

\newcommand{\cU}{\mathcal{U}}
\newcommand{\cV}{\mathcal{V}}
\newcommand{\eps}{\varepsilon}

\newcommand{\closure}[1]{\mathsf{cl}(#1)}

\newcommand{\says}[3]{}

\definecolor{darkgreen}{RGB}{0,160,0}
\definecolor{darkred}{RGB}{220,20,60}
\definecolor{darkblue}{RGB}{0,0,160}

\title{Bounded Edit Distance: Optimal Static and Dynamic Algorithms for Small Integer Weights}

\author{Egor Gorbachev}{Saarland University and Max Planck Institute for Informatics, Saarland Informatics Campus, Germany}{egorbachev@cs.uni-saarland.de}{https://orcid.org/0009-0005-5977-7986}{This work is part of the project TIPEA that has received funding from the European Research Council (ERC) under the European Unions Horizon 2020 research and innovation programme (grant agreement No.\ 850979).}

\author{Tomasz Kociumaka}{INSAIT, Sofia University ``St. Kliment Ohridski'', Sofia, Bulgaria}{tomasz.kociumaka@insait.ai}{https://orcid.org/0000-0002-2477-1702}{Partially funded by the Ministry of Education and Science of Bulgaria's support for INSAIT, Sofia University ''St. Kliment Ohridski'', as part of the Bulgarian National Roadmap for Research Infrastructure.}

\authorrunning{E. Gorbachev and T. Kociumaka} 

\Copyright{Egor Gorbachev and Tomasz Kociumaka} 

\acknowledgements{The second author would like to thank Itai Boneh (Reichman University) for helpful discussions.}

\begin{document}

\pagenumbering{gobble}

\maketitle

\begin{abstract}
The \emph{edit distance} (also known as the Levenshtein distance) of two strings is the minimum number of character insertions, deletions, and substitutions needed to transform one string into the other.
The textbook algorithm determines the edit distance of two length-$n$ strings in $\Oh(n^2)$ time, and one of the foundational results of fine-grained complexity is that any polynomial-factor improvement upon this quadratic runtime would violate the Orthogonal Vectors Hypothesis.
In the \emph{bounded} version of the problem, where the complexity is parameterized by the value $k$ of the edit distance, the classic algorithm of Landau and Vishkin [JCSS'88] achieves $\Oh(n+k^2)$ time, which is optimal (up to sub-polynomial factors and conditioned on OVH) as a function of $n$ and $k$.

While the Levenshtein distance is a fundamental theoretical notion, most practical applications use \emph{weighted edit distance}, where the weight (cost) of each edit can be an arbitrary real number in $[1,\infty)$ that may depend on the edit type and the characters involved.
Unfortunately, the Landau--Vishkin algorithm does not generalize to the weighted setting and, for many decades, a simple $\Oh(nk)$-time dynamic programming procedure remained the state of the art for \emph{bounded weighted edit distance}.
Only recently, Das, Gilbert, Hajiaghayi, Kociumaka, and Saha [STOC'23] provided an $\Oh(n+k^5)$-time algorithm; shortly afterward, Cassis, Kociumaka, and Wellnitz [FOCS'23] presented an $\Ohtilde(n+\sqrt{nk^3})$-time solution (where $\Ohtilde(\cdot)$ hides $\mathrm{poly}\log n$ factors) and proved this runtime optimal for $\sqrt{n} \le k \le n$ (up to sub-polynomial factors and conditioned on the All-Pairs Shortest Paths Hypothesis).

Notably, the hard instances constructed to establish the underlying conditional lower bound use fractional weights with large denominators, reaching $\poly(n)$, which stands in contrast to weight functions used in practice (e.g., in bioinformatics) that, after normalization, typically attain small integer values. 
Our first main contribution is a surprising discovery that the $\Ohtilde(n+k^2)$ running time of the Landau--Vishkin algorithm can be recovered if the weights are small integers (e.g., if some specific weight function is fixed in a given application).
In general, our solution takes $\Ohtilde(n+\min\{W,\sqrt{k}\}\cdot k^2)$ time for integer weights not exceeding a threshold $W$.
Despite matching time complexities, we do not use the Landau--Vishkin framework; instead, we build upon the recent techniques for arbitrary weights.
For this, we exploit extra structure following from integer weights and avoid further bottlenecks using several novel ideas to give faster and more robust implementations of multiple steps of the previous approach. 

Next, we shift focus to the \emph{dynamic} version of the \emph{unweighted} edit distance problem, which asks to maintain the edit distance of two strings that change dynamically, with each update modeled as a single edit (character insertion, deletion, or substitution).
For many years, the best approach for dynamic edit distance combined the Landau--Vishkin algorithm with a dynamic strings implementation supporting efficient substring equality queries, such as one by Mehlhorn, Sundar, and Uhrig [SODA'94]; the resulting solution supports updates in $\Ohtilde(k^2)$ time.
Recently, Charalampopoulos, Kociumaka, and Mozes [CPM'20] observed that a framework of Tiskin [SODA'10] yields a dynamic algorithm with an update time of $\Ohtilde(n)$.
This is optimal in terms of $n$: significantly faster updates would improve upon the static $\Oh(n^2)$-time algorithm and violate OVH.
With the state-of-the-art update time at $\Ohtilde(\min\{n,k^2\})$, an exciting open question is whether $\Ohtilde(k)$ update time is possible.
Our second main contribution is an affirmative answer to this question: we present a deterministic dynamic algorithm that maintains the edit distance in $\Ohtilde(k)$ worst-case time per update. 
Surprisingly, this result builds upon our new static solution and hence natively supports small integer weights: if they do not exceed a threshold $W$, the update time becomes~$\Ohtilde(W^2 k)$.%

\end{abstract}

\newpage
\tableofcontents
\newpage
\pagenumbering{arabic}

\section{Introduction}\label{sec:introduction}
Quantifying similarity between strings (sequences, texts) is a crucial algorithmic task applicable across many domains, with the most profound role in bioinformatics and computational linguistics.
The \emph{edit distance} (Levenshtein distance~\cite{Levenshtein66}) is among the most popular measures of string (dis)similarity. 
For two strings $X$ and $Y$, the edit distance $\ed(X,Y)$ is the minimum number of character insertions, deletions, and substitutions (jointly called \emph{edits}) needed to transform $X$ into~$Y$.
The textbook algorithm~\cite{Vin68,NW70,Sel74,WF74} takes $\Oh(n^2)$ time to compute the edit distance of two strings of length at most $n$.
Despite substantial efforts, its runtime has been improved only by a poly-logarithmic factor~\cite{MP80,Gra16}.
Fine-grained complexity provided an explanation for this quadratic barrier: any polynomial-factor improvement would violate the Orthogonal Vectors Hypothesis~\cite{ABW15,BK15,AHWW16,BI18} and thus also the Strong Exponential Time Hypothesis~\cite{IP01,IPZ01}.

An established way of circumventing this lower bound is to consider the \emph{bounded} edit distance problem, where the running time is expressed in terms of not only the length $n$ of the input strings but also the value $k$ of the edit distance. 
Building upon the ideas of Ukkonen~\cite{Ukk85} and Myers~\cite{Mye86}, Landau and Vishkin~\cite{LV88} presented an elegant $\Oh(n+k^2)$-time algorithm for bounded edit distance.
The fine-grained hardness of the unbounded version prohibits any polynomial-factor improvements upon this runtime: a hypothetical $\Oh(n+k^{2-\epsilon})$-time algorithm, even restricted to instances satisfying $k=\Theta(n^{\kappa})$ for some constant $\frac12 < \kappa \le 1$, would violate the Orthogonal Vectors Hypothesis.
Although bounded edit distance admits a decades-old optimal solution, the last few years brought numerous exciting developments across multiple models of computation.
This includes sketching and streaming algorithms~\cite{BZ16,JNW21,KPS21,BK23,KS24}, sublinear-time approximation algorithms~\cite{GKS19,KS20,GKKS22,BCFN22,BCFK24}, algorithms for preprocessed strings~\cite{GRS20,BCFN22b}, algorithms for compressed input~\cite{GKLS22}, and quantum algorithms~\cite{GJKT24}, just to mention a few settings.
Multiple papers have also studied generalizations of the bounded edit distance problem, including weighted edit distance~\cite{GKKS23,DGHKS23,CKW23}, tree edit distance~\cite{Tou07,AJ21,DGHKSS22,DGHKS23}, and Dyck edit distance \mbox{\cite{BO16,FGKKPS22,D23,DGHKS23}}.

In this work, we provide conditionally optimal algorithms for bounded edit distance with \emph{small integer weights} and settle the complexity of the \emph{dynamic} version of the bounded edit distance problem.

\paragraph*{Weighted Edit Distance}
Although the theoretical research on the edit distance problem has predominantly focused on the \emph{unweighted} Levenshein distance, most practical applications require a more general \emph{weighted} edit distance, where each edit is associated with a cost depending on the edit type and the characters involved.
Many early works~\cite{Sel74,WF74,Sel80} and application-oriented textbooks~\cite{Wat95,Gus97,JM08,MBCT2015} introduce edit distance already in the weighted variant.
Formally, it can be conveniently defined using a weight function $w : \Esigma^2 \to \mathbb{R}$, where $\Esigma=\Sigma\cup\{\eps\}$ denotes the alphabet extended with a special symbol $\eps$ representing the empty string (or the lack of a character).
For every $a,b\in \Sigma$, the cost of inserting $b$ is $w(\eps,b)$, the cost of deleting $a$ is $w(a,\eps)$, and the cost of substituting $a$ for $b$ is $w(a,b)$.
The weighted edit distance $\wed(X,Y)$ of two strings $X$ and $Y$ is the minimum total weight of edits transforming $X$ into $Y$.
Consistently with previous works, we assume that the weight function is normalized so that $w(a,a)=0$ and $w(a,b)\ge 1$ hold for every $a,b\in \Esigma$ with $a\ne b$.\footnote{Further weight functions $w'$ can be handled by setting $w'(a,b)=\alpha(a)+\beta(b)+\gamma\cdot  w(a,b)$ for appropriate parameters $\alpha,\beta:\Esigma\to \mathbb{R}$ and $\gamma\in \mathbb{R}_{>0}$. This normalization step does not change the set of optimal alignments.} As a result, $\wed(X,X)=0$ and $\ed(X,Y)\le \wed(X,Y)$ hold for all strings $X,Y\in \Sigma^*$.

As shown in \cite{Sel74,WF74}, the textbook $\Oh(n^2)$-time dynamic-programming algorithm natively supports arbitrary weights. 
A simple optimization, originating from \cite{Ukk85}, allows for an improved running time of $\Oh(nk)$ if $k\coloneqq \wed(X,Y)$ does not exceed $n$. 
For almost four decades, this remained the best running time for bounded weighted edit distance.
Only in 2023, Das, Gilbert, Hajiaghayi, Kociumaka, and Saha~\cite{DGHKS23} managed to improve upon the $\Oh(nk)$ time, albeit only for $k\le \sqrt[4]{n}$: their algorithm runs in $\Oh(n+k^5)$ time.
Soon afterward, Cassis, Kociumaka, and Wellnitz~\cite{CKW23} developed an $\Ohtilde(n+\sqrt{nk^3})$-time\footnote{The $\Ohtilde(\cdot)$ notation hides factors poly-logarithmic in the input size $n$, that is, $\log^{c} n$ for any constant $c$.} algorithm; the runtime of this solution exceeds neither $\Ohtilde(n+k^3)$ nor $\Ohtilde(nk)$, and it interpolates smoothly between $\Ohtilde(n)$ for $k= \sqrt[3]{n}$ and $\Ohtilde(nk)$ for $k = n$. 
Unexpectedly, the running time of $\Ohtilde(n+\sqrt{nk^3})$ is optimal for $\sqrt{n}\le k \le n$: any polynomial-factor improvement would violate the All-Pairs Shortest Paths Hypothesis~\cite{CKW23}.
This result provides a strict separation between the weighted and the unweighted variants of the bounded edit distance problem.
Along with settling the complexity for $\sqrt[3]{n}\le k \le \sqrt{n}$, where the conditional lower bound degrades to $n+k^{2.5-o(1)}$, the most important open question posed in~\cite{CKW23} is the following one:
\begin{quote}
    \centering{\textit{Can the $\sqrt{nk^{3-o(1)}}$ lower bound be circumvented for some natural weight function classes?}}
\end{quote}

Arguably, the weight functions devised to establish the lower bound in~\cite{CKW23} have a rather atypical structure: in particular, they use fractional costs with large denominators (reaching $\Theta(n^c)$ for a constant $c\gg 1$).
In contrast, weight functions arising in practice, such as those originating from the \texttt{BLOSUM} \cite{HH92} and \texttt{PAM}~\cite{DSO78} substitution matrices commonly used for amino-acids, are equivalent (after normalization) to weight functions with small integer values (below 30).
Unfortunately, the literature does not provide specialized solutions for this case, and small integer values are already sufficient to violate the monotonicity property that the Landau--Vishkin algorithm~\cite{LV88} hinges on.%
\footnote{As already observed in~\cite{DGHKS23}, we have $2=\wed(\mathtt{ab},\mathtt{c})<\wed(\mathtt{a},\eps)=3$ if $w(\mathtt{b},\eps)=w(\mathtt{a},\mathtt{c})=1$ and $w(\mathtt{a},\eps)=3$.}
The first main contribution of our work is that, up to a modest polylogarithmic-factor overhead, the time complexity of the Landau--Vishkin algorithm (optimal under the Orthogonal Vectors Hypothesis) can be recovered for small integer weights. 
In the most basic version, our result reads as follows:
\begin{theorem}\label{thm:static_fixed}
Fix an alphabet $\Sigma$ and a weight function $w : \Esigma^2 \to \ZZ_{\ge 0}$. 
Given strings $X,Y\in \Sigma^{\le n}$, the weighted edit distance $k \coloneqq \wed(X,Y)$ can be computed in $\Oh(n+k^2 \log^2 n)$ time.\lipicsEnd
\end{theorem}
\begin{remark}
We note that $\ZZ_{\ge 0}$ in \cref{thm:static_fixed} can be replaced with $\mathbb{Q}_{\ge 0}$:
In order to reduce to the integer case, it suffices to scale all the weights up by lowest common multiple of all the denominators.
Since we assume the weight function to be fixed (rather than a part of the input), this scaling factor is constant, so the time complexity is not affected even though the value $k$ grows.\lipicsEnd
\end{remark}

As discussed in the technical overview (\cref{sec:overview}), the algorithm behind \cref{thm:static_fixed} builds upon the approach of \cite{CKW23}, with the better running time enabled by the extra structure arising from the small integer weights.
Nevertheless, we also need several novel ideas of independent interest (applicable to arbitrary weights) to eliminate the bottlenecks that pop up as we improve the time complexity to $\Ohtilde(n+k^2)$.
The resulting algorithm is very robust: it seamlessly supports integer weight functions given as a part of the input (with a modest $\Oh(W)$-factor overhead, where $W$ is the maximum cost of a single edit), and it outputs not only the weighted edit distance, but also the underlying optimal alignment (an optimal sequence of edits transforming one string into the other).
\begin{theorem}\label{thm:static_small}
    Given strings $X, Y \in \Sigma^{\le n}$ and oracle access to a weight function $w : \Esigma^2 \to \fragmentcc{0}{W}$, 
    the weighted edit distance $k\coloneqq \wed(X, Y)$ can be computed in $\Oh(n+W\cdot k^2\log^2 n)$ time.
    The algorithm also outputs a $w$-optimal sequence of edits transforming $X$ into $Y$.\lipicsEnd
\end{theorem}
Even though the time complexity of \cref{thm:static_small} is optimal (up to sub-polynomial factors and conditioned on OVH) for $W=n^{o(1)}$, the algorithm of \cite{CKW23} for arbitrary weights becomes faster when $W$ is large, e.g., $W\ge k$. 
In order to address this issue, we present a variant of our solution with an extra subroutine that lets us effectively cap the dependency on $W$ at $\Oh(\sqrt{k}\log n)$: the resulting runtime of $\tOh(n+k^{2.5})$ improves upon the upper bound of $\Ohtilde(n+\sqrt{nk^3})$ from \cite{CKW23} if $k\le \sqrt{n}$.
\begin{theorem}\label{thm:static_large}
    Given strings $X, Y \in \Sigma^{\le n}$ and oracle access to a weight function $w : \Esigma^2 \to \ZZ_{\ge 0}$, 
    the weighted edit distance $k\coloneqq \wed(X, Y)$ can be computed in $\Oh(n+k^{2.5}\log^3 n)$ time.
    The algorithm also outputs a $w$-optimal sequence of edits transforming $X$ into $Y$.\lipicsEnd
\end{theorem}
\begin{remark}
    As a corollary of independent interest, for an arbitrary normalized weight function $w':\Esigma^2 \to \mathbb{R}_{\ge 0}$ and a real parameter $\epsilon \in (0,1]$, a $(1+\epsilon)$-factor approximation of $k\coloneqq \ed^{w'}(X,Y)$ can be computed in $\Oh(n+(k/\epsilon)^{2.5}\log^3 n)$ time. 
For this, it suffices to use \cref{thm:static_large} with $w(a,b)\coloneqq \lfloor{w'(a,b)/\epsilon}\rfloor$.\lipicsEnd
\end{remark}

\subparagraph*{Related Results}
The only weight function class systematically studied before this work consists of \emph{uniform} weight functions that assign an integer weight $W_{\textsf{indel}}$ to all insertions and deletions and another integer weight $W_{\textsf{sub}}$ to all substitutions. 
Tiskin~\cite{Tis08} observed that the $\Oh(n+k^2)$ running time can be generalized from the unweighted case of $W_{\textsf{indel}}=W_{\textsf{sub}}=1$ to $W_{\textsf{indel}},W_{\textsf{sub}}=\Oh(1)$. 
In a more fine-grained study, Goldenberg, Kociumaka, Krauthgamer, and Saha~\cite{GKKS23} adapted the Landau--Vishkin algorithm~\cite{LV88} so that it runs in $\Oh(n+k^2/W_{\textsf{indel}})$ time \mbox{for any integers $W_{\textsf{indel}},W_{\textsf{sub}}\ge 1$.}

\paragraph*{Dynamic Edit Distance}
Dynamic algorithms capture a natural scenario when the input data changes frequently, and the solution needs to be maintained upon every update.
Although the literature on dynamic algorithms covers primarily graph problems (see \cite{HHS22} for a recent survey), there is a large body of work on dynamic strings.
The problems studied in this setting include testing equality between substrings~\cite{MSU97}, longest common prefix queries \cite{MSU97,ABR00,GKKLS18}, text indexing~\cite{ABR00,GKKLS15,NIIBT20,KK22}, approximate pattern matching \cite{CKW20,CGKMU22,CKW22}, suffix array maintenance \cite{AB20,KK22}, longest common substring \cite{AB18,CGP20}, longest increasing subsequence~\cite{MS20, KS21, GJ21}, Lempel--Ziv factorization \cite{NIIBT20}, detection of repetitions \cite{ABCK19}, and last but not least edit distance~\cite{Tis08,CKM20,KMS23} (see \cite{Koc22} for a~survey~talk).

A popular model of dynamic strings, and arguably the most natural one in the context of edit distance, is where each update is a single edit, i.e., a character insertion, deletion, or substitution in one of the input strings.
A folklore dynamic edit-distance algorithm can be obtained from the Landau--Vishkin framework~\cite{LV88} using a dynamic strings implementation that can efficiently test equality between substrings.
With the data structure of Mehlhorn, Sundar, and Uhrig~\cite{MSU97}, one can already achieve the worst-case update time of $\Ohtilde(k^2)$.
Modern optimized alternatives~\cite{GKKLS18,KK22} yield $\Oh(k^2\log n)$ update time with high probability and $\Oh(k^2\log^{1+o(1)} n)$ update time deterministically.
Unfortunately, these results are not meaningful for $k\ge \sqrt{n}$: then, it is better to recompute the edit distance from scratch, in $\Oh(n+k^2)=\Oh(k^2)$ time, upon every update.
In particular, it remained open if sub-quadratic update time can be achieved for the unbounded version of dynamic edit distance.

Early works contributed towards answering this question by studying a restricted setting allowing updates only at the endpoints of the maintained strings~\cite{LMS98,KP04,IIST05,Tis08}. 
The most general of these results is an algorithm by Tiskin~\cite{Tis08} that works in $\Oh(n)$ time per update subject to edits at both endpoints of both strings.
More recently, Charalampopoulos, Kociumaka, and Mozes \cite{CKM20} applied Tiskin's toolbox~\cite{Tis08,Tiskin10} in a dynamic edit distance algorithm that supports arbitrary updates in $\Oh(n\log^2 n)$ time.
Any significantly better update time of $\Oh(n^{1-\epsilon})$ would immediately yield an $\Oh(n^{2-\epsilon})$-time static algorithm and thus violate the Orthogonal Vectors Hypothesis.
The fine-grained lower bound, however, does not prohibit improvements for $k\ll n$, and the state-of-the-art update time of $\Ohtilde(\min\{n,k^2\})$ motivates the following tantalizing open question, explicitly posed in~\cite{Koc22}:
\begin{quote}
    \centering{\textit{Is there a dynamic edit distance algorithm that supports updates in $\Ohtilde(k)$ time?}}
\end{quote}
The second main result of our work is an affirmative answer to this question.
\begin{theorem}\label{thm:unweighed_dynamic}
    There exists a deterministic dynamic algorithm that maintains strings $X,Y\in \Sigma^{*}$ subject to edits and, upon every update, computes $k\coloneqq \ed(X,Y)$ in $\Oh(k\log^2 n)$ time, where $n\coloneqq |X|+|Y|$.
    The algorithm can be initialized in $\Oh(n\log^{o(1)}n+k^2\log^2 n)$ time and, along with $\ed(X,Y)$, it also outputs an optimal sequence of edits transforming $X$ into $Y$.\lipicsEnd
\end{theorem}
We note that the fine-grained lower bounds prohibit improving the update time to $\Ohtilde(k^{1-\epsilon})$ for any $\epsilon > 0$, even for instances restricted to $k=\Theta(n^{\kappa})$ for any constant $0 < \kappa\le 1$.%
\footnote{To see this, consider a static edit distance instance $(\bX,\bY)$ with $\ed(\bX,\bY)=\Theta(n^{\kappa})$ and $|\bX|+|\bY|=\Theta(n^{\kappa})$. Initialize a dynamic edit distance algorithm with $X=Y=\mathtt{0}^n$ and perform $\Theta(n^{\kappa})$ insertions so that $X=\mathtt{0}^n\bX$ and $Y=\mathtt{0}^n\bY$. With an update time of $\Ohtilde(k^{1-\epsilon})$, we could compute $\ed(\bX,\bY)=\ed(X,Y)$ in $\Ohtilde(n^{\kappa}\cdot n^{\kappa(1-\epsilon)})=\Ohtilde(n^{\kappa(2-\epsilon)})$ time, violating the Orthogonal Vectors Hypothesis. 
Considering multiple instances $(\bX,\bY)$, one can further prove that the $\Ohtilde(k^{1-\epsilon})$ update time cannot be achieved even after $\Oh(n^c)$-time initialization for an arbitrarily large constant $c$.}

Notably, even though our solution is significantly more complicated than the state-of-the-art dynamic algorithm for unbounded edit distance~\cite{CKM20}, we do not incur any additional polylogarithmic factors in the update time.
What is perhaps more surprising, our algorithm natively supports small integer weights and, unlike in the static case, we are not aware of any significantly simpler approach tailored for the unweighted version of the dynamic bounded edit distance problem.
\begin{theorem}\label{thm:dynamc_informal}
    Let $w : \Esigma^2 \to \fragmentcc{0}{W}$ be a weight function supporting constant-time oracle access.
    There exists a deterministic dynamic algorithm that maintains strings $X,Y\in \Sigma^{*}$ subject to edits and, upon every update,
    computes $k\coloneqq \wed(X,Y)$ in $\Oh(W^2 k\log^2 n)$ time, where $n\coloneqq |X|+|Y|$.
    The algorithm can be initialized in $\Oh(n\log^{o(1)}n+Wk^2\log^2 n)$ time and, along with $\wed(X,Y)$, it also outputs a $w$-optimal sequence of edits transforming $X$ into $Y$.\lipicsEnd
\end{theorem}

\paragraph*{Open Questions}
In the unweighted setting, the $\Oh(k\log^2 n)$ update time of \cref{thm:unweighed_dynamic} does not give much room for improvement: the most pressing challenge is to reduce the $\Oh(\log^2 n)$-factor overhead in \cite{CKM20}.
For small integer weights, the dependency on the largest weight $W$ remains to be studied.
The linear dependency in the running time $\Ohtilde(n+Wk^2)$ of \cref{thm:static_small} feels justified, but we hope that the quadratic dependency in the update time $\Ohtilde(W^2 k)$ of \cref{thm:dynamc_informal} can be reduced.
An exciting open question is to determine the optimality of our $\Ohtilde(n+k^{2.5})$-time static algorithm for large integer weights (for $k \le \min\{W^2,\sqrt{n}\}$).
This runtime matches the lower bound of~\cite{CKW23}, but the hard instances constructed there crucially utilize fractional weights.
The complexity of dynamic bounded edit distance also remains widely open for large weights:
we are only aware of how to achieve $\Ohtilde(k^{2.5})$ and $\Ohtilde(k^{3})$ update time for integer and general weights, respectively, and these solutions simply combine static algorithms with a dynamic strings implementation supporting efficient substring equality tests. 
With extra effort, the techniques developed in this paper should lead to a faster support of updates that do not change $\wed(X,Y)$ too much. 
We do not know, however, how to improve upon the aforementioned simple approach for updates that drastically alter $\wed(X,Y)$.

\section{Technical Overview}\label{sec:overview}

In this section, we provide an overview of our algorithms. 
We focus on the static algorithm behind \cref{thm:static_fixed}, which already incorporates most novel ideas and techniques.
In the overview, we assume for simplicity that the parameter $k$ is an upper bound on $\wed(X,Y)$ known to the algorithms.

\subparagraph*{Basic Concepts}
Consistently with most previous work, we interpret the edit distance problem using the \emph{alignment graph} of the input strings.
For strings $X,Y\in \Sigma^*$ and a weight function $w:\Esigma^2\to \mathbb{R}_{\ge 0}$, the alignment graph $\AGw(X,Y)$ is a directed grid graph with vertex set $\fragmentcc{0}{|X|}\times \fragmentcc{0}{|Y|}$ and the following edges (see \cref{fig:alignment-graph}):
\begin{itemize}
    \item $(x,y)\to (x,y+1)$ of weight $w(\eps,Y[y])$, representing an insertion of $Y[y]$;
    \item $(x,y)\to (x+1,y)$ of weight $w(X[x],\eps)$, representing a deletion of $X[x]$;
    \item $(x,y)\to (x+1,y+1)$ of weight $w(X[x],Y[y])$, representing a match (if $X[x]=Y[y]$) or a substitution of $X[x]$ for $Y[y]$.
\end{itemize}

\begin{figure}[htb]
    \vspace{-.35cm}
    \begin{center}\begin{tikzpicture}[scale=.89]\draw[line width=5pt, black!35, line cap=round] (0, 6) -- (4, 2) -- (8, 2) -- (8, 0);
\node[inner sep = 2pt,circle,fill=black] (n0_0) at (0,6) {};
\node[inner sep = 2pt,circle,fill=black] (n0_1) at (0,4) {};
\draw[-latex,darkred,thick] (n0_0) -- (n0_1) node[midway, above=-2.5, sloped, text=blue, font=\small] {$w(\varepsilon, \mathtt{b})$};
\node[inner sep = 2pt,circle,fill=black] (n0_2) at (0,2) {};
\draw[-latex,darkred,thick] (n0_1) -- (n0_2) node[midway, above=-2.5, sloped, text=blue, font=\small] {$w(\varepsilon, \mathtt{a})$};
\node[inner sep = 2pt,circle,fill=black] (n0_3) at (0,0) {};
\draw[-latex,darkred,thick] (n0_2) -- (n0_3) node[midway, above=-2.5, sloped, text=blue, font=\small] {$w(\varepsilon, \mathtt{b})$};
\node[inner sep = 2pt,circle,fill=black] (n1_0) at (2,6) {};
\draw[-latex,darkred,thick] (n0_0) -- (n1_0) node[midway, above=-2.5, text=blue, font=\small] {$w(\mathtt{b}, \varepsilon)$};
\node[inner sep = 2pt,circle,fill=black] (n1_1) at (2,4) {};
\draw[-latex,darkred,thick] (n0_1) -- (n1_1) node[midway, above=-2.5, text=blue, font=\small] {$w(\mathtt{b}, \varepsilon)$};
\draw[-latex,darkred,thick] (n1_0) -- (n1_1) node[midway, above=-2.5, sloped, text=blue, font=\small] {$w(\varepsilon, \mathtt{b})$};
\draw[-latex,darkgreen,thick] (n0_0) -- (n1_1)node[midway, above=-2.5, sloped, text=blue, font=\small] {$w(\mathtt{b}, \mathtt{b})$};
\node[inner sep = 2pt,circle,fill=black] (n1_2) at (2,2) {};
\draw[-latex,darkred,thick] (n0_2) -- (n1_2) node[midway, above=-2.5, text=blue, font=\small] {$w(\mathtt{b}, \varepsilon)$};
\draw[-latex,darkred,thick] (n1_1) -- (n1_2) node[midway, above=-2.5, sloped, text=blue, font=\small] {$w(\varepsilon, \mathtt{a})$};
\draw[-latex,darkred,thick] (n0_1) -- (n1_2)node[midway, above=-2.5, sloped, text=blue, font=\small] {$w(\mathtt{b}, \mathtt{a})$};
\node[inner sep = 2pt,circle,fill=black] (n1_3) at (2,0) {};
\draw[-latex,darkred,thick] (n0_3) -- (n1_3) node[midway, above=-2.5, text=blue, font=\small] {$w(\mathtt{b}, \varepsilon)$};
\draw[-latex,darkred,thick] (n1_2) -- (n1_3) node[midway, above=-2.5, sloped, text=blue, font=\small] {$w(\varepsilon, \mathtt{b})$};
\draw[-latex,darkgreen,thick] (n0_2) -- (n1_3)node[midway, above=-2.5, sloped, text=blue, font=\small] {$w(\mathtt{b}, \mathtt{b})$};
\node[inner sep = 2pt,circle,fill=black] (n2_0) at (4,6) {};
\draw[-latex,darkred,thick] (n1_0) -- (n2_0) node[midway, above=-2.5, text=blue, font=\small] {$w(\mathtt{a}, \varepsilon)$};
\node[inner sep = 2pt,circle,fill=black] (n2_1) at (4,4) {};
\draw[-latex,darkred,thick] (n1_1) -- (n2_1) node[midway, above=-2.5, text=blue, font=\small] {$w(\mathtt{a}, \varepsilon)$};
\draw[-latex,darkred,thick] (n2_0) -- (n2_1) node[midway, above=-2.5, sloped, text=blue, font=\small] {$w(\varepsilon, \mathtt{b})$};
\draw[-latex,darkred,thick] (n1_0) -- (n2_1)node[midway, above=-2.5, sloped, text=blue, font=\small] {$w(\mathtt{a}, \mathtt{b})$};
\node[inner sep = 2pt,circle,fill=black] (n2_2) at (4,2) {};
\draw[-latex,darkred,thick] (n1_2) -- (n2_2) node[midway, above=-2.5, text=blue, font=\small] {$w(\mathtt{a}, \varepsilon)$};
\draw[-latex,darkred,thick] (n2_1) -- (n2_2) node[midway, above=-2.5, sloped, text=blue, font=\small] {$w(\varepsilon, \mathtt{a})$};
\draw[-latex,darkgreen,thick] (n1_1) -- (n2_2)node[midway, above=-2.5, sloped, text=blue, font=\small] {$w(\mathtt{a}, \mathtt{a})$};
\node[inner sep = 2pt,circle,fill=black] (n2_3) at (4,0) {};
\draw[-latex,darkred,thick] (n1_3) -- (n2_3) node[midway, above=-2.5, text=blue, font=\small] {$w(\mathtt{a}, \varepsilon)$};
\draw[-latex,darkred,thick] (n2_2) -- (n2_3) node[midway, above=-2.5, sloped, text=blue, font=\small] {$w(\varepsilon, \mathtt{b})$};
\draw[-latex,darkred,thick] (n1_2) -- (n2_3)node[midway, above=-2.5, sloped, text=blue, font=\small] {$w(\mathtt{a}, \mathtt{b})$};
\node[inner sep = 2pt,circle,fill=black] (n3_0) at (6,6) {};
\draw[-latex,darkred,thick] (n2_0) -- (n3_0) node[midway, above=-2.5, text=blue, font=\small] {$w(\mathtt{a}, \varepsilon)$};
\node[inner sep = 2pt,circle,fill=black] (n3_1) at (6,4) {};
\draw[-latex,darkred,thick] (n2_1) -- (n3_1) node[midway, above=-2.5, text=blue, font=\small] {$w(\mathtt{a}, \varepsilon)$};
\draw[-latex,darkred,thick] (n3_0) -- (n3_1) node[midway, above=-2.5, sloped, text=blue, font=\small] {$w(\varepsilon, \mathtt{b})$};
\draw[-latex,darkred,thick] (n2_0) -- (n3_1)node[midway, above=-2.5, sloped, text=blue, font=\small] {$w(\mathtt{a}, \mathtt{b})$};
\node[inner sep = 2pt,circle,fill=black] (n3_2) at (6,2) {};
\draw[-latex,darkred,thick] (n2_2) -- (n3_2) node[midway, above=-2.5, text=blue, font=\small] {$w(\mathtt{a}, \varepsilon)$};
\draw[-latex,darkred,thick] (n3_1) -- (n3_2) node[midway, above=-2.5, sloped, text=blue, font=\small] {$w(\varepsilon, \mathtt{a})$};
\draw[-latex,darkgreen,thick] (n2_1) -- (n3_2)node[midway, above=-2.5, sloped, text=blue, font=\small] {$w(\mathtt{a}, \mathtt{a})$};
\node[inner sep = 2pt,circle,fill=black] (n3_3) at (6,0) {};
\draw[-latex,darkred,thick] (n2_3) -- (n3_3) node[midway, above=-2.5, text=blue, font=\small] {$w(\mathtt{a}, \varepsilon)$};
\draw[-latex,darkred,thick] (n3_2) -- (n3_3) node[midway, above=-2.5, sloped, text=blue, font=\small] {$w(\varepsilon, \mathtt{b})$};
\draw[-latex,darkred,thick] (n2_2) -- (n3_3)node[midway, above=-2.5, sloped, text=blue, font=\small] {$w(\mathtt{a}, \mathtt{b})$};
\node[inner sep = 2pt,circle,fill=black] (n4_0) at (8,6) {};
\draw[-latex,darkred,thick] (n3_0) -- (n4_0) node[midway, above=-2.5, text=blue, font=\small] {$w(\mathtt{a}, \varepsilon)$};
\node[inner sep = 2pt,circle,fill=black] (n4_1) at (8,4) {};
\draw[-latex,darkred,thick] (n3_1) -- (n4_1) node[midway, above=-2.5, text=blue, font=\small] {$w(\mathtt{a}, \varepsilon)$};
\draw[-latex,darkred,thick] (n4_0) -- (n4_1) node[midway, above=-2.5, sloped, text=blue, font=\small] {$w(\varepsilon, \mathtt{b})$};
\draw[-latex,darkred,thick] (n3_0) -- (n4_1)node[midway, above=-2.5, sloped, text=blue, font=\small] {$w(\mathtt{a}, \mathtt{b})$};
\node[inner sep = 2pt,circle,fill=black] (n4_2) at (8,2) {};
\draw[-latex,darkred,thick] (n3_2) -- (n4_2) node[midway, above=-2.5, text=blue, font=\small] {$w(\mathtt{a}, \varepsilon)$};
\draw[-latex,darkred,thick] (n4_1) -- (n4_2) node[midway, above=-2.5, sloped, text=blue, font=\small] {$w(\varepsilon, \mathtt{a})$};
\draw[-latex,darkgreen,thick] (n3_1) -- (n4_2)node[midway, above=-2.5, sloped, text=blue, font=\small] {$w(\mathtt{a}, \mathtt{a})$};
\node[inner sep = 2pt,circle,fill=black] (n4_3) at (8,0) {};
\draw[-latex,darkred,thick] (n3_3) -- (n4_3) node[midway, above=-2.5, text=blue, font=\small] {$w(\mathtt{a}, \varepsilon)$};
\draw[-latex,darkred,thick] (n4_2) -- (n4_3) node[midway, above=-2.5, sloped, text=blue, font=\small] {$w(\varepsilon, \mathtt{b})$};
\draw[-latex,darkred,thick] (n3_2) -- (n4_3)node[midway, above=-2.5, sloped, text=blue, font=\small] {$w(\mathtt{a}, \mathtt{b})$};
\draw (1, -0.35) node[below]{$\mathtt{b}$};
\draw (3, -0.35) node[below]{$\mathtt{a}$};
\draw (5, -0.35) node[below]{$\mathtt{a}$};
\draw (7, -0.35) node[below]{$\mathtt{a}$};
\draw (-0.35,5) node[left]{$\mathtt{b}$};
\draw (-0.35,3) node[left]{$\mathtt{a}$};
\draw (-0.35,1) node[left]{$\mathtt{b}$};
\draw (0, 0) node[left,font=\small,text=orange] {$\textup{in}_{1}$};
\draw (0, 2) node[left,font=\small,text=orange] {$\textup{in}_{2}$};
\draw (0, 4) node[left,font=\small,text=orange] {$\textup{in}_{3}$};
\draw (0, 6) node[above left,font=\small,text=orange] {$\textup{in}_{4}$};
\draw (2, 6) node[above,font=\small,text=orange] {$\textup{in}_{5}$};
\draw (4, 6) node[above,font=\small,text=orange] {$\textup{in}_{6}$};
\draw (6, 6) node[above,font=\small,text=orange] {$\textup{in}_{7}$};
\draw (8, 6) node[above,font=\small,text=orange] {$\textup{in}_{8}$};
\draw (0, 0) node[below,font=\small,text=violet] {$\textup{out}_{1}$};
\draw (2, 0) node[below,font=\small,text=violet] {$\textup{out}_{2}$};
\draw (4, 0) node[below,font=\small,text=violet] {$\textup{out}_{3}$};
\draw (6, 0) node[below,font=\small,text=violet] {$\textup{out}_{4}$};
\draw (8, 0) node[below right,font=\small,text=violet] {$\textup{out}_{5}$};
\draw (8, 2) node[right,font=\small,text=violet] {$\textup{out}_{6}$};
\draw (8, 4) node[right,font=\small,text=violet] {$\textup{out}_{7}$};
\draw (8, 6) node[right,font=\small,text=violet] {$\textup{out}_{8}$};
    \end{tikzpicture}\end{center}\vspace{-.45cm}

    \caption{The alignment graph $\AGw(X, Y)$ for $X = \mathtt{baaa}$ and $Y = \mathtt{bab}$. In the unweighted case, the \textcolor{darkgreen}{green} edges have weight $0$, and the \textcolor{red}{red} ones have weight $1$. In the weighted case, the \textcolor{red}{red} edges have some weights defined by the weight function $w$. An optimal alignment for a weight function $w$ satisfying $w(\mathtt{a}, \varepsilon) = w(\varepsilon, \mathtt{b}) = 1$ and $w(\mathtt{a}, \mathtt{b}) = w(\mathtt{b}, \mathtt{a}) = w(\varepsilon, \mathtt{a}) = w(\mathtt{b}, \varepsilon) = 3$ is given in \textcolor{black!50}{gray}. The input vertices for the boundary matrix $\BMw(X, Y)$ have \textcolor{orange}{orange} labels, and the output vertices have \textcolor{violet}{violet} labels.}
    \label{fig:alignment-graph}
\end{figure}

Every alignment $\cA$ mapping $X$ onto $Y$, denoted $\cA : X \onto Y$, can be interpreted as a path in $\AGw(X,Y)$ from the top-left corner $(0,0)$ to the bottom-right corner $(|X|,|Y|)$, and the cost of the alignment, denoted $\wed_{\cA}(X,Y)$ is the length of the underlying path.
Consequently, the edit distance $\wed(X,Y)$ is simply the distance from $(0,0)$ to $(|X|,|Y|)$ in the alignment graph.

Many edit-distance algorithms compute not only the distance from $(0,0)$ to $(|X|,|Y|)$ but the entire \emph{boundary distance matrix} $\BMw(X,Y)$ that stores the distances from every \emph{input} vertex on the top-left boundary to every \emph{output} vertex on the bottom-right boundary of the alignment graph $\AGw(X,Y)$.
Crucially, the planarity of $\AGw(X,Y)$ implies that $M=\BMw(X,Y)$ satisfies the \emph{Monge property}, i.e., $M_{i,j}+M_{i+1,j+1} \le M_{i,j+1}+M_{i+1,j}$ holds whenever all four entries are finite.\footnote{In this overview, we ignore the fact that some entries of $M$ are infinite.
Our actual algorithms augment the alignment graph with backward edges so that the finite distances are preserved and the infinite distances become finite.}

\paragraph*{High-Level Algorithm Structure: Divide and Conquer}

The most important insight from~\cite{CKW23} is that the problem of computing $\wed(X,Y) \le k$ can be reduced to instances satisfying $\sed(X)=\Oh(k)$.
The value $\sed(X)$, called the \emph{self-edit distance} of~$X$, is the distance from $(0,0)$ to $(|X|,|X|)$ in the unweighted alignment graph $\AG(X,X)$ with the edges $(x,x)\to (x+1,x+1)$ on the main diagonal removed; see \cref{sec:sed} for a discussion of self-edit distance and its most important properties.

At a high level, the reduction follows a divide-and-conquer scheme: the strings are partitioned into two \emph{halves} $X=X_LX_R$ and $Y=Y_LY_R$, the values $\wed(X_L,Y_L)$ and $\wed(X_R,Y_R)$ are computed recursively, and then the distance $\wed(X,Y)$ is derived. 
The original approach from~\cite{CKW23}, also applied in a quantum algorithm for the unweighted edit distance~\cite{GJKT24}, partitions $X$ so that the two halves have the same length (up to $\pm 1$) and $Y$ so that $\wed(X,Y)=\wed(X_L,Y_L)+\wed(X_R,Y_R)$.
For the latter, it computes a $w$-optimal alignment between appropriate fragments of self-edit distance $\Oh(k)$ taken from the middles of $X$ and~$Y$. 
Intuitively, this is sufficient because whether a given vertex $(x,y)\in \fragmentcc{0}{|X|}\times \fragmentcc{0}{|Y|}$ belongs to a $w$-optimal alignment of cost at most $k$ depends only on \emph{contexts} of self-edit distance $\Oh(k)$ around positions $x$ in $X$ and $y$ in $Y$.
Unfortunately, the resulting scheme is unable to tell how to split the budget $k$ between the two recursive calls, and thus both calls utilize exponential search to estimate the local weighted edit distance, causing significant complications and, in case of~\cite{GJKT24}, also a large polylogarithmic-factor overhead on top of the $\Ohtilde(k^2)$ time complexity.
Moreover, this scheme is incompatible with the dynamic setting because a single update may change the partition of $Y$ and thus affect both recursive calls.

\subparagraph*{Technical Contribution 1: Simple and Robust Divide-and-Conquer Scheme}
We circumvent the aforementioned issues using a novel divide-and-conquer approach that relies on an \emph{approximately optimal} alignment $\cA : X\onto Y$ of (weighted) cost $\wed_{\cA}(X,Y)=\Oh(k)$.
In the context of \cref{thm:static_fixed}, we can pick $\cA$ to be an optimal \emph{unweighted} alignment whose cost satisfies $\wed_{\cA}(X,Y)=\Oh(\ed(X,Y))=\Oh(\wed(X,Y))=\Oh(k)$.
In the dynamic setting, it suffices to occasionally rebuild $\cA$. 

\begin{figure}[t!]
    \begin{subfigure}[t]{0.48\textwidth}
    \begin{center}
        \begin{tikzpicture}[y=-.85cm,x=0.85cm]
            \useasboundingbox (-0.5, -0.5) rectangle (8, 8);
            \scope[transform canvas={scale=.68}]
                \def\W{12}
\def\H{12}
\def\cnt{6}
\def\sizeofdot{4pt}
\def\gp{0.0}

\draw[line width=3pt, violet!30, line cap=round] (0, 0) -- (1, 1) -- (1, 1.5) -- (4.5, 5) -- (5.5, 5) -- (6.5, 6) -- (6.5, 8) -- (9, 10.5) -- (10.5, 10.5) -- (\H, \W);

\node (XL) at (\H / 5, -.5) {\Large $X_L$};
\node (XR) at (4 * \H / 5, -.5) {\Large $X_R$};

\draw[latex-latex] (\H / 3, -0.5) to node[midway, above] {$\sed = \Theta(k)$} (2 * \H / 3, -0.5);
\draw[latex-latex] (\H / 3, -0.5) -- (\H / 2, -0.5); 
\draw[latex-latex] (2*\H / 3, -0.5) -- (\H / 2, -0.5); 
\draw[dashed, gray] (\H / 3, 0) -- (\H / 3, \W);
\draw[dashed, gray] (2 * \H / 3, 0) -- (2 * \H / 3, \W);
\draw (\H/2, 0) node[above]{$x_m$};

\draw[line width=4.5pt, brown!30, line cap=round] (0, 0) -- (1.5, 0) -- (4, 2.5) -- (4, 3.5) -- (4.5,4) -- (4.5,5) -- (5, 5.5) -- (\H / 2, 5.5);
\draw[line width=4.5pt, brown!30, line cap=round] (\H / 2, 5.5) -- (6, 6.5) -- (7.5, 8) -- (8.5, 8) -- (11, 10.5) -- (11, 11) -- (\H, \W);

\draw[line width=4.5pt, blue!50,line cap=round] (\H / 3, 4.5) -- (4.5, 5) -- (6.5, 5) -- (7.5, 6) -- (7.5, 8) -- (8, 8.5) -- (2 * \H / 3, 9.5);

\node[circle,draw=darkgreen, fill=darkgreen, inner sep=0pt,minimum size=\sizeofdot] (olddot) at (4.5,5) {};
\node[circle,draw=darkgreen, fill=darkgreen, inner sep=0pt,minimum size=\sizeofdot] (olddot) at (7.5,8) {};

\draw[line width=2pt, darkgreen!70, line cap=round](0, 0) -- (1.5, 0) -- (4, 2.5) -- (4, 3.5) -- (4.5, 4) -- (4.5,5) -- (6.5, 5) -- (7.5, 6) -- (7.5, 8) -- (8.5, 8) -- (11, 10.5) -- (11, 11) -- (\H, \W);

\draw[dashed, gray] (0, 5.5) -- (\H, 5.5);

\draw (0, 5.5) node[left]{$y_m$};

\node[circle,draw=violet, fill=violet, inner sep=0pt,minimum size=\sizeofdot] (olddot) at (\H / 2, 5.5) {};

\draw[thick] (0, 0) rectangle (\H, \W);
\draw (\H / 2, 0) -- (\H / 2, \W);
\draw[blue, thick] (\H / 3, 4.5) rectangle (2 * \H / 3, 9.5);
            \endscope
        \end{tikzpicture}
    \end{center}
    \caption{The approximate alignment $\cA$ is drawn in \textcolor{violet!50}{violet}.
    The two optimal alignments $\cB_L$ and $\cB_R$ computed recursively are given in \textcolor{brown!50}{brown}. The optimal alignment $\cB_M$ in the \textcolor{blue}{blue} central area of small self-edit distance is given in \textcolor{blue!70}{blue}. The globally optimal \textcolor{darkgreen!70}{green} alignment $\cB$ is a combination of $\cB_L$, $\cB_M$, and $\cB_R$.} 
    \label{fig:new-divide-and-conquer}
    \end{subfigure}%
    \hspace{.04\textwidth}%
    \begin{subfigure}[t]{0.48\textwidth}
    \begin{center}
        \begin{tikzpicture}[y=-.85cm,x=0.85cm]
            \useasboundingbox (-.5, -.5) rectangle (8,8);
            \scope[transform canvas={scale=.68}]
                \def\shft{0}
\def\W{12}
\def\H{12}
\def\myi{2}
\def\cnt{6}
\def\gp{2}
\def\kp{1.5}
\def\charwidth{1 / 15}

\foreach \x in {0,5*\gp}
    \draw[fill=darkblue!40, thick] (\x, {max(\x - \kp, 0)}) rectangle (\x + \gp, {min(\x + \gp + \kp, \H)});

\foreach \x in {\gp,2*\gp,3*\gp} {
    \draw[fill=darkblue!40, thick] (\x, {max(\x - \kp, 0)}) rectangle (\x + \gp, {min(\x + \gp + \kp, \H)});
}

\foreach \x in {4*\gp} {
    \draw[fill=darkblue!40, thick] (\x, {max(\x - \kp, 0)}) rectangle (\x + \gp, {min(\x + \gp + \kp, \H)});
}

\draw (\kp,0.1) -- (\kp,-0.1) node[above] {$k$};
\draw (0.1,\kp) -- (-0.1,\kp) node[left] {$k$};
\draw[fill=darkgreen!40] (0, 0) -- (\kp, 0) -- (\H, \W - \kp) -- (\H, \W) -- (\H - \kp, \W) -- (0, \kp) -- (0, 0);

\foreach \x in {0,\gp,...,{\the\numexpr \W - \gp}}
    \draw (\x, {max(\x - \kp, 0)}) rectangle (\x + \gp, {min(\x + \kp + \gp, \H)});

\draw[thick] (0, 0) rectangle (\H, \W);

\draw (0,0.1) -- (0,0) node[above] {$x_0$};
\draw (\W,0.1) -- (\W,0) node[above] {$x_{\cnt}$};
\foreach \x in {1,2,...,{\the\numexpr \cnt - 1}}
    \draw (\x * \gp,0.1) -- (\x * \gp, -0.1) node[above] {$x_{\x}$};

\draw[dashed] (\myi * \gp + \gp, 0) -- (\myi * \gp + \gp, \myi * \gp);
\draw[dashed] (0,\myi * \gp - \kp) node[left] {$x_{\myi} - k$} -- (\myi * \gp,\myi * \gp - \kp);
\draw[dashed] (0,\myi * \gp + \gp + \kp) node[left] {$x_{\the\numexpr \myi + 1} + k$} -- (\myi * \gp,\myi * \gp + \kp+ \gp);

\node[left, darkred] at (\myi * \gp, \myi * \gp) {$V_{\myi}$};
\node[right, darkred] at (\myi * \gp + \gp, \myi * \gp + \gp) {$V_{\the\numexpr \myi + 1}$};

\draw[very thick, darkred] (\myi * \gp, \myi * \gp - \kp) -- (\myi * \gp, \myi * \gp + \kp);
\draw[very thick, darkred] (\myi * \gp + \gp, \myi * \gp+\gp-\kp) -- (\myi * \gp + \gp, \myi * \gp + \gp + \kp);

\node[darkred] at (\myi * \gp + \gp / 2, \myi * \gp + \gp / 2) {\Large $D_{\myi, {\the\numexpr \myi + 1}}$};
            \endscope
        \end{tikzpicture}
    \end{center}

    \caption{The decomposition of the alignment graph of $X$ and $Y$ into \textcolor{darkblue!60}{subgraphs} $G_i$ of size $\Theta(k) \times \Theta(k)$ that cover the whole \textcolor{darkgreen!60}{stripe} of width $O(k)$ around the main diagonal. Matrix $D_{i, i + 1}$ represents the distances between the vertices of $V_i$ and $V_{i+1}$.}
    \label{fig:substitutions-only}
    \end{subfigure}
    \caption{Our algorithms' setup for the general case and the case of small self-edit distance, respectively.}
\end{figure}

The procedure described next and illustrated in \cref{fig:new-divide-and-conquer} constructs a $w$-optimal alignment $\cB : X\onto Y$ using an arbitrary alignment $\cA : X \onto Y$ of cost $\Oh(k)$.
We partition $X$ into two halves of the same length (up to $\pm 1$) and $Y$ so that $(x_m,y_m)\coloneqq (|X_L|,|Y_L|)\in \cA$, that is, $\wed_{\boldsymbol\cA}(X,Y)=\wed_{\boldsymbol\cA}(X_L,Y_L)+\wed_{\boldsymbol\cA}(X_R,Y_R)$.
Then, we recursively build $w$-optimal alignments $\cB_L : X_L\onto Y_L$ and $\cB_R:X_R \onto Y_R$. 
Now, the recursive calls have predictable structure (guided by $\cA$) with local budgets given by the local costs of $\cA$.
Unfortunately, we are no longer guaranteed that $\wed(X,Y)=\wed(X_L,Y_L)+\wed(X_R,Y_R)$, so $\cB$ cannot be obtained by just concatenating $\cB_L$ and $\cB_R$.
Nevertheless, it turns out that the resulting alignment needs to be fixed only in a small (in terms of self-edit distance) neighborhood of $(x_m, y_m)$.
To be precise, we find a context $X_M$ of position $x_m$ in $X$ and compute the $w$-optimal alignment $\cB_M : X_M\onto Y_M$, where $Y_M$ is the image of $X_M$ under~$\cA$.
If the context has a sufficiently large self-edit distance (on both sides of $x_m$), then, as shown in \cref{cor:intersect}, the alignment $\cB_M$ intersects both $\cB_L$ and $\cB_R$, and the sought $w$-optimal alignment $\cB : X\onto Y$ can be obtained from the concatenation of $\cB_L$ and $\cB_R$ by following $\cB_M$ between the two intersection points. 
The recursion terminates whenever $X$ is very short (when a naive $\Oh(nk)$-time algorithm is used) or $\cA$ perfectly matches $X$ with $Y$ (when $\cA$ is already $w$-optimal).

The scheme above reduces computing $\wed(X,Y)$ to the case when $\sed(X)=\Oh(k)$ with an $\Oh(\log n)$-factor overhead due to the depth of recursion.
In \cref{sec:general}, we provide an optimized implementation that avoids this overhead by partitioning $X$ so that the cost of $\cA$ is split equally between the two halves. 
We also discuss computing $\cA$ in the contexts of \cref{thm:static_small,thm:static_large}.
A similar instantiation of this scheme can also be applied to significantly simplify the results of~\cite{CKW23}.

\paragraph*{The Case of Small Self-Edit Distance}
The remaining task is to compute $\wed(X,Y)$ under the assumption that $\wed(X,Y)\le k$ and $\sed(X)= \Oh(k)$ hold for a known parameter $k$.
On a high level, we follow an $\Ohtilde(n+k^3)$-time algorithm from~\cite{CKW23}, whose setup we recall next and illustrate in \cref{fig:substitutions-only}.
The optimal self-alignment $X\onto X$ yields a decomposition $X=\bigodot_{i=0}^{m-1} X_i$ of the string $X$ into phrases $X_i=X\fragmentco{x_{i}}{x_{i+1}}$ of length $\Theta(k)$ so that all but $\Oh(k)$ phrases $X_i$ satisfy $X_i = X_{i-1}$.
For each phrase, we define a fragment $Y_i=Y\fragmentco{y_i}{y'_{i+1}}$, where $y_i = \max\{x_i-k,0\}$ and $y'_{i+1}=\min\{x_{i+1}+k,|Y|\}$, and a subgraph $G_i$ of $\AGw(X,Y)$ induced by $\fragmentcc{x_i}{x_{i+1}}\times \fragmentcc{y_i}{y'_{i+1}}$.
The union $G$ of subgraphs $G_i$ contains every vertex $(x,y)$ with $|x-y|\le k$, so the promise $\wed(X,Y)\le k$ guarantees $\wed(X,Y)=\dist_G((0,0),(|X|,|Y|))$: every alignment pays at least $1$ to move between diagonals, so the optimal alignment cannot deviate by more than $k$ from the main diagonal. 
Furthermore, there are $\Oh(k)$ indices $i$ such that $(X_i, Y_i) \neq (X_{i-1}, Y_{i-1})$ and, 
as summarized in \cref{lem:decomp2}, the results of \cite{CKW23} let us efficiently construct the set $F$ of such indices $i$ along with the underlying fragments $X_i$ and $Y_i$.

For each $i\in \fragmentcc{0}{m}$, we define $V_i=\{x_i\}\times \fragmentcc{y_i}{y'_i}$, where we set $y_m=|Y|$ and $y'_0=0$ so that $V_0=\{(0,0)\}$ and $V_m=\{(|X|,|Y|)\}$.
For $i\in \fragmentoo{0}{m}$, the set $V_i$ consists of the vertices shared by $G_{i-1}$ and $G_i$.
Thus, if we define $D_{i,j}$ as the matrix of distances (in $G$) between vertices in $V_i$ and $V_j$, then the sought value $\wed(X,Y)$ is the only entry of $D_{0,m}=\bigotimes_{i=0}^{m-1}D_{i,i+1}$, where $\bigotimes$ denotes the $(\min,+)$ product of matrices.
Since $D_{i,i+1}=D_{i-1,i}$ holds for all $i\notin F$, it suffices to construct $D_{i,i+1}$ for all $i\in F$, raise each matrix to an appropriate power, and then multiply all the $\Oh(k)$ matrix powers.
Overall, this requires constructing $\Oh(k)$ individual matrices $D_{i,i+1}$ and performing $\Oh(k\log n)$ products of the form $D_{p,r}=D_{p,q}\otimes D_{q,r}$.
The algorithm of \cite{CKW23} performs each of these operations individually using Klein's planar multiple-source shortest path algorithm~\cite{Kle05} to construct $D_{i,i+1}$ in $\Oh(k^2\log k)$ time and the SMAWK method~\cite{SMAWK87} for $\Oh(k^2)$-time $(\min,+)$ Monge matrix multiplication.
Both steps contribute to an $\Ohtilde(k^3)$-time bottleneck.

\subparagraph*{Technical Contribution 2: Application of Core-Sparse Monge Matrix Multiplication}
The key advantage of small integer weights is the \emph{bounded difference} property~\cite{BGSW19} satisfied by the underlying distance matrices: if the edit costs are in $\fragmentcc{0}{W}$, then the differences between adjacent entries in any distance matrix belong to the integer range $\fragmentcc{-W}{W}$. 
We observe that an $\Oh(s)\times \Oh(s)$ Monge matrix $M$ satisfying this property has just $\Oh(Ws)$ so-called \emph{core elements}: pairs $(i,j)$ such that $M_{i,j}+M_{i+1,j+1}\ne M_{i,j+1}+M_{i+1,j}$.
As shown by Russo~\cite{Russo10}, such matrices can be stored in $\Ohtilde(Ws)$ space and their $(\min,+)$ product can be computed in $\Ohtilde(Ws)$ time.
We apply this result (recently improved in~\cite{GGK24}) in \cref{sec:fast-monge-matrix-multiplication}, which provides a complete framework for \emph{core-sparse} Monge matrices. 
Using this framework, the distance matrices can be stored in $\Ohtilde(k)$ space and their $(\min,+)$ products can be computed in $\Ohtilde(k)$ time, thus circumventing one of the two bottlenecks.

\subparagraph*{Technical Contribution 3: Exploiting Shared Structures between Graphs \boldmath$G_i$}
With the improved multiplication time, the remaining bottleneck is the construction of individual distance matrices. 
Building one such matrix requires $k^{2-o(1)}$ time (because it stores weighted edit distance of length-$\Theta(k)$ strings), so we have to exploit some shared structures between the graphs $G_i$ for $i\in F$.
Combining core-sparse Monge matrix multiplication with insights behind the $\Oh(n+\sqrt{nk^3})$-time algorithm of \cite{CKW23} lets us construct these matrices in $\Ohtilde(k^{2.5})$ total time, which is still unsatisfactory.

Our improved solution shares the following key observation with \cite{CKW23}: due to $\sed(X)=\Oh(k)$, the string $X$ can be factorized into $\Oh(k)$ individual characters (that the optimal self-alignment edits) and $\Oh(k)$ factors with an earlier occurrence $\Oh(k)$ positions earlier (that the self-alignment matches perfectly). 
Using this factorization, we show that distinct phrases $X_{i-1}\ne X_i$ are still very similar \emph{on average}:
for each $i\in F$, we issue a sequence of $k_i$ operations extending $X_{i-1}$ to $X_{i-1}X_{i}$ so that $\sum_{i\in F} k_i = \Ohtilde(k)$ and each operation extends the current string with a single character (an insertion) or a substring of the current string (a copy-paste).
Finally, $X_i$ is obtained from $X_{i-1}X_i$ with a prefix deletion.
Similar $\Ohtilde(k)$ updates transform $Y_{i-1}$ to $Y_i$ as we iterate over all $i\in F$.

As explained below, we will store boundary distance matrices using a dynamic algorithm so that $\BMw(X_i,Y_i)$ can be obtained from $\BMw(X_{i-1},Y_{i-1})$ using just $\Oh(k_i)$ operations implemented in $\Ohtilde(k)$ time each, for a total of $\Ohtilde(k^2)$ time across all $i\in F$.
The techniques described so far to compute $\wed(X,Y)$ in the case of small self-edit distance are presented in \cref{sec:small-sed-algo}.

\subparagraph*{Dynamic Algorithm for Unbounded Unweighted Edit Distance}\label{overview:dynamic-unbounded}
To maintain $\BMw(X,Y)$, we build upon a dynamic algorithm of Charalampopoulos, Kociumaka, and Mozes~\cite{CKM20}, which supports arbitrary edits as updates in $\Ohtilde(|X|+|Y|)$ time each.
Although the original formulation is limited to \emph{unweighted} edit distance, it can be easily adapted to a fixed integer weight function using our core-sparse Monge matrix framework instead of Tiskin's \emph{simple unit-Monge} matrices~\cite{Tis08,Tiskin10}.

To describe the algorithm, let us first assume for simplicity that $|X|=|Y|=n$ is a power of two. 
Then, $\BMw(X,Y)$ can be constructed in $\Ohtilde(n^2)$ time by decomposing, for each scale $s\in \fragmentcc{0}{\log n}$, the strings $X$ and $Y$ into fragments of length $2^s$ and building the boundary distance matrix for each pair of fragments of length~$2^s$.  
For $s=0$, this matrix can be easily constructed in constant time; for $s>0$, we can combine four boundary distance matrices of the smaller scale.
This approach readily supports dynamic strings subject to substitutions: upon each update, it suffices to recompute all the affected boundary matrices. 
For each scale $s$, the number of such matrices is $\Oh(n/2^s)$, and each of them is constructed in $\Ohtilde(2^s)$ time; this yields a total update time of $\Ohtilde(n)$.
To support insertions and deletions as updates, the algorithm of \cite{CKM20} relaxes the hierarchical decomposition so that the lengths of fragments at each scale $s$ may range from $2^{s-2}$ to $2^{s+1}$. 
With appropriate rebalancing, the simple strategy of recomputing all the affected boundary matrices still takes $\Ohtilde(n)$ time per update.

\subparagraph*{Technical Contribution 4: Supporting Copy-Pastes with Balanced Straight Line Programs}
Even with the extension to small integer weights, the algorithm of \cite{CKM20} falls short of our requirements as it only supports character edits as updates,
while we also need the copy-paste operation that extracts a substring and appends its copy at the end of the string.
Consider the hierarchical decomposition that the algorithm of \cite{CKM20} maintains for $X$.
Although a copy-paste on $X$ adds a lot of new fragments to the decomposition, we expect most of them to match fragments of the pre-existing decomposition of $X$; for such fragments, no new boundary distance matrices need to be computed.
To keep track of matching fragments, it is convenient to represent the hierarchical decomposition as a directed acyclic graph obtained by gluing together isomorphic subtrees.
In data compression, such a representation of a binary tree is called \emph{a straight-line program} (SLP).
Each node of the DAG can then be interpreted as a symbol in a context-free grammar, with leaves (sink nodes) corresponding to terminals and the remaining nodes to non-terminals. 
Each non-terminal has exactly one production corresponding to the pair of outgoing edges. 
These properties guarantee that the language associated with each symbol consists of one string: the \emph{expansion} of the symbol.

Established tools~\cite{Ryt03,CLLPPSS05} can be used to maintain an SLP with operations that create a symbol whose expansion is the concatenation of the expansions of two existing symbols or a prescribed substring of the expansion of an existing symbol.
Crucially, these operations add only logarithmically many auxiliary symbols.
Unfortunately, unlike the hierarchical decompositions of \cite{CKM20}, these SLPs lack a structure of levels, so it is not clear for which pairs of symbols the boundary distance matrix should be stored.
To address this issue, we rely on an additional invariant maintained in~\cite{CLLPPSS05}: for every non-terminal $A$ with production $A\to A_LA_R$, the lengths of expansions of $A_L$ and $A_R$ are within a constant factor of each other.
Based on this, we choose to store the boundary distance matrix for every pair of symbols $(A,B)$ (originating from the SLP of $X$ and the SLP of $Y$, respectively) such that the expansions of $A$ and $B$ are also within a constant factor of each other.
This way, when a new symbol $A$ is added to the SLP of $X$, we need to build the boundary distance matrix for $\Oh(|Y|/\len(A))$ symbols $B$ in the SLP of $Y$, where $\len(A)$ denotes the length of the expansion of $A$, and each matrix is constructed in $\Ohtilde(\len(A))$ time, for a total of $\Ohtilde(|Y|)$.
The details are described in \cref{sec:hierarchical}, where we also discuss combining boundary distance matrices.

\paragraph*{Faster Static Algorithm for Large Integer Weights}
The overview above explains how to achieve $\Ohtilde(n+k^2)$ runtime for a fixed integer weight function, but the underlying techniques yield an $\Ohtilde(n+W\cdot k^2)$-time solution if a weight function $w:\Esigma^2 \to \fragmentcc{0}{W}$ is given through an oracle.
The only major challenge arising in this generalization is to construct an approximate alignment for the divide-and-conquer procedure.
A simple but effective trick is to compute a $w'$-optimal alignment for a weight function $w'$ defined as $w'(a,b)=\ceil{w(a,b)/2}$. 
Since $w$ and $w'$ are within a factor of $2$ from each other, this yields a $2$-approximate $w$-optimal alignment.
Moreover, repeating this reduction $\lceil\log W\rceil$ times, we arrive at the problem of computing an unweighted alignment, which we solve in time $\Oh(n+k^2)$ using the Landau--Vishkin algorithm~\cite{LV88}. 

If the weights are large, the $(\min,+)$-product of $m\times m$-sized integer-valued Monge matrices takes $\Thtilde(m^2)$ time to be computed, just like in the fractional case.
In our application, however, we care about distances at most $k$, and thus only small entries need to be computed truthfully. 
In \cref{sec:fast-monge-matrix-multiplication,sec:monge-appendix}, we develop a novel procedure of independent interest that, given an $m\times m$ Monge matrix $M$ with non-negative integer entries, constructs an $m\times m$ Monge matrix $M'$ whose core is of size $\Oh(m\sqrt{k})$ yet $\min\{M_{i,j},k+1\}=\min\{M'_{i,j},k+1\}$ holds for every entry. 
Applying this reduction to all distance matrices that our algorithm constructs, we get an $\Ohtilde(n+k^{2.5})$-time static algorithm, \mbox{which is faster than all known alternatives when $k \le \min\{W^2,\sqrt{n}\}$}.

\paragraph*{Dynamic Algorithm for Bounded Edit Distance}
In order to turn the procedure behind \cref{thm:static_fixed} into a dynamic algorithm, we maintain the strings $X$ and $Y$ in a standard dynamic data structure~\cite{MSU97,GKKLS18,KK22} that provides $\Ohtilde(1)$-time random access and equality tests (whether two fragments match perfectly) at the cost of an $\Ohtilde(1)$-time additive overhead for each update.
Below, we highlight further major adaptations.

\subparagraph*{The Case of Small Self-Edit Distance}
The central component of our solution is a subroutine that, given a threshold $k$, maintains $\wed(X,Y)$ for $\Oh(k)$ updates, under a promise that $\sed(X)=\Oh(k)$ and $\wed(X,Y)=\Oh(k)$. 
Its high-level structure resembles the static counterpart illustrated in \cref{fig:substitutions-only}, but the subgraphs $G_i$ are slightly taller so that their union $G$ covers a sufficiently wide stripe even after $\Oh(k)$ updates.
At the initialization time, we still identify $\Oh(k)$ indices $i \in \fragmentco{0}{m}$ such that $(X_i,Y_i)\ne (X_{i-1},Y_{i-1})$, and we derive the corresponding matrices $D_{i,i+1}$ from the boundary distance matrices $\BMw(X_i,Y_i)$.
The key difference is that we rely on the dynamic algorithm employed to construct boundary distance matrices being \emph{fully persistent}: each update creates a new instance of the algorithm and one can still issue updates to the original instance.
As a result, $\BMw(X_i,Y_i)$ (and hence $D_{i,i+1}$) can be updated in $\Ohtilde(k)$ time whenever an update affects $X_i$ or $Y_i$.
The $(\min,+)$-product $D_{0,m}=\bigotimes_{i=0}^{m-1} D_{i,i+1}$ also needs to be maintained dynamically: for this, we employ a perfectly balanced binary tree with leaves storing the subsequent matrices $D_{i,i+1}$ and internal nodes maintaining the product of its children's matrices.
In order to secure $\Ohtilde(k^2)$-time initialization, identical subtrees are glued together at initialization time (so that each run of identical matrices $D_{i,i+1}$ requires creating just $\Oh(\log m)$ nodes), and the updates are implemented in a functional way (we create a new copy of each modified node since the original may still be linked from elsewhere).
With these modifications, in \cref{sec:dyn-smallsed}, we achieve $\Ohtilde(k^2)$-time initialization and $\Ohtilde(k)$-time updates.

\subparagraph*{Dynamic Divide-and-Conquer Implementation}
In its basic version, our implementation of the divide-and-conquer algorithm is given a threshold $k$ and maintains $\wed(X,Y)$ for $\Oh(k)$ updates under a promise that $\wed(X,Y)=\Oh(k)$. 
At the initialization time, we use the $\Ohtilde(k^2)$-time static algorithm to build an optimal alignment $\cA$ that will guide our recursive decomposition: we partition $X=X_LX_R$ and $Y=Y_LY_R$ so that $\wed_{\boldsymbol \cA}(X_L,Y_L)+\wed_{\boldsymbol \cA}(X_R,Y_R)=\wed_{\boldsymbol \cA}(X,Y)$ and recursively maintain $\wed(X_L,Y_L)$ and $\wed(X_R,Y_R)$.
Crucially, the recursion has negligible overhead because each update affects only one of these subroutines.
As in the static case, the optimal alignments $\cB_L : X_L \onto Y_L$ and $\cB_R : X_R \onto Y_R$ are combined into an optimal alignment $\cB: X\onto Y$ using an optimal alignment $\cB_M : X_M \onto Y_M$ for an appropriate fragment $X_M$ around the midpoint of $X$ and its image $Y_M$ under $\cA$ (see \cref{fig:new-divide-and-conquer}). 
In the dynamic case, we need to be more generous when initializing $X_M$ since updates may increase the cost of $\cA$ and decrease the self-edit distances. 
Nevertheless, we can still maintain $\wed(X_M,Y_M)$ using the subroutine from the previous paragraph.

The remaining challenge is to make sure that, despite occasional rebuilding, the update time is $\Ohtilde(k)$ in the worst case rather than just amortized. 
Such a deamortization is needed at each level of the recursive procedure, so standard techniques would yield a polylogarithmic-factor penalty.
In \cref{sec:dyn-general}, we propose a tailor-made approach whose (rather technical) complexity analysis proves that we incur just a constant-factor overhead, and hence achieve $\Oh(k\log^2 n)$ update time for any fixed weight function, matching the results of~\cite{CKM20} for $k=\Theta(n)$ and unweighted edit distance.

\section{Preliminaries}\label{sec:prelim}

We partially follow the narration of \cite{CKW23} in the preliminaries.

\paragraph*{Strings} A string $X = X\position{0}\cdots X\position{n-1} \in \Sigma^n$
is a sequence of $|X| = n$ characters over an
alphabet $\Sigma$; \(|X|\) is the \emph{length} of \(X\).
For a \emph{position} $i \in \fragmentco{0}{n}$, we
say that $X\position{i}$ is the $i$-th character of $X$.
We denote the empty string over $\Sigma$ by $\emptystring$.
Given indices $i, j \in \fragmentcc{0}{|X|}$ satisfying $i \leq j$, we say that $X\fragmentco{i}{j} \coloneqq X\position{i}\cdots X\position{j-1}$ is a
\emph{fragment} of $X$.
We may also write $X\fragmentcc{i}{j-1}, X\fragmentoc{i-1}{j-1}$,
or $X\fragmentoo{i-1}{j}$
for the fragment $X\fragmentco{i}{j}$.

We say that $X$ occurs as a substring of a string $Y$, if there are
$i, j \in \fragmentcc{0}{|Y|}$ satisfying $i \le j$ such that $X = Y\fragmentco{i}{j}$.

\paragraph*{Alignments and (Weighted) Edit Distances}
We start with the crucial notion of an \emph{alignment}, which gives us a formal way to
describe a sequence of edits to transform one string into another.

\begin{definition}[Alignment,~{\cite[Definition 2.3]{DGHKS23}}]
    A sequence $\cA = (x_t, y_t)_{t=0}^m$ is an \emph{alignment} of $X\fragmentco{x}{x'}$ onto
    $Y\fragmentco{y}{y'}$, denoted by $\cA: X\fragmentco{x}{x'} \onto Y\fragmentco{y}{y'}$ if
    $(x_0, y_0) = (x, y)$, $(x_m, y_m) = (x', y')$, and $(x_{t+1}, y_{t+1}) \in \{(x_t+1, y_t), (x_t, y_t+1), (x_t+1, y_t+1)\}$ for all $t \in \fragmentco{0}{m}$.

    We write
    $\Als(X\fragmentco{x}{x'}, Y\fragmentco{y}{y'})$
    for the set of all alignments of $X\fragmentco{x}{x'}$ onto $Y\fragmentco{y}{y'}$.
    \lipicsEnd
\end{definition}

For an alignment $\cA = (x_t, y_t)_{t=0}^m\in \Als(X\fragmentco{x}{x'},
Y\fragmentco{y}{y'})$ and an index $t \in \fragmentco{0}{m}$, we say that
\begin{itemize}
    \item $\cA$ \emph{deletes} $X\position{x_t}$ if $(x_{t+1}, y_{t+1}) = (x_t+1, y_t)$;
    \item $\cA$ \emph{inserts} $Y\position{y_t}$ if $(x_{t+1}, y_{t+1}) = (x_t, y_t+1)$;
    \item $\cA$ \emph{aligns} $X\position{x_t}$ to $Y\position{y_t}$, denoted by
        $X\position{x_t} \aonto{\cA} Y\position{y_t}$
    if $(x_{t+1}, y_{t+1}) = (x_t+1, y_t+1)$;
    \item $\cA$ \emph{matches} $X\position{x_t}$ with $Y\position{y_t}$
        if $X\position{x_t}
        \aonto{\cA} Y\position{y_t}$ and
    $X\position{x_t} = Y\position{y_t}$;
    \item $\cA$ \emph{substitutes} $X\position{x_t}$ for $Y\position{y_t}$ if
        $X\position{x_t} \aonto{\cA} Y\position{y_t}$ but
    $X\position{x_t} \neq Y\position{y_t}$.
\end{itemize}
Insertions, deletions, and substitutions are jointly called (character) \emph{edits}.

Given $\cA = (x_t, y_t)_{t=0}^m \in \Als(X, Y)$, we define the \emph{inverse alignment}
as $\cA^{-1} \coloneqq (y_t, x_t)_{t=0}^m \in \Als(Y, X)$.

Given an alphabet $\Sigma$, we set $\Esigma \coloneqq \Sigma \cup \{\emptystring\}$.
We call $w$ a \emph{weight function} if $w : \Esigma \times \Esigma \to \Real_{\geq 0}$,
and for $a, b \in \Esigma$, we have $\w{a}{b} = 0$ if and only if $a = b$.
Note that \(w\) does not need to satisfy the triangle inequality nor does \(w\) need to be
symmetric.

We write $\wed_\cA(X\fragmentco{x}{x'}, Y\fragmentco{y}{y'})$ for
the \emph{cost} of an alignment $\cA \in \Als(X\fragmentco{x}{x'}, Y\fragmentco{y}{y'})$ with
respect to a weight function $w$, that is, for the total cost of edits made by $\cA$, where
\begin{itemize}
    \item the cost of deleting $X\position{x}$ is $\w{X\position{x}}{\emptystring}$,
    \item the cost of inserting $Y\position{y}$ is $\w{\emptystring}{Y\position{y}}$,
    \item the cost of aligning $X\position{x}$ with $Y\position{y}$ is
        $\w{X\position{x}}{Y\position{y}}$.
\end{itemize}

We define the \emph{weighted edit distance} of strings $X, Y \in \Sigma^*$
with respect to a weight function $w$ as
$\wed(X, Y) \coloneqq \min_{\cA \in \Als(X, Y)} \wed_\cA(X, Y)$.
For an integer $k \geq 0$, we also define a capped version
\[
    \wed_{\leq k}(X, Y) \coloneqq
        \begin{cases}
            \wed(X, Y) & \text{if } \wed(X, Y) \leq k,\\
            \infty & \text{otherwise}.
        \end{cases}
\]

\begin{definition}[Alignment Graph,~{\cite[Definition 3.2]{CKW23}}]\label{def:alignment-graph}
For strings $X,Y\in \Sigma^*$ and a weight function $w: \Esigma^2 \to \Real_{\ge 0}$,
we define the \emph{alignment graph} $\AG^w(X,Y)$ as follows.
\(\AG^{w}(X, Y)\) has vertices $\fragmentcc{0}{|X|}\times \fragmentcc{0}{|Y|}$,
\begin{itemize}
    \item horizontal edges $(x,y)\to (x+1,y)$ of cost $\w{X\position{x}}{\emptystring}$
        for $(x,y)\in \fragmentco{0}{|X|}\times \fragmentcc{0}{|Y|}$,
    \item vertical edges $(x,y)\to (x,y+1)$ of cost $\w{\emptystring}{Y\position{y}}$
        for $(x,y)\in \fragmentcc{0}{|X|}\times \fragmentco{0}{|Y|}$, and
    \item diagonal edges $(x,y)\to (x+1,y+1)$ of cost $\w{X\position{x}}{Y\position{y}}$
        for $(x,y)\in \fragmentco{0}{|X|}\times \fragmentco{0}{|Y|}$.
        \lipicsEnd
\end{itemize}
\end{definition}

We visualize the alignment graph $\AG^w(X, Y)$ as a grid graph with $|X|+1$ columns
and $|Y|+1$ rows. We think of the vertex $(0,0)$ as the top left vertex of the grid,
and a vertex $(x, y)$ in the $x$-th column and $y$-th row; see \cref{fig:alignment-graph}.

Observe that we can interpret $\Als(X\fragmentco{x}{x'},Y\fragmentco{y}{y'})$
as the set of $(x,y)\leadsto (x',y')$ paths in $G \coloneqq \AG^w(X,Y)$.
Moreover, $\wed_\cA(X\fragmentco{x}{x'},Y\fragmentco{y}{y'})$
is the cost of $\cA$ interpreted as a path in $G$, and
thus $\ed^w(X\fragmentco{x}{x'},Y\fragmentco{y}{y'})=\dist_G((x,y), (x',y'))$.

If for every $a, b \in \Esigma$ we have that $\w{a}{b} = 1$ if $a \neq b$
and $\w{a}{b} = 0$ otherwise, then $\wed(X, Y)$ corresponds
to the standard \emph{unweighted} edit distance (also known as Levenshtein distance~\cite{Levenshtein66}).
For this case, we drop the superscript $w$ in $\wed$, $\wed_\cA$, $\wed_{\le k}$, and $\AGw$.

We say that an alignment $\cA \in \Als(X, Y)$ is $w$-optimal if $\wed_\cA(X, Y) = \wed(X, Y)$ and an optimal unweighted alignment if $\ed_\cA(X, Y) = \ed(X, Y)$. 

A weight function $w$ is called \emph{normalized} if
$w(a,b)\ge 1$ holds for all $a,b\in \Esigma$ with $a\ne b$.

\begin{fact}[Slight Generalization of {\cite[Proposition 2.16]{DGHKS23}}] \label{lm:baseline-wed}
    Given strings $X, Y \in \Sigma^*$, an integer $k \geq 1$, and (oracle access to)
    a normalized weight function
    $w: \Esigma^2 \to \RR_{\geq 0}$, the value $\wed_{\leq k}(X, Y)$ can be computed in
    $\Oh(\min\{|X| + 1, |Y| + 1\} \cdot \min\{k, |X| + |Y| + 1\})$ time.
    Furthermore, if $\wed(X, Y) \le k$, the algorithm returns a $w$-optimal alignment of $X$ onto $Y$.
\end{fact}

\begin{proof}
    Consider any $w$-optimal alignment $\cA$.
    If $\wed(X, Y) \le k$, $\cA$ may contain at most $k$ vertical and horizontal edges as each one of them costs at least one.
    Therefore, all vertices $(x, y)$ of $\cA$ satisfy $|x - y| \le k$.
    These vertices lie on $\min\{2k + 1, |X| + |Y| + 1\}$ diagonals of $\AGw(X, Y)$.
    Each such diagonal consists of at most $\min\{|X|+ 1, |Y| + 1\}$ vertices.
    Therefore, the subgraph of $\AGw(X, Y)$ induced by such vertices has size $\Oh(\min\{|X|+ 1, |Y| + 1\} \cdot \min\{k, |X| + |Y| + 1\})$ and can be computed in the same time complexity.
    Hence, finding $\wed(X, Y)$ if $\wed(X, Y) \le k$ is equivalent to finding the shortest path from $(0, 0)$ to $(|X|, |Y|)$ in such a subgraph.
    As the graph is acyclic, we can find such a shortest path in $\Oh(\min\{|X|+ 1, |Y| + 1\} \cdot \min\{k, |X| + |Y| + 1\})$ time along with the distance.
    If this distance is at most $k$, we return it along with the path we found.
    Otherwise, we return $\infty$.
\end{proof}

The \emph{breakpoint representation} of an alignment $\cA=(x_t,y_t)_{t=0}^m\in \Als(X,Y)$
is the subsequence of $\cA$ consisting of pairs $(x_t,y_t)$ such that $t\in \{0,m\}$ or
$\cA$ does not match $X\position{x_t}$ with $Y\position{y_t}$.
Note that the size of the breakpoint representation is at most $2+\ed_\cA(X,Y)$ and that it
can be used to retrieve the entire alignment:
for any two consecutive elements $(x',y'),(x,y)$ of the breakpoint representation, it
suffices to add $(x-\delta,y-\delta)$ for $\delta \in \fragmentoo{0}{\max(x-x',y-y')}$.

Given an alignment $\cA = (x_t, y_t)_{t=0}^m \in \Als(X, Y)$, for every $\ell, r \in
\fragmentcc{0}{m}$ with $\ell \leq r$
we say that $\cA$ aligns $X\fragmentco{x_\ell}{x_r}$ onto $Y\fragmentco{y_\ell}{y_r}$ and
denote it by
$X\fragmentco{x_\ell}{x_r} \aonto{\cA} Y\fragmentco{y_\ell}{y_r}$.
We denote the cost of the induced alignment
of $X\fragmentco{x_\ell}{x_r}$ onto $Y\fragmentco{y_\ell}{y_r}$
by $\wed_\cA(X\fragmentco{x_\ell}{x_r}, Y\fragmentco{y_\ell}{y_r})$.

Given $\cA: X \onto Y$ and a fragment $X\fragmentco{x_\ell}{x_r}$ of $X$, we write $\cA(X\fragmentco{x_\ell}{x_r})$ for the fragment
$Y\fragmentco{y_\ell}{y_r}$ of $Y$ where
\[
    y_\ell \coloneqq \min\{y \mid (x_\ell, y) \in \cA\} \quad \text{ and } \quad
        y_r \coloneqq \begin{cases} |Y| & \text{if } x_r = |X|, \\ \min\{y \mid (x_r, y) \in \cA\} &\text{otherwise.}\end{cases}
\]
Intuitively, $\cA(X\fragmentco{x_\ell}{x_r})$ is the fragment that
$\cA$ aligns $X\fragmentco{x_\ell}{x_r}$ onto.

\begin{fact}[{\cite[Lemma 3.7]{CKW23}}] \label{fct:split-alignment}
    Let $X, Y \in \Sigma^*$ denote strings and let $\cA \in \Als(X,Y)$ denote an alignment.
    For every $(x,y)\in \cA$, we have
    \[\wed(X,Y)\le \wed(X\fragmentco{0}{x}, Y\fragmentco{0}{y}) +
    \wed(X\fragmentco{x}{|X|}, Y\fragmentco{y}{|Y|})\le \wed_{\cA}(X,Y).\lipicsEnd\]
\end{fact}

\paragraph*{The PILLAR Model}
Charalampopoulos, Kociumaka, and Wellnitz~\cite{CKW20} introduced the
\modelname{} model.
The \modelname{} model provides an abstract interface to a set of primitive operations on
strings which can be efficiently implemented in different settings. Thus, an algorithm
developed using the \modelname{} interface does not only yield algorithms in the standard
setting, but also directly yields algorithms in diverse other settings,
for instance, fully compressed, dynamic, etc.

\newcommand{\calX}{\mathcal{X}}
In the~\modelname{} model we are given a family $\calX$ of strings to preprocess. The
elementary objects are fragments $X\fragmentco{\ell}{r}$ of strings $X \in \calX$.
Initially, the model gives access to each
$X \in \calX$ interpreted as $X\fragmentco{0}{|X|}$. Other fragments can be retrieved via
an \extractOpName{} operation:
\begin{itemize}
    \item $\extractOpName(S, \ell, r)$: Given a fragment $S$ and positions $0 \leq \ell \leq r \leq |S|$, extract
    the fragment $S\fragmentco{\ell}{r}$, which is defined as $X\fragmentco{\ell' + \ell}{\ell' + r}$ if $S = X\fragmentco{\ell'}{r'}$ for
    $X \in \calX$.
\end{itemize}
Moreover, the \modelname model provides the following primitive operations for fragments $S$ and $T$~\cite{CKW20}:
\begin{itemize}
    \item $\lceOp{S}{T}$: Compute the length of~the longest common prefix of~$S$ and $T$.
    \item $\lcbOp{S}{T}$: Compute the length of~the longest common suffix of~$S$ and $T$.
    \item $\accOpName(S,i)$: Assuming $i\in \fragmentco{0}{|S|}$, retrieve the character $\accOp{S}{i}$.
    \item $\lenOpName(S)$: Retrieve the length $|S|$ of~the string $S$.
\end{itemize}
Observe that in the original definition~\cite{CKW20}, the \modelname{} model
also includes an \(\ipmOpName\) operation to find all (internal) exact occurrences of one
fragment in another. We do not need the \(\ipmOpName\) operation in this
work.\footnote{We still use the name \modelname, and not {\tt PLLAR}, though.}
\modelname{} operations can be implemented to work in constant time in the static setting \cite[Section 7.1]{CKW20} and in polylogarithmic time in the dynamic setting \cite[Section 8]{KK22}, \cite[Section 7.2]{CKW22}.

We use the following version of the classic Landau--Vishkin algorithm \cite{LV88} in the \modelname{} model.

\begin{fact}[{\cite[Lemma 6.1]{CKW20}}]\label{lm:k2-ed}
    There is a \modelname{} algorithm that, given strings $X, Y \in \Sigma^*$, in time $\Oh(k^2)$ computes $k = \ed(X, Y)$ along with the breakpoint representation of an optimal unweighted alignment of $X$ onto $Y$.\lipicsEnd
\end{fact}

\paragraph*{Monge Matrix Toolbox}\label{sec:alg:sec:alg-periodic:sec:planar-toolbox}

In this paper, we extensively use distance matrices in alignment graphs of pairs of strings. As alignment graphs are planar, such matrices are so-called Monge matrices.

\begin{definition} \label{def:monge}
    A matrix $A$ of size $p \times q$ is called \emph{Monge} if for $i \in [1 \dd p)$ and $j \in [1 \dd q)$, we have $A_{i, j} + A_{i + 1, j + 1} \le A_{i, j + 1} + A_{i + 1, j}$.
    \lipicsEnd
\end{definition}

\begin{fact}[{\cite[Section 2.3]{FR06}}]\label{fct:monge}
    Consider a directed planar graph $G$ with non-negative edge weights. For vertices
    $u_1,\ldots,u_{p},v_{q},\ldots,v_1$ lying (in this cyclic order, potentially with repetitions) on the outer face
    of $G$, define a $p\times q$ matrix $D$ with $D_{i,j}=\dist_G(u_i,v_j)$.
    If all entries of $D$ are finite, $D$ is Monge.
\end{fact}

\begin{proof}[Proof Sketch]
    We need to prove $D_{i, j} + D_{i + 1, j + 1} \le D_{i, j + 1} + D_{i + 1, j}$ for $i \in \fragmentco{1}{p}$ and $j \in \fragmentco{1}{q}$.
    Note that $D_{i, j + 1}=\dist_G(u_i,v_{j+1})$ and $D_{i + 1, j}=\dist_G(u_{i+1},v_j)$.
    Fix a shortest path $P$ from $u_i$ to $v_{j + 1}$ and a shortest path $Q$ from $u_{i + 1}$ to $v_j$; see \cref{fig:monge-planar}.
    As these four vertices lie in the order $u_i, u_{i + 1}, v_{j + 1}, v_j$ on the outer face of $G$ (some of these vertices may coincide with each other), the paths $P$ and $Q$ intersect at a common vertex $w \in V(G)$.
    As $P$ and $Q$ are shortest paths, we have
    \begin{align*}
        D_{i, j + 1}\ + D_{i + 1, j} &= (\dist_G(u_i, w) + \dist_G(w, v_{j + 1})) + (\dist_G(u_{i + 1}, w) + \dist_G(w, v_j))\\
                                     &= (\dist_G(u_i, w) + \dist_G(w, v_{j})) + (\dist_G(u_{i + 1}, w) + \dist_G(w, v_{j + 1}))\\
                                     &\ge D_{i, j} + D_{i + 1, j + 1}.\qedhere
    \end{align*}
\end{proof}

\begin{figure}
    \begin{center}
        \includegraphics[scale=0.5]{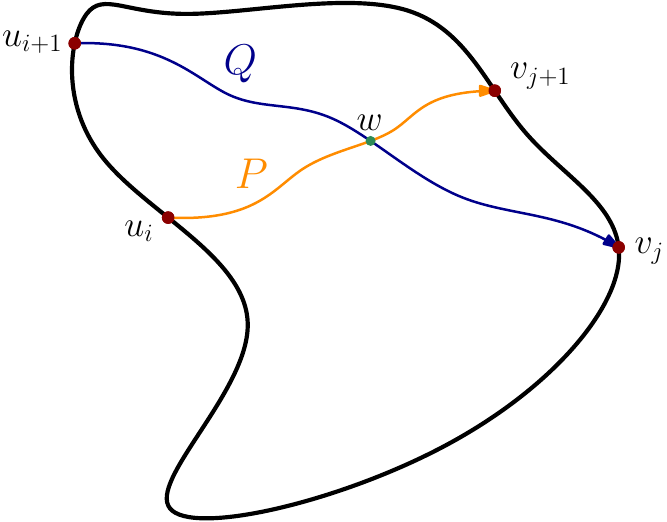}
    \end{center}

    \caption{The sum of the distances from $u_i$ to $v_{j + 1}$ and from $u_{i + 1}$ to $v_{j}$ is not smaller than the sum of the distances from $u_i$ to $v_j$ and from $u_{i + 1}$ to $v_{j + 1}$.}
    \label{fig:monge-planar}
\end{figure}

When talking about distance matrices in a graph, combining paths corresponds to the min-plus product of such matrices.

\begin{definition}
    Given a matrix $A$ of size $p \times q$ and a matrix $B$ of size $q \times r$, define their \emph{min-plus product} $A \otimes B$ as a matrix $C$ of size $p \times r$, where $C_{i, k} = \min_j A_{i, j} + B_{j, k}$ for $i \in [1\dd p], k \in [1\dd r]$.
    \lipicsEnd
\end{definition}

Furthermore, if matrices $A$ and $B$ are Monge, then their min-plus product $A \otimes B$ is also Monge.

\begin{fact}[{\cite[Theorem 2]{Tiskin10}}] \label{fct:monge-product-is-monge}
    Let $A, B$, and $C$ be matrices, such that $A \otimes B = C$.
    If $A$ and $B$ are Monge, then $C$ is also Monge.
\lipicsEnd
\end{fact}

One of the most celebrated results concerning Monge matrices is the SMAWK algorithm of \cite{SMAWK87}.

\begin{fact}[\cite{SMAWK87}]\label{thm:smawk-row-maxima}
    There is an algorithm that, given random access to the entries of a $p \times q$ Monge matrix, computes its row and column minima and their positions in time $\Oh(p + q)$.
\lipicsEnd
\end{fact}

When talking about Monge matrices, notions of matrix core and density matrix are of fundamental importance.

\begin{definition}
    For a matrix $A$ of size $p \times q$, define its \emph{density matrix} $\dens{A}$ of size $(p - 1) \times (q - 1)$, where $\dens{A}_{i, j} = A_{i, j + 1} + A_{i + 1, j} - A_{i, j} - A_{i + 1, j + 1}$ for $i \in [1 \dd p), j \in [1 \dd q)$.

    We define \emph{the core} of a matrix $A$ as $\core{A} \coloneqq \{(i, j, \dens{A}_{i, j}) \mid i \in [1 \dd p), j \in [1 \dd q), \dens{A}_{i, j} \neq 0 \}$. We denote \emph{the size of the core} $\delta(A) \coloneqq |\core{A}|$.
    \lipicsEnd
\end{definition}

Note that for a Monge matrix $A$, all entries of its density matrix are non-negative, and thus $\core{A}$ consists of triples $(i, j, v)$ with some positive values $v$.

For any matrix $A$ of size $p \times q$ we write $A[a \dd b)[c \dd d)$ for any $a, b \in [1 \dd p]$ with $a < b$ and $c, d \in [1 \dd q]$ with $c < d$ to denote a contiguous submatrix of $A$ consisting of all entries on the intersection of rows $[a \dd b)$ and columns $[c \dd d)$ of $A$.

\section{\boldmath Efficient Algorithms for \texorpdfstring{$(\min,+)$}{(min,+)} Multiplication of Core-Sparse Monge Matrices} \label{sec:fast-monge-matrix-multiplication}

As discussed in \cref{sec:alg:sec:alg-periodic:sec:planar-toolbox}, we will require tools to work with Monge matrices and their min-plus products.
SMAWK algorithm \cite{SMAWK87} allows min-plus multiplication of Monge matrices in near-linear time in terms of the sizes of the input matrices.
However, it is still not fast enough for our applications.
The Monge matrices we will be working with are ``sparse'' in the sense that they have small cores.
We would like to use algorithms that operate on Monge matrices, and their time complexities depend on the core sizes, not on matrix sizes.
Such tools were first developed for so-called ``unit-Monge matrices'' by Tiskin in~\cite{Tis08,Tiskin10}.
Later, similar results for the general case were presented by Russo in~\cite{Russo10} and very recently improved by Gawrychowski, Gorbachev, and Kociumaka in~\cite{GGK24}.
In this section, whenever we speak about matrices, we fix some positive integer $N$ and assume that the dimensions of all the matrices are bounded by $N$.

\begin{fact}[{\cite[Lemma 3.8]{GGK24}}] \label{lm:core-bound}
Let $A$ and $B$ denote Monge matrices of sizes $p \times q$ and $q \times r$, respectively.
We have $\delta(A \otimes B) \le 2 \cdot (\delta(A) + \delta(B))$.
\lipicsEnd
\end{fact}

\begin{definition}
    For a matrix $A$, we call the \emph{condensed representation of $A$} the values in the topmost row and the leftmost column of $A$ and the core of $A$.
\lipicsEnd
\end{definition}

\begin{fact}[{\cite[Theorem 4.7]{GGK24}}] \label{lm:ggk-monge-mult}
Let $A$ and $B$ denote Monge matrices of sizes $p \times q$ and $q \times r$, respectively.
 There is an algorithm that, given the condensed representations of $A$ and $B$, in time $\Oh((N + \delta(A) + \delta(B)) \cdot \log N)$ computes the condensed representation of $A \otimes B$.
\lipicsEnd
\end{fact}

However, note that the condensed representation of a Monge matrix does not provide fast access to the entries of this matrix.
To address this issue, we design the following data structure.

\begin{lemma} \label{lm:build_mtx_ds}
    There exists an algorithm that, given the condensed representation of a matrix $A$ of size $p \times q$, in time $\Oh(N + \delta(A) \log N)$ builds a \emph{Core-based Matrix Oracle data structure} $\mds(A)$ that provides the following interface.
    \begin{description}
        \item[Core range query:] given $a, b \in [1 \dd p]$ with $a < b$ and $c, d \in [1 \dd q]$ with $c < d$, in time $\Oh(\log N)$ returns $\ssum(\dens{A}[a\dd b)[c \dd d)) \coloneqq \sum_{i \in [a \dd b)} \sum_{j \in [c \dd d)} \dens{A}_{i, j}$.
        \item[Random access:] given $i \in [1 \dd p]$ and $j \in [1 \dd q]$, in time $\Oh(\log N)$ returns $A_{i, j}$.
        \item[Full core listing:] in time $\Oh(1 + \delta(A))$ returns $\core{A}$.
    \end{description}
\end{lemma}

\begin{proof}
    To store the elements of $\core{A}$, we use the following classical data structure.

    \begin{claim}[Weighted Orthogonal Range Counting {\cite[Theorem 3]{Wil85}}]\label{fct:WORC}
        There exists an algorithm that, given a collection of $n$ points $(x_i, y_i)_{i \in [1 \dd n]}$ with integer weights $(w_i)_{i \in [1 \dd n]}$, in time $\Oh(n \log n)$ builds a data structure that supports the following queries in $\Oh(\log n)$ time: given an orthogonal range $R=[x,x']\times[y,y']$, compute $\sum_{i : (x_i,y_i)\in R}w_i$.
    \lipicsClaimEnd
    \end{claim}

    We implement $\mds(A)$ for some matrix $A$ of size $p \times q$ in the following way.
    We store the condensed representation of $A$ explicitly.
    Furthermore, we store the weighted orthogonal range counting data structure of \cref{fct:WORC} over the elements of $\core{A}$ that in logarithmic time provides query access to the sum of the values of core elements of $A$ (and thus, values of $\dens{A}$) on a rectangle.
    We build it in time $\Oh(\delta(A) \log N)$ using \cref{fct:WORC}.
    Hence, the algorithm for building $\mds$ works in time $\Oh(N + \delta(A) \log N)$.

    We now describe how the queries are handled.
    We start with core range queries.
    Note that elements of $\core{A}$ are exactly all non-zero entries of $\dens{A}$.
    Hence, range queries to the entries of $\core{A}$ using the weighted orthogonal range counting data structure of \cref{fct:WORC} provides $\Oh(\log N)$-time queries for the sum of the elements of $\dens{A}$ in a range.

    For a random access query, say, $A_{i, j}$ is queried.
    Note that $A_{i, j} = A_{1, j} + A_{i, 1} - A_{1, 1} - \ssum(\dens{A}[1\dd i)[1 \dd j))$ by the definition of $\dens{A}$.
    We store $A_{1, j}, A_{i, 1}$, and $A_{1, 1}$ explicitly, and using a single core range query, $\ssum(\dens{A}[1\dd i)[1 \dd j))$ can be computed in time $\Oh(\log N)$.

    Finally, to answer the full core listing query, we just return $\core{A}$ that we store explicitly.
\end{proof}

It is now easy to see that given $\mds(A)$ and $\mds(B)$, one can efficiently compute $\mds(A \otimes B)$.

\begin{lemma} \label{lm:matrix_mult_w}
    There is an algorithm that given $\mds(A)$ and $\mds(B)$ for Monge matrices $A$ of size $p \times q$ and $B$ of size $q \times r$, builds $\mds(A \otimes B)$ in time $\Oh((N + \delta(A) + \delta(B)) \log N)$.
\end{lemma}

\begin{proof}
    Using \cref{lm:ggk-monge-mult}, we can obtain the condensed representation of $C \coloneqq A \otimes B$ in time $\Oh((N + \delta(A) + \delta(B)) \log N)$ as $\mds(A)$ and $\mds(B)$ store the condensed representations of $A$ and $B$ respectively explicitly.
    Applying Lemma \ref{lm:build_mtx_ds}, we get $\mds(C)$ in time $\Oh(N + \delta(C) \log N) = \Oh(N + (\delta(A) + \delta(B)) \log N)$ due to \cref{lm:core-bound}.
    The overall time complexity is $\Oh((N + \delta(A) + \delta(B)) \log N)$.
\end{proof}

Furthermore, we develop the following tool to perform basic operations with matrices using $\mds$:

\begin{lemma} \label{lm:ds_tools}
    The following facts hold.

    \begin{itemize}
        \item There is an algorithm that given $\mds(A)$ for some Monge matrix $A$, computes $\mds(A^T)$ in time $\Oh((N + \delta(A)) \log N)$.
        \item There is an algorithm that given $\mds(A)$ for some Monge matrix $A$, computes $\mds(C)$ in time $\Oh((N + \delta(A)) \log N)$ for any contiguous submatrix $C$ of $A$.
        \item There is an algorithm that given $\mds(A)$ and $\mds(B)$ for some Monge matrices $A$ and $B$ of sizes $p \times q$, computes $\mds(A + B)$ in time $\Oh((N + \delta(A) + \delta(B)) \log N)$.
        \item There is an algorithm that given $\mds(A)$ and $\mds(B)$ for Monge matrices $A$ of size $p \times q$ and $B$ of size $p \times r$ such that $(A \mid B)$ is Monge,\footnote{For matrices $A$ of size $p \times q$ and $B$ of size $p \times r$, we denote by $(A \mid B)$ the matrix $C$ of size $p \times (q + r)$ such that $C_{i, j} = A_{i, j}$ for $j \le q$ and $C_{i, j} = B_{i, j - q}$ for $j > q$. We call this operation ``stitching'' of matrices $A$ and $B$.} computes $\mds((A \mid B))$ in time $\Oh((N + \delta(A) + \delta(B)) \log N)$.
    \end{itemize}
\end{lemma}

\begin{proof}
    Note that $\mds(A)$ provides $\Oh(\log N)$-time random access to the entries of $A^T$ and any contiguous submatrix $C$ of $A$.
    Furthermore, given $\core{A}$, one can compute $\core{A^T}$ and $\core{C}$ in time $\Oh(\delta(A)+1)$.
    Moreover, $A^T$ and $C$ are Monge.
    Thus, we can build $\mds(A^T)$ and $\mds(C)$ in time $\Oh(N \log N + \delta(A) \log N)$ using \cref{lm:build_mtx_ds}, thus proving the first two claims.

    \smallskip
    
    Note that $\mds(A)$ and $\mds(B)$ provide $\Oh(\log N)$-time random access to the entries of $A + B$.
    Furthermore, note that $\dens{(A + B)} = \dens{A} + \dens{B}$. Thus, $\delta(A + B) \le \delta(A) + \delta(B)$, and $\core{A + B}$ can be computed in time $\Oh((\delta(A) + \delta(B)) \log N)$ by merging $\core{A}$ and $\core{B}$.
    Moreover, $A + B$ is Monge.
    Thus, we can build $\mds(A + B)$ in time $\Oh(N \log N + (\delta(A) + \delta(B)) \log N)$ using \cref{lm:build_mtx_ds}, thus proving the third claim.
    
    \smallskip

    Finally, $\mds(A)$ and $\mds(B)$ provide $\Oh(\log N)$-time random access to the entries of $(A \mid B)$.
    Furthermore, $\core{(A \mid B)}$ consists of the core entries of $A$, the core entries of $B$ (shifted), and potentially some core elements on the boundary between $A$ and $B$ that can be found in time $\Oh(N \log N)$ by accessing the entries of the rightmost column of $A$ and the leftmost column of $B$ through $\mds$.
    Thus, $\delta((A \mid B)) \le \delta(A) + \delta(B) + N$, and in time $\Oh(\delta(A) + \delta(B) + N \log N)$ we can compute $\core{(A \mid B)}$.
    Therefore, if $(A \mid B)$ is Monge, we can build $\mds((A\mid B))$ in time $\Oh((N + \delta(A) + \delta(B)) \log N)$ using \cref{lm:build_mtx_ds}, thus proving the fourth claim.
\end{proof}

In the regime of small integer edit weights, we will operate on bounded-difference Monge matrices.

\begin{definition}\label{def:bd}
    Let $\Delta\in \Zp$ be a positive integer and let $A\in \ZZ^{p\times q}$ be an integer matrix.
    We say that $A$ is $\Delta$-bounded-difference if
    \begin{itemize}
        \item $|A_{i,j}-A_{i,j+1}|\le \Delta$ holds for each $i\in [1\dd p]$ and $j\in [1\dd q)$, and
        \item $|A_{i,j}-A_{i+1,j}|\le \Delta$ holds for each $i\in [1\dd p)$ and $j\in [1\dd q]$.\lipicsEnd
    \end{itemize}
\end{definition}

Crucially, integer-valued bounded-difference Monge matrices have small cores.

\begin{lemma}\label{lem:bdtocore}
    If a Monge matrix $A\in \ZZ^{p\times q}$ is $\Delta$-bounded-difference, then $\delta(A)\le 2\Delta(\min\{p,q\}-1)$.
\end{lemma}
\begin{proof}
Observe that 
$A_{1,q}+A_{p,1} - A_{1,1}-A_{p,q} = \sum_{(i,j,v)\in \core{A}} v \ge \sum_{(i,j,v)\in \core{A}} 1 = \delta(A)$.
The bounded-difference property implies that $|A_{1,q}-A_{1,1}|\le \Delta(q-1)$ and $|A_{p,1}-A_{p,q}|\le \Delta(q-1)$.
Hence, $\delta(A) \le A_{1,q}+A_{p,1} - A_{1,1}-A_{p,q}\le 2\Delta(q-1)$. 
Symmetrically, $\delta(A)\le 2\Delta(p-1)$.
\end{proof}

Moreover, the bounded-difference property is preserved under the min-plus product.

\begin{lemma}\label{lem:bdproduct}
    Let $A\in \ZZ^{p\times q}$ and $B\in \ZZ^{q\times r}$ be $\Delta$-bounded-difference matrices.
    Then, $C\coloneqq A\otimes B$ is also a $\Delta$-bounded-difference matrix.
\end{lemma}
\begin{proof}
Consider two entries $C_{i,j}$ and $C_{i',j'}$, and pick $k,k'\in [1\dd q]$ such that $C_{i,j}=A_{i,k}+B_{k,j}$ and $C_{i',j'} = A_{i',k'}+B_{k',j'}$.
Observe that $C_{i,j}\le A_{i,k'}+B_{k',j} \le A_{i',k'}+\Delta|i-i'|+B_{k',j'}+\Delta|j-j'| = C_{i',j'}+\Delta(|i-i'|+|j-j'|)$.
Symmetrically, $C_{i',j'} \le  A_{i',k}+B_{k,j'} \le A_{i,k}+\Delta|i-i'|+B_{k,j}+\Delta|j-j'| = C_{i,j}+\Delta(|i-i'|+|j-j'|)$.
Hence, $|C_{i,j}-C_{i',j'}|\le \Delta(|i-i'|+|j-j'|)$.
In particular, this holds if $(i',j')=(i,j+1)$ and $(i',j')=(i+1,j)$, which means that $C$ is indeed $\Delta$-bounded-difference.
\end{proof}

While the bounded-difference property is helpful in the regime of small integer edit weights, if we do not assume edit weights to be small, matrix cores may be proportional to the sizes of the matrices, thus neglecting the advantage our algorithms have over SMAWK.
To address this issue, we define the notion of $k$-equivalence.

\begin{definition}
    Let $a$, $b$, and $k$ be non-negative real numbers.
    We call $a$ and $b$ \emph{$k$-equivalent} and write $a \meq{k} b$ if $a=b$ or $\min\{a,b\}> k$.
    Furthermore, let $A$ and $B$ be matrices of sizes $p \times q$ with non-negative entries.
    We call these two matrices \emph{$k$-equivalent} and write $A \meq{k} B$ if $A_{i, j} \meq{k} B_{i, j}$ for all $i \in [1 \dd p]$ and $j \in [1 \dd q]$.
\lipicsEnd
\end{definition}

As we compute $\wed_{\le k}$, if some values in distance matrices are larger than $k$, their exact values do not serve any purpose to us.
Hence, instead of considering the exact distance matrix $A$, it will suffice to maintain its ``capped'' version: any matrix $A'$ that is $k$-equivalent to $A$.
As it turns out, if we allow such a relaxation, it becomes possible to find such a matrix $A'$ with small core.

We observe that $k$-equivalence is an equivalence relation, and it is preserved under some basic operations with matrices.

\begin{observation} \label{c-equiv-preserv}
    The following properties of $k$-equivalence hold for any valid matrices $A, B, C, D$.
    \begin{itemize}
    \item $A \meq{k} A$.
    \item If $A \meq{k} B$, then $B \meq{k} A$.
    \item If $A \meq{k} B$ and $B \meq{k} C$, then $A \meq{k} C$.
    \item If $A \meq{k} B$, then $A^T \meq{k} B^T$.
    \item If $A \meq{k} B$, then $A[a \dd b)[c \dd d) \meq{k} B[a \dd b)[c \dd d)$ for any $a < b$ and $c < d$.
    \item If $A \meq{k} B$ and $C \meq{k} D$, then $(A \mid C) \meq{k} (B \mid D)$.
    \item If $A \meq{k} B$ and $C \meq{k} D$, then $A + C \meq{k} B + D$.\footnote{Note that in contrast to the other properties of $k$-equivalence, this one and the next one rely on the fact that the entries of $A$ and $B$ are required to be non-negative.}
    \item If $A \meq{k} B$ and $C \meq{k} D$, then $A \otimes C \meq{k} B \otimes D$.
\lipicsEnd
\end{itemize}
\end{observation}

We extend the list of tools for $\mds$ with ``matrix capping''.
That is, we develop an algorithm that given some Monge matrix $C$ with a large core, computes a matrix $k$-equivalent to $C$ that has a relatively small core.

\begin{restatable}{lemma}{lmreducematrix} \label{lm:reduce_matrix}
    There is an algorithm that, given some positive integer $k$ and the condensed representation of some Monge matrix  $C \in \mathbb{Z}_{\ge 0}^{p \times q}$, in time $\Oh((\delta(C) + N) \log N)$ builds the condensed representation of some Monge matrix $C' \in \mathbb{Z}_{\ge 0}^{p \times q}$ such that $\delta(C') \le \Oh(N \sqrt k)$ and $C' \meq{k} C$.
\lipicsEnd
\end{restatable}

The proof of this fact is rather technical, and thus is deferred to \cref{sec:monge-appendix-cap}.

We apply \cref{lm:reduce_matrix} to compute the capped product of Monge matrices.

\begin{lemma} \label{lm:matrix_mult_k}
    There is an algorithm that, given $\mds(A)$ and $\mds(B)$ for Monge matrices $A \in \mathbb{Z}_{\ge 0}^{p \times q}$ and $B \in \mathbb{Z}_{\ge 0}^{q \times r}$ and some positive integer $k$, in time $\Oh((N + \delta(A) + \delta(B)) \log N)$ builds $\mds(C')$, where $C' \in \mathbb{Z}_{\ge 0}^{p \times r}$ is some Monge matrix such that $\delta(C') \le \Oh(N \sqrt k)$ and $C' \meq{k} A \otimes B$.
\end{lemma}

\begin{proof}
    Applying \cref{lm:matrix_mult_w}, we get $\mds(C)$ for $C \coloneqq A \otimes B$ in time $\Oh((N + \delta(A) + \delta(B)) \log N)$.
    By applying \cref{lm:reduce_matrix}, we obtain $\mds(C')$ in time $\Oh((\delta(C) + N) \log N) = \Oh((N + \delta(A) + \delta(B)) \log N)$ due to \cref{lm:core-bound}.
\end{proof}

Note that all operations in \cref{lm:ds_tools} except for ``stitching'' are applicable for any input Monge matrices.
That is, given two Monge matrices $A$ and $B$, matrix $(A \mid B)$ is not necessarily Monge.
In practice, if we can ensure that $(A \mid B)$ is Monge, \cref{lm:ds_tools} is applicable.
However, if we are given some $A' \meq{k} A$ and $B' \meq{k} B$, it is not guaranteed that $(A' \mid B')$ is Monge even if $(A \mid B)$ is.
To address this issue, we develop the following lemma.
Note that it does not cover all such situations, but only a very specific one that we will encounter.

\newcommand{\maybe}{\lipicsEnd}
\begin{restatable}{lemma}{cappedmatrixstitching} \label{capped-matrix-stitching}
    Let $C \in \mathbb{Z}^{(n_1 + n_2) \times m}_{\ge 0}$ be a Monge matrix, and let $k$ be a positive integer.
    Let $N = n_1 + n_2 + m$.
    Let $A \coloneqq C[1 \dd n_1][1 \dd m]$ and $B \coloneqq C(n_1 \dd n_1 + n_2][1 \dd m]$.
    Let $\bj \in [1 \dd m]$ be such that $C_{i, j} \le C_{i - 1, j}$ for all $(i, j) \in (1 \dd n_1 + 1] \times [\bj \dd m]$ and $C_{i, j} \le C_{i, j + 1}$ for all $(i, j) \in [1 \dd n_1 + 1] \times [\bj \dd m)$.
    Let $B'$ be a Monge matrix that is $k$-equivalent to $B$ such that $B'_{1, j} = B_{1, j}$ for all $j \le \bj$.
    There is an algorithm that given $k$, $\mds(A)$, $\mds(B')$, and $\bj$, in time $\Oh((N + \delta(A) + \delta(B')) \log N)$ builds $\mds(C')$ for some Monge matrix $C' \in \mathbb{Z}^{(n_1 + n_2) \times m}_{\ge 0}$ such that $C' \meq{k} C$, $\delta(C') = \Oh(N + \delta(A) + \delta(B'))$, and $C'_{i, 1} = C_{i, 1}$ for $i \in [1 \dd n_1]$.
\maybe
\end{restatable}
\renewcommand{\maybe}{}

The proof of \cref{capped-matrix-stitching} is deferred to \cref{sec:monge-appendix-capped-stitching}.

\section{Hierarchical Alignment Graph Decompositions} \label{sec:hierarchical}

As discussed in \cref{overview:dynamic-unbounded}, already our static algorithm for bounded edit distance with small integer weights utilizes a dynamic subroutine for unbounded edit distance. 
This auxiliary result extends the algorithm of~\cite{CKM20} in two directions: we allow integer weights and more powerful updates such as copy-pastes.
The algorithm of~\cite{CKM20} maintains the edit distance $\ed(X,Y)$ using a hierarchical decomposition of the alignment graph $\AG(X,Y)$.
In \cref{subsec:hierarchical-composable}, we generalize the underlying notions to weighted edit distance, and we apply the results of \cref{sec:fast-monge-matrix-multiplication} to efficiently combine the information within the hierarchical decomposition.
In \cref{subsec:hierarch-align}, we build upon this procedure to maintain a \emph{hierarchical alignment data structure} $\boxds(X,Y)$, which employs weight-balanced straight-line programs (SLPs) in order to handle copy-pastes and other so-called \emph{mega-edits}.

\subsection{Composable Alignment Graph Representation}\label{subsec:hierarchical-composable}

In this section we develop a tool that, given distance information (boundary distance matrix) for some $X_L$ and $Y$ and some $X_R$ and $Y$, allows to compute distance information for $X \coloneqq X_L \cdot X_R$ and~$Y$.
This distance information provides oracle access to the distances between the vertices of the outer face of the alignment graph $\AGw(X, Y)$.
Note that some pairs of such vertices of $\AGw(X, Y)$ are not reachable from each other.
To circumvent this issue, we define a strongly connected extension of $\AGw(X, Y)$ that we call an \emph{augmented alignment graph}.

\begin{definition}[Augmented Alignment Graph]
    For any two strings $X, Y \in \Sigma^*$ and a weight function $w: \overline{\Sigma}^2 \to [0,W]$, define a directed graph $\overline{\AG}^w(X, Y)$ obtained from the alignment graph $\AG^w(X, Y)$ by adding, for every edge of $\AG^w(X, Y)$, a back edge of weight $W + 1$.
\lipicsEnd
\end{definition}

Note that, in contrast to $\AGw(X, Y)$, any vertex of $\oAGw(X, Y)$ is reachable from any other vertex of $\oAGw(X, Y)$.
We further show that $\oAGw(X, Y)$ preserves distances from $\AGw(X, Y)$ between reachable pairs of vertices, and for all other relevant pairs of vertices, shortest paths between them behave predictably.

\renewcommand{\maybe}{\lipicsEnd}
\begin{restatable}{lemma}{gridgraphextension} \label{grid-graph-extension}
    Consider strings $X, Y \in \Sigma^{*}$ and a normalized weight function $w : \overline{\Sigma}^2 \to [0,W]$.
    Any two vertices $(x, y)$ and $(x', y')$ of the graph $\oAGw(X,Y)$ satisfy the following properties:
\begin{description}
    \item[Monotonicity.] Every shortest path $(x, y) = (x_0, y_0) \to (x_1, y_1) \to \dots \to (x_m, y_m) = (x', y')$ from $(x, y)$ to $(x', y')$ in $\overline{\AG}^w(X, Y)$ is (non-strictly) monotone in both coordinates $(x_i)_{i \in [0 \dd m]}$ and $(y_i)_{i \in [0 \dd m]}$.
    \item[Distance preservation.] If $x \le x'$ and $y \le y'$, then
        \[\dist_{\oAGw(X, Y)}((x, y), (x', y')) = \dist_{\AGw(X, Y)}((x, y), (x', y')).\]
    \item[Path irrelevance.] If ($x \le x'$ and $y \ge y'$) or ($x \ge x'$ and $y \le y'$), every path from $(x, y)$ to $(x', y')$ in $\oAGw(X, Y)$, that is monotone in both coordinates, is a shortest path between these two vertices.
    \maybe
\end{description}     
\end{restatable}
\renewcommand{\maybe}{}

The proof of this lemma is deferred to \cref{subsec:hierarchical-appendix1}.

\smallskip
We will talk of \emph{breakpoint representations} of paths in $\oAGw(X, Y)$ that are monotone in both coordinates.
If such a path starts in $(x, y)$ and ends in $(x', y')$ with $x \le x'$ and $y \le y'$, then the breakpoint representation of such a path is identical to the breakpoint representation of the corresponding alignment.
Otherwise, if $x > x'$ or $y > y'$, we call the breakpoint representation of such a path a sequence of all its vertices.
Note that it is a consistent definition as all edges in such a path have nonzero weights.

\medskip

We now show that similarly to the proof of \cref{lm:baseline-wed}, short paths in $\oAGw(X, Y)$ cannot deviate too much from the diagonal they start on.

\begin{lemma} \label{lm:paths_dont_deviate_too_much}
    Let $X, Y \in \Sigma^*$ be two strings, $k$ be an integer, and $w : \Esigma^2 \to \RR_{\ge 0}$ be a normalized weight function.
    Let $(x, y)$ and $(x', y')$ be two vertices of $\oAGw(X, Y)$.
    Let $P$ be a path from $(x, y)$ to $(x', y')$ in $\oAGw(X, Y)$ of weight at most $k$.
    We have that all nodes $(x^*, y^*) \in P$ satisfy $|(x - y) - (x^* - y^*)| \le k$ and $|(x' - y') - (x^* - y^*)| \le k$.
\end{lemma}

\begin{proof}
    Consider the prefix of the path $P$ from $(x, y)$ to $(x^*, y^*)$.
    It has weight at most $k$ and thus contains at most $k$ non-zero edges.
    Note that each edge of weight zero is a diagonal edge.
    Hence, while going from $(x, y)$ to $(x^*, y^*)$ along $P$, the difference $\hx - \hy$ between the coordinates of the current vertex $(\hx, \hy)$ change at most $k$ times.
    Furthermore, one edge changes this difference by at most one.
    Therefore, we obtain $|(x -  y) - (x^* - y^*)| \le k$.
    Analogously, we get $|(x' - y') - (x^* - y^*)| \le k$.
\end{proof}

We now precisely define what the boundary distance matrix of $X$ and $Y$ is.

\begin{definition}[Rectangle Perimeter]
    For a rectangle $\fragmentcc{a}{b} \times \fragmentcc{c}{d}$, we define its \emph{perimeter} as $\perim(\fragmentcc{a}{b} \times \fragmentcc{c}{d}) \coloneqq (b - a) + (d - c) + 1$.

    Furthermore, for two strings $X$ and $Y$, we define the perimeter of their product as $\perim(X \times Y) \coloneqq \perim([0 \dd |X|] \times [0 \dd |Y|]) = |X| + |Y| + 1$.
\lipicsEnd
\end{definition}

\begin{definition}[Input and Output Vertices of an Alignment Graph]
    Let $X, Y \in \Sigma^*$ be two strings and $w : \Esigma^2 \to [0,W]$ be a weight function.
    Furthermore, let $R = [a \dd b] \times [c \dd d]$ for $0 \le a \le b \le |X|$ and $0 \le c \le d \le |Y|$ be some rectangle in $\oAGw(X, Y)$.
    
    Let \emph{the sequence of input vertices} $\inp(R)$ of $R$ be a sequence $(p_i)_i$ of vertices of $\oAGw(X, Y)$ of length $\perim(R)$, where $p = ((a, d), (a, d - 1), \ldots, (a, c), (a + 1, c), \ldots, (b, c))$; see \cref{fig:alignment-graph}.
    Let the \emph{left input vertices} of $R$ be the prefix of vertices of $\inp(R)$ with the first coordinate equal to $a$, and let the \emph{top input vertices} of $R$ be the suffix of vertices of $\inp(R)$ with the second coordinate equal to $c$.
    Note that the left and the top input vertices of $R$ intersect by one vertex $(a, c)$.

    Analogously, let \emph{the sequence of output vertices} $\outp(R)$ of $R$ be a sequence $(q_i)_i$ of vertices of $\oAGw(X, Y)$ of length $\perim(R)$, where $q = ((a, d), (a + 1, d), \ldots, (b, d), (b, d - 1), \ldots, (b, c))$.
    Let the \emph{bottom output vertices} of $R$ be the prefix of vertices of $\outp(R)$ with the second coordinate equal to $d$, and let the \emph{right output vertices} of $R$ be the suffix of vertices of $\outp(R)$ with the first coordinate equal to $b$.
    Note that the bottom and the right output vertices of $R$ intersect by one vertex $(b, d)$.

    Furthermore, we define $\inp(X \times Y) \coloneqq \inp([0 \dd |X|] \times [0 \dd |Y|])$ and $\outp(X \times Y) \coloneqq \outp([0 \dd |X|] \times [0 \dd |Y|])$.
    These are called the input and the output vertices of the augmented alignment graph $\oAGw(X, Y)$.
\lipicsEnd
\end{definition}

\begin{definition}[Boundary Distance Matrix]
    Let $X, Y \in \Sigma^*$ be two strings and $w : \Esigma^2 \to [0,W]$ be a weight function.
    \emph{The boundary distance matrix} $\BMw(X, Y)$ of these two strings is a matrix $M$ of size $\perim(X \times Y) \times \perim(X \times Y)$ defined as follows: $M_{i, j} = \dist_{\oAGw(X, Y)}(p_i, q_j)$, where $p \coloneqq \inp(X \times Y)$ and $q \coloneqq \outp(X \times Y)$.
\lipicsEnd
\end{definition}

Note that $\BMw(X, Y)$ is Monge due to \cref{fct:monge}.
Furthermore, if $X[a \dd b)$ and $Y[c \dd d)$ are two fragments of $X$ and $Y$, then $\BMw(X[a \dd b), Y[c \dd d))$ is a distance matrix from the input vertices of $[a \dd b] \times [c \dd d]$ in $\oAGw(X, Y)$ to the output vertices of $[a \dd b] \times [c \dd d]$ as due to the monotonicity property of \cref{grid-graph-extension}, no such path exits $[a \dd b] \times [c \dd d]$.

Since there are edges (of cost at most $W+1$) between subsequent input and output vertices of $\oAGw(X,Y)$,
the boundary distance matrix satisfies the bounded-difference property (\cref{def:bd}).

\begin{observation}\label{obs:bmw_is_bd}
Let $X, Y \in \Sigma^*$ be two strings and let $w : \overline{\Sigma}^2 \to [0 \dd W]$ be a weight function.
The matrix $\BMw(X, Y)$ is $(W+1)$-bounded-difference, and thus its core is of size $\Oh(W\cdot(|X|+|Y|))$ due to \cref{lem:bdtocore}.
\lipicsEnd
\end{observation}

We are now ready to describe the composition procedure of $\BMw$.
This procedure has two variants: one that computes all distances exactly, and a relaxed version that computes $k$-equivalent values.

\renewcommand{\maybe}{\lipicsEnd}
\begin{restatable}{lemma}{boundarydistancematrixcombination} \label{boundary-distance-matrix-combination}
    There is an algorithm that given strings $X, Y \in \Sigma^*$, oracle access to a weight function $w : \Esigma^2 \to [0 \dd W]$, and $\mds(A)$ and $\mds(B)$, where $A = \BMw(X[a \dd b), Y[c \dd d))$ and $B = \BMw(X[a\dd b), Y[d \dd e))$ for some $0 \le a < b \le |X|$ and $0 \le c < d < e \le |Y|$, in time $\Oh((m + \delta(A) + \delta(B)) \log m)$ for $m = \perim([a \dd b] \times [c \dd e])$ builds $\mds(C)$, where $C = \BMw(X[a\dd b), Y[c\dd e))$.

    Furthermore, there is an algorithm that, given a positive integer $k$, $\mds(A')$ and $\mds(B')$ for $A' \meq{k} \BMw(X[a \dd b), Y[c \dd d))$ and $B' \meq{k} \BMw(X[a \dd b), Y[d \dd e))$, in time $\Oh((m + \delta(A') + \delta(B')) \log m)$ builds $\mds(C')$ for some matrix $C'$ such that $C' \meq{k} \BMw(X[a\dd b), Y[c\dd e))$ and $\delta(C') = \Oh(m \sqrt k)$.

    The same statements hold for two rectangles adjacent horizontally, not vertically.
    \maybe
\end{restatable}
\renewcommand{\maybe}{}

On the high level, to obtain $C$, one needs to slightly extend the matrices $A$ and $B$ and then $(\min, +)$-multiply them.
As the complete proof of this lemma is rather technical, it is deferred to \cref{subsec:hierarchical-appendix2}.

\subsection{Hierarchical Alignment Data Structure: Extension of the Unbounded Dynamic Algorithm to Mega-Edits}\label{subsec:hierarch-align}

While~\cite{CKM20} supports only edits (insertions, deletions, and substitutions) as updates, as discussed in \cref{overview:dynamic-unbounded}, for our purposes, we will require a more powerful dynamic data structure that supports the following operations we call mega-edits.

\begin{definition}[Mega-Edit]\label{def:mega-edit}
    Let $X \in \Sigma^*$ be a string.
    A \emph{mega-edit} is one of the following transformations applied to $X$:
    \begin{description}
        \item[Character Insertion:] Insert an arbitrary character $c \in \Sigma$ at some position $i \in \fragmentcc{0}{|X|}$ in $X$.
        \item[Character Deletion:] Delete some character $X[i]$ for $i \in \fragmentco{0}{|X|}$ from $X$.
        \item[Character Substitution:] Replace $X[i]$ in $X$ for some $i \in \fragmentco{0}{|X|}$ with some other character $c \in \Sigma$.
        \item[Substring Removal:] Remove fragment $X\fragmentco{\ell}{r}$ from $X$ for some $\ell, r \in \fragmentcc{0}{|X|}$ with $\ell < r$.
        \item[Copy-Exponentiate-Paste:] For some $\ell, r, p \in \fragmentcc{0}{|X|}$ with $\ell < r$, and an integer $s \ge 1$, insert $X\fragmentco{\ell}{r}^{\infty}\fragmentco{0}{s}$ at position $p$ of $X$.
    \lipicsEnd
    \end{description}
\end{definition}

Furthermore, the data structure we need should be fully persistent.
To achieve this, we follow a design pattern known from functional programming: each instance of our \emph{hierarchical alignment data structure} is an immutable object, and each update creates a new instance that shares parts of the internal state with the original instance.
We describe two variants of the data structure: one that computes all distances exactly, and a relaxed version that computes $k$-equivalent values.

\renewcommand{\maybe}{\lipicsEnd}
\begin{restatable}{theorem}{prphadstwo}\label{prp:hadstwo}
    Consider non-empty strings  $X, Y \in \Sigma^+$ and a weight function $w : \Esigma^2 \to \fragmentcc{0}{W}$.
    A \emph{hierarchical alignment data structure} $\boxds(X,Y)$ can be constructed in $\Oh(W\cdot n^2 \log n)$ time (where $n = |X| + |Y|$) given $X$, $Y$, $W$, and $\Oh(1)$-time oracle to $w$,
    and supports the following operations:
    \begin{description}
        \item[Matrix Retrieval:] Retrieve an immutable reference to $\mds(\BMw(X, Y))$ in constant time.
        \item[Alignment Retrieval:] Given an input vertex $p$ and an output vertex $q$ of $\oAGw(X, Y)$, construct the breakpoint representation of a shortest path from $p$ to $q$ in $\oAGw(X, Y)$.
            This operation takes $\Oh((d+1)\log^2 n)$ time, where $d=\dist_{\oAGw(X,Y)}(p,q)$ and $n = |X| + |Y|$.
        \item[Mega-Edit in $X$:] Given a mega-edit transforming $X$ into a nonempty string $X'$, construct a hierarchical alignment data structure $\boxds(X', Y)$ in time $\Oh(W n \log^2 n)$, where $n = |X| + |X'| + |Y|$.
        \item[Mega-Edit in $Y$:] Given a mega-edit transforming $Y$ into a nonempty string $Y'$, construct a hierarchical alignment data structure $\boxds(X, Y')$ in time $\Oh(W n \log^2 n)$, where $n = |X| + |Y| + |Y'|$.
    \end{description}
    Additionally, for every positive integer $k\in \Zp$, a \emph{relaxed hierarchical alignment data structure} $\boxdsk(X, Y)$ can be built in time $\Oh(n^2 \sqrt k \log n)$ and supports mega-edits in $\Oh(n \sqrt k \log^2 n)$ time, but matrix retrieval queries only return a reference to $\mds(M)$ for some Monge matrix $M$ such that $M \meq{k} \BMw(X, Y)$ and $\coresiz(M)=\Oh((|X| + |Y|)\sqrt{k})$, and alignment retrieval queries fail (report an error) whenever $d>k$.\maybe
\end{restatable}
\renewcommand{\maybe}{}

As discussed in \cref{sec:overview}, our main innovation compared to the algorithm of \cite{CKM20} is that we employ weight-balanced straight-line programs to keep track of isomorphic parts of the (augmented) alignment graph $\oAGw(X,Y)$. 
Thus, before discussing the implementation of the hierarchical alignment data structure $\boxds(X,Y)$, we introduce the relevant terminology and notation.

\subsubsection{Balanced Straight Line Programs}
For a context-free grammar $\G$, we denote by $\N_\G$ and $\Sigma_\G$ the set of non-terminals and the set of terminals, respectively. The set of \emph{symbols} is $\S_\G:=\Sigma_\G\cup \N_\G$.
A \emph{straight-line program} (SLP) is a context-free grammar $\G$ such that:
\begin{itemize}
\item each non-terminal $A\in \N_\G$ has a unique production $A\to \rhs_\G(A)$, where $\rhs_\G(A)=BC$ consists of two symbols $B,C\in \S_\G$,
\item the set $\S_\G$ admits a partial order $\prec$ such that $B\prec A$ and $C\prec A$ if $\rhs_\G(A)=BC$.
\end{itemize}
See \cref{fig:SLP} for an example.
The \emph{expansion} function $\exp_\G: \S_G \to \Sigma_\G^*$ is defined as follows:
\[\exp_\G(A) = \begin{cases}
  A & \text{if $A\in \Sigma_\G$},\\
  \exp_\G(B)\exp_\G(C) & \text{if $A\in \N_\G$ with $\rhs_\G(A)=BC$.}
\end{cases}\]
Moreover, $\exp_\G$ is lifted to $\exp_\G : \S_\G^*\to \Sigma_\G^*$ with $\exp_\G(A_1\cdots A_a)=\exp_\G(A_1)\cdots \exp_\G(A_a)$ for $A_1\cdots A_a\in \S_\G^*$. When the grammar $\G$ is clear from context, we omit the subscript.

For notational convenience, for symbols $A$ and $B$ of some fixed grammar $\slp$, we write $\oAGw(A, B)$  to denote $\oAGw(\exp(A), \exp(B))$ and $\BMw(A, B)$ to denote $\BMw(\exp(A), \exp(B))$.

\begin{figure}[t]
    \begin{center}
    \begin{tikzpicture}[scale=0.55,y=-1cm]
        \input{SLP.tex}
    \end{tikzpicture}
    \end{center}

    \caption{An example of a straight-line program. Here $\exp(A) = \mathtt{abccdcdabccd}$.}
    \label{fig:SLP}
\end{figure}

The \emph{parse tree} $\Tr(A)$ of a symbol $A\in \S_\G$ is a rooted ordered tree with each node $\nu$ associated to a symbol $\symb(\nu)\in \S_\G$. 
The root of $\Tr(A)$ is a node $\rho$ with $\symb(\rho)=A$.
If $A \in \Sigma_\G$, then $\rho$ has no children. 
If $A\in \N_\G$ and $\rhs(A)=BC$, then $\rho$ has two children, and the subtrees rooted in these children are (copies of) $\Tr(B)$ and $\Tr(C)$, respectively. The root of $\Tr(B)$ is called the left child of the root of $\Tr(A)$, and the root of $\Tr(C)$ is called the right child of the root of $\Tr(A)$. Furthermore, we say that $B$ is the left child of $A$, and $C$ is the right child of $A$.
Furthermore, let $\dep(A) \coloneqq \set{\symb(\nu) \mid \nu \in \Tr(A)}$ be the set of \emph{descendants} of $A$; note that $\dep(A)=\{A\}$ if $A\in \Sigma_\G$ and $\dep(A)=\{A\}\cup \dep(B)\cup \dep(C)$ if $A\in \N_G$ and $\rhs_\G(A)=BC$.

The weight of a symbol $A\in \S_\G$ is defined as $\len(A):= |\exp(A)|$; we assume that all algorithms store this weight along with the symbol $A$.
We say that an SLP $\G$ is \emph{weight-balanced} if every nonterminal $A\in \N_\G$ with $\rhs(A)=BC$ satisfies $\frac{1}{3} \le \len(B) / \len(C) \le 3$.\footnote{Weight-balanced SLPs from \cite{CLLPPSS05} allow for $\alpha / (1 - \alpha) \le \len(B) / \len(C) \le (1 - \alpha) / \alpha$ for any balancedness parameter $\alpha \le 1 - \sqrt2 / 2 \approx 0.29$. We use $\alpha = 0.25$.}

\SetKwFunction{Merge}{Merge}
\SetKwFunction{Split}{Substring}
\SetKwFunction{Power}{Power}

\begin{theorem}[{see \cite[Section VII.E]{CLLPPSS05} and \cite[Lemma 8]{Gaw11}}]\label{thm:slp}
A weight-balanced SLP $\slp$ can be extended (while remaining weight-balanced) using the following operations:
\begin{description}
    \item[$\Merge(B,C)$:] Given symbols $B,C\in \S_\slp$, in $\Oh(1 + |\log (\len(B) / \len(C))|)$ time returns a symbol $A$ such that $\exp(A)=\exp(B)\cdot \exp(C)$.
    \item[$\Split(A,\ell, r)$:] Given a symbol $A\in \S_\slp$ and integers $\ell, r\in\fragmentcc{0}{\len(A)}$ such that $\ell < r$, in $\Oh(\log \len(A))$ time returns a symbol $B$ such that $\exp(B)=\exp(A)\fragmentco{\ell}{r}$.
    \item[$\Power(A,\ell)$:] Given a symbol $A\in \S_\slp$ and an integer $\ell\ge 1$, in $\Oh(1 + \log \len(A)+\log \ell)$ time returns a symbol $B$ such that $\exp(B)=\exp(A)^\infty\fragmentco{0}{\ell}$.
\end{description}
The algorithms implementing these operations represent each symbol along with its weight $\len(\cdot)$ and, for non-terminals, the production $\rhs(\cdot)$, and make no other assumptions about the implementation of $\slp$.
In particular, every symbol present in the parse trees of the output symbols is either newly added to $\slp$ or present in the parse trees of the input symbols. The number of newly added symbols is bounded by the time complexity of each operation.
\end{theorem}
\begin{proof}
    The implementations of $\Merge(B,C)$ and $\Split(A,\ell, r)$ are provided in \cite[Section VII.E]{CLLPPSS05}.
    The original analysis focuses on the number of symbols added to $\G$, but several subsequent works, including \cite{Gaw11}, noted that the running time can be bounded in the same way.
    Our implementation of $\Power(A,\ell)$ is based on the proof of \cite[Lemma 8]{Gaw11}. 
    Since the original description is stated in the context of a particular application (rather than in an abstract form), we repeat the simple algorithm and its analysis below.

    Denote $t \coloneqq \max\{0,\lceil \log_2 \frac{\ell}{\len(A)}\rceil\}$.
    We create a sequence of symbols $B_0,B_1, B_2, \ldots, B_t$ such that $B_0 = A$ and $\rhs(B_i) = B_{i - 1} B_{i - 1}$ for $i\in \fragmentcc{1}{t}$.
    If $\len(B_t) = \ell$, we return $B_t$.
    Otherwise, we retrieve $B=\Split(B_t, 0, \ell)$ and return $B$ as the answer.

    The children of $B_i$ are perfectly balanced, so $\slp$ remains a weight-balanced SLP.
    At the same time $\exp(B_i) = \exp(B_{i - 1})^2 = \exp(A)^{2^i}$.
    The correctness follows from the facts that $\exp(B_t)$ is a prefix of $\exp(A)^{\infty}$ and $\len(B_t) \ge \ell$.
    The time complexity is $\Oh(t+1)$ for the creation of $B_0,\ldots,B_t$ and $\Oh(\log \len(B_t))$ for the final $\Split$ operation.
    Since $\len(B_t)\le 2\max\{\len(A),\ell\}$ and $t\le \log_2 \len(B_t)$, this is $\Oh(1 + \log\len(A)+\log \ell)$.
\end{proof}

\subsubsection{Hierarchical Alignment Data Structure Implementation}

The hierarchical alignment data structure $\boxds(X,Y)$ is arranged recursively.
We build a weight-balanced SLP $\slp$ with symbols $A, B \in \S_{\slp}$ such that $\exp(A) = X$ and $\exp(B) = Y$, and the recursive structure of the SLP guides the recursive structure of $\boxds(X,Y)$.
We formalize this using an auxiliary recursive data structure $\boxdsslp(A,B)$.

\begin{definition}\label{def:boxdsslp}
    Consider a weight-balanced SLP $\slp$ and a weight function $w : \Esigma_\slp^2\to [0\dd W]$. 
    For two symbols $A,B\in \S_{\slp}$, the data structure $\boxdsslp(A, B)$ stores (an immutable reference to) $\mds(\BMw(A, B))$.
    Furthermore, unless $\len(A)=\len(B)=1$, it recursively stores (immutable references to) two smaller instances of $\boxdsslp$:
    If $\len(A) \ge \len(B)$, it stores $\boxdsslp(A_L, B)$ and $\boxdsslp(A_R, B)$, where $\rhs_\slp(A) = A_L A_R$.
    Otherwise, i.e., if $\len(A) < \len(B)$, it shores $\boxdsslp(A, B_L)$ and $\boxdsslp(A, B_R)$, where $\rhs_\slp(B) = B_L B_R$.

    Furthermore, for every $k\in \Zp$, we define the relaxed version $\boxdskslp(A, B)$ of the data structure that, instead of $\mds(\BMw(A, B))$, stores $\mds(M)$ for some matrix $M$ satisfying $M \meq{k} \BMw(A, B)$ and $\coresiz(M) = \Oh((\len(A) + \len(B)) \sqrt k)$; the recursive references are to smaller instances of $\boxdskslp$.
    \lipicsEnd
\end{definition}

Before proving \cref{prp:hadstwo}, we show  two helper lemmas. 
The first of them implements the alignment retrieval operation of $\boxds(X,Y)$.
\begin{restatable}{lemma}{lmboxdsslpquery}\label{lm:boxdsslp-query}
    Let $A$ and $B$ be two symbols of a weight-balanced SLP $\slp$.
    There is an algorithm that, given $\boxdsslp(A, B)$ and an input vertex $p$ and an output vertex $q$ of $\oAGw(A, B)$, computes the breakpoint representation of a shortest path from $p$ to $q$ in $\oAGw(A, B)$ in $\Oh((d + 1) \log^2 n)$ time, where $d=\dist_{\oAGw(A,B)}(p,q)$ and $n = \len(A) + \len(B)$.

    Furthermore, there is an algorithm that, given $\boxdskslp(A, B)$ for some $k \ge 1$ instead of $\boxdsslp(A, B)$, in $\Oh((d + 1) \log^2 n)$ time either computes the breakpoint representation of a shortest path from $p$ to $q$ in $\oAGw(A, B)$ or reports that $d > k$.
\end{restatable}

\begin{proof}[Proof Sketch]
    We compute the breakpoint representation of a shortest path by tracking it down the hierarchical decomposition of $\oAGw(A, B)$ provided by $\boxdsslp(A,B)$.
    If $d=0$, the path simply follows a diagonal and its breakpoint representation is trivial.
    Similarly, it is easy to construct the breakpoint representation if $\len(A)=\len(B)=1$.
    Otherwise, the recursive structure of $\boxdsslp(A,B)$ lets us partition $\oAGw(A,B)$ into two subgraphs, which we interpret as two smaller augmented alignment graphs.
    If $p$ and $q$ are both located in the same subgraph, then we recurse on the instance of $\boxdsslp$ describing this subgraph.
    Otherwise, due to \cref{lm:paths_dont_deviate_too_much}, there are $\Oh(d)$ potential points where the path can cross the vertical or horizontal line separating the subgraphs.
    We find an optimal splitting point using random access to the boundary distance matrices of these subgraphs, and then we recurse on both subgraphs.
    The total time complexity is $\Oh((d+1)\log^2 n)$ because we spend $\Oh(d \log n)$ time per recursive call, the cost of the path is split among the (at most) two recursive calls, and the depth of the recursion is logarithmic.
    See \cref{subsec:hierarchical-appendix3} for the formal proof of \cref{lm:boxdsslp-query}.
\end{proof}

Recall that $\dep(A)$ denotes the set of descendants of a symbol $A$ of an SLP, i.e., the set of symbols that occur in the parse tree of $A$.
The following lemma constitutes the central component of our implementation of the mega-edit operations of $\boxds(X,Y)$.
\begin{restatable}{lemma}{lmboxdsslpupdate}\label{lm:boxdsslp-update}
    Let $A, A', B$, and $B'$ be symbols of a weight-balanced SLP $\slp$.
    Given $\boxdsslp(A, B)$, one can build $\boxdsslp(A', B')$ in $\Oh((u+\log n) W n \log n)$ time, where $u = |\dep(A') \setminus \dep(A)| + |\dep(B') \setminus \dep(B)|$ and $n = \len(A) + \len(B) + \len(A') + \len(B')$.
    Furthermore, given $\boxdskslp(A, B)$ for some $k\in \Zp$, one can build $\boxdskslp(A', B')$ in $\Oh((u + \log n) n \sqrt k \log n)$ time.
\end{restatable}

\begin{proof}[Proof Sketch]
    We employ a natural recursive algorithm to build $\boxdsslp(A', B')$.
    If $\len(A')=\len(B')=1$, then $\BMw(A', B')$ (and thus $\boxdsslp(A',B')$) can be constructed in constant time using oracle access to~$w$.
    Otherwise, we recurse and then build $\mds(\BMw(A', B'))$ using \cref{boundary-distance-matrix-combination}; this final step takes $\Oh(W n \log n)$ time due to \cref{obs:bmw_is_bd}.
    Furthermore, for the current recursive call involving $(C, D)$, if the data structure $\boxdsslp(C, D)$ is a part of the recursive structure of $\boxdsslp(A, B)$, we make sure to just reuse it instead of building it from scratch.
    With careful analysis, we show that the algorithm takes $\Oh(W n \log n)$ time per new SLP symbol and $\Oh((u + \log n) W n \log n)$ time in total.

    For $\boxdskslp$, we replace the first algorithm of \cref{boundary-distance-matrix-combination} with the second algorithm of the same lemma, which automatically gives us a bound on the core sizes of all matrices we obtain.
    See \cref{subsec:hierarchical-appendix3} for the formal proof of \cref{lm:boxdsslp-update}
\end{proof}

We now derive \cref{prp:hadstwo} from \cref{lm:boxdsslp-query,lm:boxdsslp-update,thm:slp}.

\prphadstwo*
\begin{proof}
    The hierarchical alignment data structure $\boxds(X, Y)$ is implemented as $\boxdsslp(A,B)$ for two symbols $A$ and $B$ satisfying $\exp(A) = X$ and $\exp(B) = Y$.

    To build $\boxds(X, Y)$, we use \cref{lm:boxdsslp-update}.
    We construct $A$ and $B$ with $\exp(A) = X$ and $\exp(B) = Y$ by building perfectly balanced parse trees on top of $X$ and $Y$ and using fresh symbols for all internal nodes. This step takes $\Oh(n)$ time and results in a grammar of $\Oh(n)$ different symbols.
    In order to streamline the presentation, consider an arbitrary terminal $\$$.
    We build $\mds(\BMw(\$, \$))$ and thus $\boxdsslp(\$, \$)$ trivially in constant time.
    By using \cref{lm:boxdsslp-update}, we can transform $\boxdsslp(\$, \$)$ into $\boxdsslp(A, B)$ in $\Oh(W n^2 \log n)$ time (or $\Oh(n^2 \sqrt k \log n)$ for the case of $\boxdskslp$).

    The matrix retrieval operation is implemented trivially as we store $\mds(\BMw(A, B))$.
    For the alignment retrieval operation, we use \cref{lm:boxdsslp-query}.

    It remains to describe how to perform a mega-edit in $X$; a mega-edit in $Y$ can be performed analogously.
    Say, $\boxdsslp(A, B)$ underlies the hierarchical alignment data structure for $X$ and $Y$.
    That is, $\exp(A) = X$, $\exp(B) = Y$, and we have access to $\boxdsslp(A, B)$.
    Consider the mega-edit that is applied to $X$ to obtain $X'$.
    Suppose that the mega-edit is a copy-exponentiate-paste operation; the remaining cases are processed analogously.
    Furthermore, assume that the parameters from the definition of a copy-exponentiate-paste operation satisfy $\ell > 0$, $r < |X|$, and $p \in \fragmentoo{0}{|X|}$.
    In this case, a symbol of $\slp$ representing $X'$ can be obtained as
    \[A' \coloneqq \Merge(\Split(A,0,p), \Merge(\Power(\Split(A,\ell,r),s),\Split(A,p,\len(A)))).\]
    In the special cases of $\ell = 0$, $r = |X|$, $p = 0$, and $p = |X|$, the formulas are analogous but simplified.
    We then apply \cref{lm:boxdsslp-update} to build $\boxdsslp(A', B)$ (which will underlie $\boxds(X', Y)$) from $\boxdsslp(A, B)$.

    In all cases, we build $A'$ from $A$ via constantly many applications of $\Split$, $\Merge$, and $\Power$, each of which takes $\Oh(\log n)$ time and creates $\Oh(\log n)$ new symbols that were not present in $\dep(A)$.
    Therefore, we have $|\dep(A') \setminus \dep(A)| = \Oh(\log n)$, and thus the application of \cref{lm:boxdsslp-update} takes $\Oh(W n \log^2 n)$ time (or $\Oh(n \sqrt k \log^2 n)$ for the case of $\boxdskslp$).
\end{proof}

\section{Self-Edit Distance}\label{sec:sed}

As discussed in \cref{sec:overview}, one crucial ingredient of our algorithms is the notions of self-alignments and self-edit distance introduced by Cassis, Kociumaka, and Wellnitz in \cite{CKW23}.
We list some existing facts about self-edit distance and introduce our main combinatorial insights in \cref{lem:intersect} and \cref{cor:intersect}.

\begin{definition}[{\cite{CKW23}}] \label{def:self-ed}
    Say that an alignment $\cA : X \onto X$ is a \emph{self-alignment}
    if $\cA$ does not align any character $X\position{x}$ to itself.
    We define the \emph{self-edit distance} of $X$ as $\sed(X) \coloneqq \min_\cA \ed_{\cA}(X, X)$,
    where the minimization ranges over all self-alignments $\cA : X \onto X$. In words,
    $\sed(X)$ is the minimum (unweighted) cost of a self-alignment.
\lipicsEnd
\end{definition}

We can interpret a self-alignment as a
$(0, 0) \leadsto (|X|, |X|)$ path in the alignment graph $\AG(X, X)$
that does not contain any edges of the main diagonal.

We list some fundamental properties of $\sed$.

\begin{lemma}[Properties of $\sed$, extension of {\cite[Lemma 4.2]{CKW23}}]\label{fct:selfed-properties}
    Let $X \in \Sigma^*$ denote a string. Then, all of the following hold:
    \begin{description}
        \item[Monotonicity.] For any $\ell' \le \ell \le r \le r'\in [0 \dd |X|]$, we have
            \(
                \sed(X\fragmentco{\ell}{r}) \leq \sed(X\fragmentco{\ell'}{r'}).
            \)
        \item[Sub-additivity.] For any $m \in [0 \dd |X|]$, we have
            \(
                \sed(X) \leq \sed(X\fragmentco{0}{m}) + \sed(X\fragmentco{m}{|X|}).
            \)
        \item[Triangle inequality.] For any $Y\in \Sigma^*$, we have
            \(
                \sed(Y) \le \sed(X)+2\cdot\ed(X,Y).
            \)
        \item[Continuity.] For any $i \in \fragmentoo{0}{|X|}$, we have
            \[
                \sed(X\fragmentco{0}{i}) \le \sed(X\fragmentco{0}{i + 1}) \le \sed(X\fragmentco{0}{i}) + 1.
            \]
    \end{description}
\end{lemma}

\begin{proof}
    The first three properties are proven in {\cite[Lemma 4.2]{CKW23}}.
    Hence, it remains to prove continuity.
    $\sed(X\fragmentco{0}{i}) \le \sed(X\fragmentco{0}{i + 1})$ follows from monotonicity of $\sed$.
    We now prove $\sed(X\fragmentco{0}{i + 1}) \le \sed(X\fragmentco{0}{i}) + 1$.

    Consider an optimal self-alignment $(x_i, y_i)_{i=0}^t \eqqcolon \cA : X\fragmentco{0}{i} \onto X\fragmentco{0}{i}$.
    We have $(x_t, y_t) = (i, i)$.
    Note that $(x_{t - 1}, y_{t - 1}) \neq (i - 1, i - 1)$ as $\cA$ is a self-alignment.
    Therefore, $(x_{t - 1}, y_{t - 1}) = (i - 1, i)$ or $(x_{t - 1}, y_{t - 1}) = (i, i - 1)$.
    Without loss of generality, assume that $(x_{t - 1}, y_{t - 1}) = (i, i - 1)$.
    Consider a self-alignment $(x'_i, y'_i)_{i=0}^{t + 1} \eqqcolon \cB : X\fragmentco{0}{i + 1} \onto X\fragmentco{0}{i + 1}$ where $(x'_i, y'_i) = (x_i, y_i)$ for $i \in [0 \dd t - 1]$, $(x'_t, y'_t) = (i + 1, i)$, and $(x'_{t + 1}, y'_{t + 1}) = (i + 1, i + 1)$.
    We obtain
    \begin{align*}
        \sed(X\fragmentco{0}{i + 1}) &\le \ed_{\cB}(X\fragmentco{0}{i + 1}, X\fragmentco{0}{i + 1})\\
                                     &\le \ed_{(x_i, y_i)_{i = 0}^{t - 1}}(X\fragmentco{0}{i}, X\fragmentco{0}{i-1}) + 2\\
                                     &= (\ed_{\cA}(X\fragmentco{0}{i}, X\fragmentco{0}{i}) - 1) + 2\\
                                     &= \sed(X\fragmentco{0}{i}) + 1.
\end{align*}
Note that this property does not hold for $i = 0$ as $\sed(X\fragmentco{0}{0}) = 0$ and $\sed(X\fragmentco{0}{1}) = 2$.
\end{proof}

Similarly to \cref{lm:k2-ed}, there is a \modelname{} algorithm that computes $\sed(X)$ if this value is bounded by $k$, in time $\Oh(k^2)$.

\begin{fact}[{\cite[Lemma 4.5]{CKW23}}] \label{lm:selfed}
    There is an $\Oh(k^2)$-time \modelname{} algorithm that, given a string $X\in \Sigma^*$
    and an integer $k\in \ZZ_{> 0}$ determines whether $\sed(X)\le k$ and, if so,
    retrieves the breakpoint representation of an optimal self-alignment $\cA : X \onto X$.
    \lipicsEnd
\end{fact}

For our purposes, we slightly extend this fact.

\begin{fact}[Extension of {\cite[Lemma 4.5]{CKW23}}] \label{lm:selfed2}
    There is an $\Oh(k^2)$-time \modelname{} algorithm that, given a string $X\in \Sigma^*$
    and an integer $k\in \ZZ_{> 0}$, computes the longest prefix $X\fragmentco{0}{i}$ and the longest suffix $X\fragmentco{j}{|X|}$ of $X$ such that $\sed(X\fragmentco{0}{i}), \sed(X\fragmentco{j}{|X|}) \le k$.
    Furthermore, the algorithm
    retrieves the breakpoint representations of some optimal self-alignments $\cA : X\fragmentco{0}{i} \onto X\fragmentco{0}{i}$ and $\cB : X\fragmentco{j}{|X|} \onto X\fragmentco{j}{|X|}$.
\end{fact}

\begin{proof}
    The original statement of {\cite[Lemma 4.5]{CKW23}} allows determining whether $\sed(X) \le k$, and if so, retrieve the optimal self-alignment of $X$.
    However, similarly to the classical Landau--Vishkin algorithm \cite{LV88}, this algorithm additionally for each $j \in \fragmentcc{1}{k}$ computes $\ell_j$ such that $X\fragmentco{0}{\ell_j}$ is the largest prefix of $X$ with $\sed(X\fragmentco{0}{\ell_j}) \le j$.
    Thus, the requested prefix is $X\fragmentco{0}{\ell_k}$, and we then run the alignment retrieval algorithm for it.

    As $\sed(Y) = \sed(Y^R)$ for any string $Y$, the computation of the longest suffix with self-edit distance at most $k$ is analogous.
\end{proof}

An important property of $\sed$ is that if $\sed(X)$ is large, then any two cheap alignments of $X$ onto $Y$ intersect.

\begin{fact}[{\cite[Lemma 4.3]{CKW23}}]\label{lem:sed}
    Consider strings $X,Y\in \Sigma^*$. If two alignments $\cA,\cB\in X\onto Y$
    intersect only at the endpoints $(0,0)$ and $(|X|,|Y|)$, then $\sed(X) \le \ed_{\cA}(X,Y) + \ed_{\cB}(X,Y)$.
\lipicsEnd
\end{fact}

\begin{corollary}\label{cor:sed}
    Consider strings $X,Y\in \Sigma^*$ and alignments $\cA, \cB : X \onto Y$.
    If $(x,y)$ and $(x',y')$ are two subsequent 
    intersection points in $\cA \cap \cB$,
    then $\sed(X\fragmentco{x}{x'})\le \ed_\cA(X,Y)+\ed_\cB(X,Y)$.
\end{corollary}
\begin{proof}
    Let $X'=X\fragmentco{x}{x'}$ and $Y'=Y\fragmentco{y}{y'}$.
    Since $(x,y),(x',y')\in \cA\cap \cB$, 
    we can restrict the alignments $\cA$ and $\cB$
    to alignments $\cA',\cB': X'\onto Y'$.
    Moreover, since $(x,y),(x',y')$ are two subsequent points in the intersection
    $\cA\cap \cB$, we conclude that $\cA'$ and $\cB'$ intersect only at the endpoints.
    Hence, \cref{lem:sed} yields $\sed(X') \le \ed_{\cA'}(X',Y') + \ed_{\cB'}(X',Y')$,
    that is, $\sed(X\fragmentco{x}{x'})\le \ed_{\cA'}(X\fragmentco{x}{x'},Y\fragmentco{y}{y'})
    + \ed_{\cB'}(X\fragmentco{x}{x'},Y\fragmentco{y}{y'})$.
    The final claim follows from $\ed_{\cA'}(X',Y')\le \ed_{\cA}(X,Y)$ and $\ed_{\cB'}(X',Y')\le \ed_{\cB}(X,Y)$.
\end{proof}

We use this fact to derive our main combinatorial lemmas that describe how to obtain an optimal alignment of $X$ onto $Y$ given optimal alignments of some fragments of $X$ and $Y$ if these fragments intersect by parts of ``sufficiently large'' self-edit distance.

\begin{lemma}\label{lem:intersect}
    Consider a normalized weight function $w: \Esigma^2 \to\Rz$, strings $X,Y\in \Sigma^*$, an alignment $(x_i,y_i)_{i=0}^t\eqqcolon\cA : X\onto Y$, and indices $0\le \ell \le r \le t$ satisfying at least one of the following conditions:
    \[\sed(X\fragmentco{x_\ell}{x_r})>4\cdot\wed_\cA(X,Y), \qquad \ell=0,\quad\text{or}\quad r=t.\]
    If $\cB_L : X\fragmentco{0}{x_r}\onto Y\fragmentco{0}{y_r}$  and $\cB_R : X\fragmentco{x_\ell}{|X|}\onto Y\fragmentco{y_\ell}{|Y|}$ are $w$-optimal alignments, then $\cB_L\cap \cB_R\ne \emptyset$ and every point $(x^*,y^*)\in \cB_L\cap \cB_R$ satisfies
    \[\wed(X,Y)=\wed_{\cB_L}(X\fragmentco{0}{x^*},Y\fragmentco{0}{y^*})+\wed_{\cB_R}(X\fragmentco{x^*}{|X|},Y\fragmentco{y^*}{|Y|}).\]
\end{lemma}

\begin{figure}
    \begin{center}
        \begin{tikzpicture}[y=-1cm, scale=.9]
            \useasboundingbox (-0.5, -0.5) rectangle (8.5, 8.5);
            \scope[transform canvas={scale=.7}]
                \def\W{12}
\def\H{12}
\def\cnt{6}
\def\sizeofdot{5pt}
\def\gp{0.0}

\draw[line width=3pt, violet!30, line cap=round] (0, 0) -- (1, 1) -- (1, 1.5) -- (4, 4.5) -- (5, 4.5) -- (6.5, 6) -- (6.5, 8) -- (9, 10.5) -- (10.5, 10.5) -- (\H, \W);

\draw (\H / 3, \W) node[below] {$x_{\ell}$};
\draw (2 * \H / 3, \W) node[below] {$x_{r}$};
\draw (0, 4.5) node[left] {$y_{\ell}$};
\draw (0, 9.5) node[left] {$y_{r}$};

\draw[latex-latex] (\H / 3, -0.3) to node[midway, above] {$\sed > 4 \cdot \wed_{\cA}(X, Y)$} (2 * \H / 3, -0.3);
\draw[dashed, gray] (\H / 3, 0) -- (\H / 3, \W);
\draw[dashed, gray] (2 * \H / 3, 0) -- (2 * \H / 3, \W);
\draw[dashed, gray] (0, 4.5) -- (\H / 3, 4.5);
\draw[dashed, gray] (0, 9.5) -- (2 * \H / 3, 9.5);

\draw[line width=4.5pt, orange!30, line cap=round] (0, 0) -- (1.5, 0) -- (4, 2.5) -- (4, 3.5) -- (4.5,4) -- (4.5,5) -- (6.5, 5) -- (7.5, 6) -- (7.5, 8) -- (8, 8.5) -- (2 * \H / 3, 9.5);

;
\draw[line width=4.5pt, blue!30, line cap=round] (\H / 3, 4.5) -- (5, 5.5) -- (\H / 2, 5.5) -- (6, 6.5) -- (7.5, 8) -- (8.5, 8) -- (11, 10.5) -- (11, 11) -- (\H, \W);

\node[circle,draw=darkgreen, fill=darkgreen, inner sep=0pt,minimum size=\sizeofdot] (olddot) at (4.5,5) {};
\node[circle,draw=darkgreen, fill=darkgreen, inner sep=0pt,minimum size=\sizeofdot] (olddot) at (7.5,8) {};

\draw[line width=2pt, darkgreen!70, line cap=round](0, 0) -- (1.5, 0) -- (4, 2.5) -- (4, 3.5) -- (4.5, 4) -- (4.5,5) -- (6.5, 5) -- (7.5, 6) -- (7.5, 8) -- (8.5, 8) -- (11, 10.5) -- (11, 11) -- (\H, \W);

\node[circle,draw=blue, fill=blue, inner sep=0pt,minimum size=\sizeofdot] at (\H / 3, 4.5) {};
\node[circle,draw=orange, fill=orange, inner sep=0pt,minimum size=\sizeofdot] at (2 * \H / 3, 9.5) {};

\draw[thick] (0, 0) rectangle (\H, \W);
            \endscope
        \end{tikzpicture}
    \end{center}

    \caption{Illustration of \cref{lem:intersect}. An alignment $\cA: X\onto Y$ is depicted in \textcolor{violet!50}{violet}. Point $(x_{\ell}, y_{\ell})$ is given in \textcolor{blue!50}{blue}, and point $(x_r, y_r)$ is given in \textcolor{orange!70}{orange}. The $w$-optimal alignment $\cB_L : X\fragmentco{0}{x_r}\onto Y\fragmentco{0}{y_r}$ is given in \textcolor{orange!70}{orange}, and the $w$-optimal alignment $\cB_R : X\fragmentco{x_\ell}{|X|}\onto Y\fragmentco{y_\ell}{|Y|}$ is given in \textcolor{blue!50}{blue}. The \textcolor{darkgreen!70}{green} dots represent the intersection points of $\cB_L$ and $\cB_R$. The globally $w$-optimal \textcolor{darkgreen!70}{green} alignment is a combination of $\cB_L$ and $\cB_R$.}
    \label{fig:first-combinatorial-lemma}
\end{figure}

\begin{proof}
    See \cref{fig:first-combinatorial-lemma}.
    Let us first consider the case when $\sed(X\fragmentco{x_\ell}{x_r})>4k$, where $k\coloneqq \wed_{\cA}(X,Y)$. 
    Consider alignments $\cC_L$ obtained by concatenating $\cB_L$ with $\cA_R\coloneqq (x_i,y_i)_{i=r}^{t}$ 
    and $\cC_R$ obtained by concatenating $\cA_L\coloneqq (x_i,y_i)_{i=0}^{\ell}$  with $\cB_R$.
    Observe that \begin{align*}
        \wed_{\cC_L}(X,Y) 
        &= \wed_{\cB_L}(X\fragmentco{0}{x_r}, Y\fragmentco{0}{y_r}) + \wed_{\cA}(X\fragmentco{x_r}{|X|}, Y\fragmentco{y_r}{|Y|})\\
        &\le \wed_{\cA}(X\fragmentco{0}{x_r}, Y\fragmentco{0}{y_r}) + \wed_{\cA}(X\fragmentco{x_r}{|X|}, Y\fragmentco{y_r}{|Y|})\\
        &= \wed_{\cA}(X,Y) = k.
    \end{align*}
    A symmetric argument yields $\wed_{\cC_R}(X,Y)\le k$.

    By \cref{cor:sed}, if $(x,y),(x',y')$ are consecutive points in $\cC_L\cap \cC_R$, then $\sed(X\fragmentco{x}{x'})\le 2k$.
    Due to $\sed(X\fragmentco{x_\ell}{x_r})> 4k$, there exists a point $(x^*,y^*)\in \cC_L\cap \cC_R$ such that $x_\ell < x^* < x_r$ and $y_\ell \le y^* \le y_r$. 
    In particular, $\cB_L \cap \cB_R \ne \emptyset$.

    For the remainder of the proof, let us pick $(x^*,y^*) \in \cB_L\cap \cB_R$.         
    Since $\sed(X\fragmentco{x_\ell}{x_r})> 4k$,
    sub-additivity of $\sed$ implies that $\sed(X\fragmentco{x_\ell}{x^*})>2k$ or $\sed(X\fragmentco{x^*}{x_r})>2k$. 
    By symmetry (up to reversal), it suffices to consider the case of $\sed(X\fragmentco{x^*}{x_r})>2k$.

    Consider a $w$-optimal alignment $\cO : X \onto Y$.
    Since $\wed_{\cO}(X,Y) \le \wed_{\cA}(X,Y) = k$, by \cref{cor:sed}, if $(x,y)$ and $(x',y')$ are two consecutive points in $\cC_L\cap \cO$, then $\sed(X\fragmentco{x}{x'})\le 2k$.
    Due to $\sed(X\fragmentco{x^*}{x_r})>2k$, there exists a point $(\bx,\by)\in \cC_L\cap \cO$ such that $x^* < \bx < x_r$ and $y^* \le \by \le y_r$. 
    In particular, $(\bx,\by)\in \cB_L \cap \cO$.
    Since the alignments $\cO$, $\cB_L$, and $\cB_R$ are $w$-optimal, we conclude that 
    \begin{align*}
        \bm{\wed(X,Y)} &= \wed_{\cO}(X,Y) \\
        &= \wed_{\cO}(X\fragmentco{0}{\bx}, Y\fragmentco{0}{\by}) +  \wed_{\cO}(X\fragmentco{\bx}{|X|}, Y\fragmentco{\by}{|Y|}) \\
        &= \wed_{\cB_L}(X\fragmentco{0}{\bx}, Y\fragmentco{0}{\by}) +  \wed_{\cO}(X\fragmentco{\bx}{|X|}, Y\fragmentco{\by}{|Y|}) \\
        &= \wed_{\cB_L}(X\fragmentco{0}{x^*}, Y\fragmentco{0}{y^*}) +\wed_{\cB_L}(X\fragmentco{x^*}{\bx}, Y\fragmentco{y^*}{\by}) +  \wed_{\cO}(X\fragmentco{\bx}{|X|}, Y\fragmentco{\by}{|Y|}) \\
        &= \wed_{\cB_L}(X\fragmentco{0}{x^*}, Y\fragmentco{0}{y^*}) +\wed(X\fragmentco{x^*}{\bx}, Y\fragmentco{y^*}{\by}) +  \wed(X\fragmentco{\bx}{|X|}, Y\fragmentco{\by}{|Y|}) \\
        &\ge \wed_{\cB_L}(X\fragmentco{0}{x^*}, Y\fragmentco{0}{y^*}) + \wed(X\fragmentco{x^*}{|X|}, Y\fragmentco{y^*}{|Y|}) \\
        &\bm{ = \wed_{\cB_L}(X\fragmentco{0}{x^*}, Y\fragmentco{0}{y^*}) + \wed_{\cB_R}(X\fragmentco{x^*}{|X|}, Y\fragmentco{y^*}{|Y|})} \\
        &= \wed(X\fragmentco{0}{x^*}, Y\fragmentco{0}{y^*}) + \wed(X\fragmentco{x^*}{|X|}, Y\fragmentco{y^*}{|Y|}) \\
        & \bm{\ge \wed(X,Y)}.
    \end{align*}

    It remains to consider the cases of $\ell=0$ and $r=t$. 
    By symmetry (up to reversal), we can solely focus on $\ell=0$.
    Then, $(x_\ell,y_\ell)=(x_0,y_0)=(0,0)\in \cB_L\cap \cB_R$.
    For the remainder of the proof, let us pick $(x^*,y^*) \in \cB_L\cap \cB_R$.     
    Since the alignments $\cB_L:X\fragmentco{0}{x_r}\onto Y\fragmentco{0}{y_r} $ and $\cB_R:X\onto Y$ are $w$-optimal, we conclude that 
    \begin{align*}
        \wed(X,Y) &= \wed_{\cB_R}(X,Y) \\
        &= \wed_{\cB_R}(X\fragmentco{0}{x^*}, Y\fragmentco{0}{y^*}) +\wed_{\cB_R}(X\fragmentco{x^*}{|X|}, Y\fragmentco{y^*}{|Y|})\\
        &= \wed(X\fragmentco{0}{x^*}, Y\fragmentco{0}{y^*}) +\wed_{\cB_R}(X\fragmentco{x^*}{|X|}, Y\fragmentco{y^*}{|Y|})\\
        &= \wed_{\cB_L}(X\fragmentco{0}{x^*}, Y\fragmentco{0}{y^*}) +\wed_{\cB_R}(X\fragmentco{x^*}{|X|}, Y\fragmentco{y^*}{|Y|}).\qedhere
    \end{align*}    
\end{proof}

\begin{proposition}\label{cor:intersect}
    Consider a normalized weight function $w: \Esigma^2 \to\Rz$, strings $X,Y\in \Sigma^*$, an alignment $(x_i,y_i)_{i=0}^t\eqqcolon\cA : X\onto Y$, and indices $0\le \ell \le p \le q \le r \le t$ such that
    \[\sed(X\fragmentco{x_\ell}{x_p})>4\cdot\wed_\cA(X,Y)\;\;\text{or}\;\;\ell=0,\quad\text{and}\quad \sed(X\fragmentco{x_q}{x_r})>4\cdot\wed_\cA(X,Y)\;\;\text{or}\;\; r=t.\]
    If $\cO_L : X\fragmentco{0}{x_p}\onto Y\fragmentco{0}{y_p}$, $\cO_M : X\fragmentco{x_\ell}{x_r}\onto Y\fragmentco{y_\ell}{y_r}$, and $\cO_R : X\fragmentco{x_q}{|X|}\onto Y\fragmentco{y_q}{|Y|}$ are $w$-optimal alignments, then $\cO_L\cap \cO_M\ne \emptyset \ne \cO_R\cap \cO_M$, and all points $(x^*_L,y^*_L)\in \cO_L\cap \cO_M$ and $(x^*_R,y^*_R)\in \cO_R\cap \cO_M$ satisfy
    \[\wed\!(X,Y)\!=\!\wed_{\cO_L}\!(X\fragmentco{0}{x^*_L},\!Y\fragmentco{0}{y^*_L})+\wed_{\cO_M}\!(X\fragmentco{x^*_L}{x^*_R},\! Y\fragmentco{y^*_L}{y^*_R})+\wed_{\cO_R}\!(X\fragmentco{x^*_R}{|X|},\!Y\fragmentco{y^*_R}{|Y|}).\]\maybe
\end{proposition}
\begin{proof}
See \cref{fig:new-divide-and-conquer} for an example situation with $p=q$.
Let us first apply \cref{lem:intersect} for an alignment $(x_i,y_i)_{i=0}^r :X\fragmentco{0}{x_r}\onto Y\fragmentco{0}{y_r}$ and $w$-optimal alignments $\cO_L : X\fragmentco{0}{x_p}\onto Y\fragmentco{0}{y_p}$ and $\cO_M : X\fragmentco{x_\ell}{x_r}\onto Y\fragmentco{y_\ell}{y_r}$. 
The assumptions are satisfied because $\ell=0$ or $\sed(X\fragmentco{x_\ell}{x_p})>4\cdot\wed_\cA(X,Y)\ge 4\cdot\wed_\cA(X\fragmentco{0}{x_r},Y\fragmentco{0}{y_r})$.
Consequently, $\cO_L\cap \cO_M\ne \emptyset$ and every $(x^*_L,y^*_L)\in \cO_L\cap \cO_M$ satisfies
\[\wed(X\fragmentco{0}{x_r},Y\fragmentco{0}{y_r})=\wed_{\cO_L}(X\fragmentco{0}{x^*_L},Y\fragmentco{0}{y^*_L})+\wed_{\cO_M}(X\fragmentco{x^*_L}{x_r},Y\fragmentco{y^*_L}{y_r}).\]
Let $\cB_L : X\fragmentco{0}{x_r}\onto Y\fragmentco{0}{y_r}$ be a $w$-optimal alignment that follows $\cO_L$ from $(0,0)$ to $(x^*_L,y^*_L)$ and then follows $\cO_M$ from $(x^*_L,y^*_L)$ to $(x_r,y_r)$.

Next, we apply \cref{lem:intersect} for the alignment $\cA :X\onto Y$ as well as $w$-optimal alignments $\cB_L : X\fragmentco{0}{x_r}\onto Y\fragmentco{0}{y_r}$ and $\cO_R : X\fragmentco{x_q}{|X|}\onto Y\fragmentco{y_q}{|Y|}$. 
The assumptions are satisfied because $r=t$ or $\sed(X\fragmentco{x_q}{x_r})>4\cdot\wed_\cA(X,Y)$.
The definition of $\cB_L$ implies that $\cB_L \cap \cO_R = \cO_M \cap \cO_R$.
Consequently,  $\cO_M \cap \cO_R = \cB_L\cap \cO_R\ne \emptyset$ and every $(x^*_R,y^*_R)\in \cB_L\cap \cO_R = \cO_M \cap \cO_R$ satisfies
\[\wed(X,Y)=\wed_{\cB_L}(X\fragmentco{0}{x^*_R},Y\fragmentco{0}{y^*_R})+\wed_{\cO_R}(X\fragmentco{x^*_R}{|X|},Y\fragmentco{y^*_R}{|Y|}).\]
Since $x^*_L \le x_p \le x_q \le x^*_R$ and $y^*_L \le y_p \le y_q \le y^*_R$, we have
\[\wed_{\cB_L}(X\fragmentco{0}{x^*_R},Y\fragmentco{0}{y^*_R})=\wed_{\cO_L}(X\fragmentco{0}{x^*_L},Y\fragmentco{0}{y^*_L})+\wed_{\cO_M}(X\fragmentco{x^*_L}{x^*_R},Y\fragmentco{y^*_L}{y^*_R}).\]
Therefore, 
\[\wed\!(X,Y)\!=\! \wed_{\cO_L}\!(X\fragmentco{0}{x^*_L},Y\fragmentco{0}{y^*_L})+\wed_{\cO_M}\!(X\fragmentco{x^*_L}{x^*_R}, Y\fragmentco{y^*_L}{y^*_R})+\wed_{\cO_R}\!(X\fragmentco{x^*_R}{|X|},Y\fragmentco{y^*_R}{|Y|})\]
holds as claimed.
\end{proof}

\section{Static Algorithm} \label{sec:full-algo}

In this section we describe our main static algorithm that, given two strings $X, Y \in \Sigma^{\le n}$ and a weight function $w : \Esigma^2 \to [0 \dd W]$, finds $k = \wed(X, Y)$ and a corresponding $w$-optimal alignment in time $\Oh(k^2 \log^2 n \cdot \min\{W, \sqrt k \log n\})$.
We first design an algorithm that works if $X$ has small self-edit distance, and then use it to design an algorithm for the general case.

\subsection{The Case of Small Self-Edit Distance}\label{sec:small-sed-algo}

If the input strings are highly compressible, in the sense that their self-edit distance is small, we use the following fact from \cite{CKW23} to decompose them into a small number of different pieces.

\begin{lemma}[{\cite[Lemma 4.6 and Claim 4.11]{CKW23}}]\label{lem:decomp2}
There is a \modelname{} algorithm that, given a positive integer $k$ and strings $X,Y \in \Sigma^{*}$ satisfying $\sed(X)\le k \le |X|$ and $\ed(X,Y)\le k$, in $\Oh(k^2)$ time builds a decomposition $X=\bigodot_{i=0}^{m-1} X_i$ and a sequence of fragments $(Y_i)_{i=0}^{m-1}$ of $Y$ such that
\begin{itemize}
    \item For $i\in \fragmentco{0}{m}$, each phrase $X_i=X\fragmentco{x_{i}}{x_{i+1}}$ is of length $x_{i+1}-x_i\in \fragmentco{k}{2k}$ and each fragment $Y_i=Y\fragmentco{y_i}{y'_{i+1}}$ satisfies $y_i=\max\{x_i-k,0\}$ and $y'_{i+1}=\min\{x_{i+1}+3k,|Y|\}$.
    \item There is a set $F \subseteq \fragmentco{0}{m}$ of size $|F|= \Oh(k)$ such that $X[x_i \dd x_{i+1})=X\fragmentco{x_{i-1}}{x_{i}}$ and $Y[y_i \dd y'_{i+1})=Y\fragmentco{y_{i-1}}{y'_{i}}$ holds for each $i\in \fragmentco{0}{m}\setminus F$ (in particular, $0\in F$).
\end{itemize}
The algorithm returns the set $F$ and, for all $i\in F$, the endpoints of $X_i=X\fragmentco{x_{i}}{x_{i+1}}$.\footnote{This determines the whole decomposition because $(x_i)_{i=\ell}^r$ is an arithmetic progression for every $\fragmentoc{\ell}{r}\subseteq \fragmentoc{0}{m}\setminus F$.}
\end{lemma}
\begin{proof}
Since \cite[Claim 4.11]{CKW23} is not a stand-alone statement, we provide some details for completeness.
A black-box application of \cite[Lemma 4.6]{CKW23} yields a decomposition $X=\bigodot_{i=0}^{m-1} X_i$ into phrases of length $|X_i|\in \fragmentco{k}{2k}$ and a set $F \subseteq [0 \dd m)$ of size at most $2k$ such that $X_i=X_{i-1}$ holds for each $i\in [0 \dd m)\setminus F$. 
The goal of \cite[Claim 4.11]{CKW23} is to add to $F$ every index $i\in \fragmentoo{0}{m}$ satisfying $Y_{i-1}\ne Y_i$.
For this, the algorithm constructs the breakpoint representation of an optimal (unweighted) alignment $\cA : X\onto Y$.
As argued in the proof of \cite[Claim 4.11]{CKW23}, the equality $Y_{i-1}=Y_i$ holds for each $i\in \fragmentcc{7}{m-8}$
such that $\fragmentcc{i-6}{i+7}\cap F = \emptyset$ and $\cA$ does not make any edit within $X_{i-6}X_{i-5}\cdots X_{i+7}$.
The number of indices violating at least one of these conditions does not exceed $14(1+|F|+k) \le 14\cdot 4k = 56k$

The application of \cite[Lemma 4.6]{CKW23} and the construction of the breakpoint representation of $\cA$ take $\Oh(k^2)$ time in the \modelname{} model (see \cref{lm:k2-ed}).
The final step of extending $F$ is implemented in $\Oh(k)$ time using a left-to-right scan of $F$ and the breakpoint representation of $\cA$ to determine, for each edit in $\cA$, the index $i$ of the affected phrase $X_i$.
\end{proof}

We now use the decomposition of \cref{lem:decomp2} to build a hierarchical alignment data structures of \cref{prp:hadstwo} for all pairs $(X_i, Y_i)$ of the decomposition.

\begin{lemma}\label{lem:addtoboxds}
Consider a weight function $w : \overline{\Sigma}^2 \to \fragmentcc{0}{W}$, a positive integer $k$, and two strings $X, Y \in \Sigma^{\le n}$ such that $\sed(X) \le k \le |X|$ and $\ed(X,Y)\le k$. 
Moreover, suppose that $X=\bigodot_{i=0}^{m-1}X_i$,  $(Y_i)_{i=0}^{m-1}$, and $F\subseteq\fragmentco{0}{m}$ satisfy the conditions in the statement of \cref{lem:decomp2}.
There is an algorithm that, given oracle access to $w$, the integer $k$, the strings $X,Y$, the set $F$, and the endpoints of $X_i$ for each $i\in F$, takes $\Oh(W\cdot k^2\log^2n)$ time plus $\Oh(k^2)$ \modelname{} operations and constructs for every $i \in F$ a hierarchical alignment data structure $\boxds(X_i, Y_i)$.

A variant of this algorithm that builds relaxed hierarchical alignment data structures $\boxdsk(X_i, Y_i)$ for all $i \in F$ takes $\Oh(k^{2.5}\log^2 n)$ time plus $\Oh(k^2)$ \modelname{} operations.
\end{lemma}
\begin{proof}
We use \cref{lm:selfed} to find the breakpoint representations of optimal self-alignments $\cA : X \onto X$ and $\cB : Y\onto Y$.
We can assume that neither alignment contains points below the main diagonal because the parts below the main diagonal can be mirrored along the main diagonal without changing the cost of the alignment.
Note that $\ed_{\cA}(X,X)\le k$ holds by our assumption and $\ed_{\cB}(Y,Y)\le \sed(X)+2\ed(X,Y)\le 3k$ follows from \cref{fct:selfed-properties}.
This initial phase of the algorithm takes $\Oh(k^2)$ time and \modelname{} operations.

We process indices $i\in F$ from left to right.

We build a hierarchical alignment data structure of \cref{prp:hadstwo} for $X_0$ and $Y_0$.
It takes time $\Oh(W \cdot k^2 \log n)$ (or $\Oh(k^{2.5}\log n)$ for the relaxed version).

For subsequent positions $i\in F\setminus\{0\}$, we assume that we already have computed $\boxds(X_j, Y_j)$ for $j=\max(F\cap \fragmentco{0}{i})$. 
Since $X_j=X_{j+1}=\cdots =X_{i-1}$ and $Y_j=Y_{j+1}=\cdots =Y_{i-1}$, we thus have $\boxds(X_{i-1}, Y_{i-1})$.
Our next goal is to build $\boxds(X_i, Y_i)$.
For this, we first restrict the alignment $\cA$ to $\cA_i : X_i \onto \cA(X_i)$.
We claim that $X_{i-1}$ can be transformed into $X_i$ via at most $2 \cdot \ed_{\cA_i}(X_i, \cA(X_i)) + 2$ mega-edits of \cref{def:mega-edit}.
The breakpoint representation of $\cA_i$ lets us decompose $X_i$ into individual characters that $\cA_i$ deletes or substitutes and fragments that $\cA_i$ matches perfectly.
We process them one by one while implicitly maintaining a string $X\fragmentco{x_{i-1}}{x}$ for some index $x$.
Initially, $x = x_i$ and $X\fragmentco{x_{i-1}}{x} = X_{i-1}$.
If $X[x]$ is a deleted or substituted character, transforming $X\fragmentco{x_{i-1}}{x}$ into $X\fragmentcc{x_{i-1}}{x}$ is a single character insertion.
The more interesting case is when $X\fragmentco{x}{x'}$ is matched perfectly to some $X\fragmentco{x-d}{x'-d}$.
Note that $d>0$ because $\cA$ is a self-alignment that does not contain points below the main diagonal and $x-d \ge x-\sed(X) \ge x_i-\sed(X) \ge x_i-k \ge x_{i-1}$ due to \cref{lm:paths_dont_deviate_too_much}.
As $X\fragmentco{x-d}{x'}$ has period $d$, we have that $X\fragmentco{x}{x'}$ is a prefix of $X\fragmentco{x-d}{x}^{\infty}$, and as $x-d \ge x_{i-1}$, we can obtain $X\fragmentco{x_{i-1}}{x'}$ from $X\fragmentco{x_{i-1}}{x}$ via a single copy-exponentiate-paste operation of \cref{def:mega-edit}.
Having processed the whole $\cA_i$, that is, when $x=x_{i+1}$, our current string is $X\fragmentco{x_{i-1}}{x_i}$.
The string $X_i$ can be obtained from it via a single cut operation of \cref{def:mega-edit}.
In total, we transformed $X_{i-1}$ into $X_i$ via at most $2 \cdot \ed_{\cA_i}(X_i, \cA(X_i)) + 2$ mega-edits as every perfectly matched fragment is followed by either a character edit or the end of $\cA_i$.
By repeatedly applying the mega-edits to $\boxds(X_{i-1}, Y_{i-1})$, we obtain $\boxds(X_i, Y_{i-1})$.

The process of building $\boxds(X_i, Y_i)$ from $\boxds(X_i, Y_{i-1})$ is very similar, with just minor differences stemming from the fact that $Y_i$ overlaps $Y_{i-1}$.
First, we restrict the alignment $\cB$ to $\cB_i : Y'_i\onto \cB(Y'_i)$, where $Y'_i=Y\fragmentco{y'_i}{y'_{i+1}}$.
We claim that $Y_{i-1}$ can be transformed into $Y_i$ via at most $2 \cdot \ed_{\cB_i}(Y'_i, \cB(Y'_i)) + 2$ mega-edits.
The breakpoint representation of $\cB_i$ lets us decompose $Y'_i$ into individual characters that $\cB_i$ deletes or substitutes and fragments that $\cB_i$ matches perfectly.
We process them one by one while implicitly maintaining a string $Y\fragmentco{y_{i-1}}{y}$ for some index $y$.
Initially, $y = y'_i$ and $Y\fragmentco{y_{i-1}}{y} = Y_{i-1}$.
If $Y[y]$ is a deleted or substituted character, transforming $Y\fragmentco{y_{i-1}}{y}$ into $Y\fragmentcc{y_{i-1}}{y}$ is a single character insertion.
If $Y\fragmentco{y}{y'}$ is matched perfectly to $Y\fragmentco{y-d}{y'-d}$, then $d>0$ because $\cB$ is a self-alignment that does not contain points below the main diagonal and $y-d \ge y'_i-\sed(Y) \ge y'_{i-1}-3k \ge y_{i-1}$; here, $y'_{i-1}-3k \ge y_{i-1}$ follows from $y'_{i-1}\ne |Y|$, which we can assume because $y'>y\ge y'_i\ge y'_{i-1}$.
As $Y\fragmentco{y-d}{y'}$ has period $d$, we have that $Y\fragmentco{y}{y'}$ is a prefix of $Y\fragmentco{y-d}{y}^{\infty}$, and as $y-d \ge y_{i-1}$, we can obtain $Y\fragmentco{y_{i-1}}{y'}$ from $Y\fragmentco{y_{i-1}}{y}$ via a single copy-exponentiate-paste operation of \cref{def:mega-edit}.
Having processed the whole $Y_i$, that is, when $y=y'_{i+1}$, our current string is $Y\fragmentco{y_{i-1}}{y'_{i+1}}$.
The string $Y_i$ can be obtained from it via a single cut operation of \cref{def:mega-edit}.
In total, we transformed $Y_{i-1}$ into $Y_i$ via at most $2 \cdot \ed_{\cB_i}(Y'_i, \cB(Y'_i)) + 2$ mega-edits.
By repeatedly applying the mega-edits to $\boxds(X_i, Y_{i-1})$, we obtain $\boxds(X_i, Y_i)$.

\smallskip

The total number of mega-edits is bounded by
\[\sum_{i \in F} 2 \cdot \ed_{\cA_i}(X_i, \cA(X_i)) + 2 + 2 \cdot \ed_{\cB_i}(Y'_i, \cB(Y'_i)) + 2 \le 2 \cdot \ed_{\cA}(X, X) + 2 \cdot |F| + 2 \cdot \ed_{\cB}(Y, Y) + 2 \cdot |F| = \Oh(k)\]
as fragments $X_i$ are disjoint and fragments $Y'_i$ are disjoint.
Thus, the total time complexity of building hierarchical alignment data structures for all pairs $(X_i, Y_i)$ for $i \in F$ is $\Oh(W \cdot k^2 \log^2 k)$ for $\boxds$ and $\Oh(k^{2.5} \log^2 k)$ for $\boxdsk$, and this dominates the total time complexity of the algorithm.
Furthermore, in both cases we pay $\Oh(k^2)$ \modelname{} operations for preprocessing.
\end{proof}

We now use the hierarchical alignment data structure constructed by \cref{lem:addtoboxds} to find a $w$-optimal alignment for the complete strings $X$ and $Y$.

\begin{lemma} \label{lm:small-sed-algo}
    There is an algorithm that, given oracle access to a weight function $w : \overline{\Sigma}^2 \to \fragmentcc{0}{W}$, a positive threshold $k$, and strings $X, Y \in \Sigma^{\le n}$ such that $\sed(X) \le k$ and $\wed(X, Y) \le k$, constructs the breakpoint representation of a $w$-optimal alignment of $X$ onto $Y$ and takes $\Oh(\min\{k^2 W \log^2 n, k^{2.5} \log^2 n\})$ time plus $\Oh(k^2)$ \modelname{} operations.
\end{lemma}
\begin{proof}
    If $|X|< k$, we solve the problem using \cref{lm:baseline-wed}.
    This algorithm takes $\Oh((|X|+1)k)$ time, which is $\Oh(k^2)$ since $|X|< k$.
    Thus, we henceforth assume that $|X|\ge k$ and, in particular, $n \ge k$.

    We first describe an algorithm that works in time $\Oh(k^2 W \log^2 n)$, and later transform it into an algorithm that works in time $\Oh(k^{2.5} \log^2 n)$.
    By running these two algorithms in parallel and returning the answer of whichever one of them terminates first, we get the desired time complexity.

    We follow a setup similar to the proof of {\cite[Lemma 4.9]{CKW23}} in \modelname{} model for $d \coloneqq k$.
    First, we run the algorithm of \cref{lem:decomp2} for $X$, $Y$, and $k$, arriving at a decomposition $X=\bigodot_{i=0}^{m - 1}X_i$
    and a sequence of fragments $(Y_i)_{i=0}^{m-1}$, such that $X_i=X\fragmentco{x_i}{x_{i+1}}$ and $Y_i=Y\fragmentco{y_i}{y'_{i+1}}$ for $y_i=\max\{x_i-k,0\}$ and $y'_{i+1}=\min\{x_{i+1}+3k,|Y|\}$ for all $i\in \fragmentco{0}{m}$.
    The decomposition is represented using a set $F$ of size $\Oh(k)$ such that $X_i=X_{i-1}$ and $Y_i=Y_{i-1}$ holds for each $i\in \fragmentco{0}{m}\setminus F$ and the endpoints of $X_i$ for $i\in F$. 
    To easily handle corner cases, we set $y'_0=0$ and $y_m=|Y|$, and we assume that $\fragmentco{0}{m}\setminus F\subseteq \fragmentcc{2}{m-5}$; if the original set $F$ does not satisfy this condition, we add the missing $\Oh(1)$ elements.

    Having constructed (and possibly extended) $F$, we apply \cref{lem:addtoboxds}.
    It yields $\boxds(X_i, Y_i)$ for every $i \in F$.

    For each $i\in \fragmentco{0}{m}$, consider a subgraph $G_i$ of $\oAGw(X, Y)$ induced by $\fragmentcc{x_i}{x_{i+1}}\times \fragmentcc{y_i}{y'_{i+1}}$ (see \cref{fig:substitutions-only}). Denote by $G$ the union of all subgraphs $G_i$.
    Note that $G$ contains all vertices $(x, y)\in \fragmentcc{0}{|X|}\times \fragmentcc{0}{|Y|}$ with $|x - y| \le k$.
    Therefore, as we know that $\wed(X, Y) \le k$, \cref{lm:paths_dont_deviate_too_much} implies that the $w$-optimal alignment of $X$ onto $Y$ can go only through vertices $(x, y)$ of $\oAGw(X, Y)$ with $|x - y| \le k$, and thus the distance from $(0, 0)$ to $(|X|, |Y|)$ in $\oAGw(X, Y)$ is the same as the distance from $(0, 0)$ to $(|X|, |Y|)$ in $G$.
    For each $i \in [0\dd m]$, denote $V_i \coloneqq \{(x, y) \in V(G) \mid x = x_i, y \in \fragmentcc{y_i}{y'_i}\}$.
    Note that $V_0=\{(0,0)\}$, $V_m=\{(|X|,|Y|)\}$, and $V_i = V(G_i) \cap V(G_{i - 1})$ for $i \in (0 \dd m)$.
Moreover, for $i, j \in [0\dd m]$ with $i \le j$, let $D_{i, j}$ denote the matrix of pairwise distances from $V_i$ to $V_j$ in $G$, where rows of $D_{i, j}$ represent vertices of $V_i$ in the decreasing order of the second coordinate, and columns of $D_{i, j}$ represent vertices of $V_j$ in the decreasing order of the second coordinate.
By \cref{grid-graph-extension} the shortest paths between vertices of $G_i$ stay within $G_i$, so $D_{i, i + 1}$ is a contiguous submatrix of $\BMw(X_i, Y_i)$. 
Hence, for each $i \in F$, we can retrieve $\mds(D_{i, i + 1})$ using \cref{lm:ds_tools} and matrix retrieval operation of $\boxds(X_i, Y_i)$. 

\begin{claim} \label{clm:D-formula}
    Let $0=j_0 < j_1 < \ldots < j_{|F| - 1} < j_{|F|}=m$ be the elements of $F\cup \{m\}$. 
    We have
    \[D_{0, m} = \bigotimes_{i=0}^{|F|-1} D_{j_i, j_i + 1}^{\otimes (j_{i + 1} - j_i)}. \]
\end{claim}
\begin{claimproof}
    Note that, for every $a,i,b \in \fragmentcc{0}{m}$ with $a\le i \le b$, the set $V_i$ is a separator in $G$ between all vertices of $V_{a}$ and $V_b$. Thus, every path from $V_{a}$ to $V_b$ in $G$ crosses $V_i$, so $D_{a, b} = D_{a, i} \otimes D_{i, b}$.
    In general, this implies $D_{0,m}=\bigotimes_{i=0}^{m-1} D_{i,i+1}$.

    It remains to prove that, for every $i\in \fragmentco{0}{|F|}$, we have $D_{j_i, j_i + 1} = D_{j_i + 1, j_i + 2} = \cdots = D_{j_{i + 1} - 1, j_{i + 1}}$. 
    If $j_{i+1}=j_i+1$, this statement is obvious.
    Otherwise, the condition on $F$ stipulated in \cref{lem:decomp2}, implies that the graphs $G_{j_i}, G_{j_i + 1}, \ldots, G_{j_{i + 1} - 1}$ are isomorphic to $\oAGw(X_{j_i},Y_{j_i})$.
    Moreover, since we assumed that $\fragmentco{0}{m}\setminus F \subseteq \fragmentcc{2}{m-5}$, we have $1 \le j_i < j_{i+1}\le m-4$.
    Since each phrase is of length at least $k$ and $|Y|\ge |X|-\wed(X,Y)\ge |X|-k$, for each $p\in \fragmentcc{1}{m-4}$,
    we have $x_{p}\ge x_{p-1}+k \ge k$ and $x_{p}\le x_{p+4}-4k \le |X|-4k \le |Y|-3k$.
    In particular, $\fragmentcc{y_p}{y'_p}=\fragmentcc{x_p-k}{x_p+3k}$.
    Hence, for each $p \in \fragmentoo{j_i}{j_{i+1}}$, not only $G_{p}$ is isomorphic to $G_{j_i}$, but also $V_{p}$ and $V_{p + 1}$ in this isomorphism correspond to $V_{j_i}$ and $V_{j_i + 1}$.
Hence, we indeed have $D_{j_i, j_i + 1} = D_{j_i + 1, j_i + 2} = \cdots = D_{j_{i + 1} - 1, j_{i + 1}}$.
Therefore, $D_{j_i, j_{i + 1}} = D_{j_i, j_i + 1}^{\otimes (j_{i + 1} - j_i)}$.
We thus obtain $D_{0, m} = \bigotimes_{i =0}^{|F|-1} D_{j_i, j_i + 1}^{\otimes (j_{i + 1} - j_i)}$.
\end{claimproof}

Note that for each $i \in \fragmentco{0}{|F|}$, we already computed $\mds(D_{j_i, j_i + 1})$.
Given $\mds(D_{a, b})$ and $\mds(D_{b, c})$ for any $a,b,c\in \fragmentcc{0}{m}$ with $a < b < c$, we can compute $\mds(D_{a, c})$ using \cref{lm:matrix_mult_w} in time $\Oh((k + \delta(D_{a, b}) + \delta(D_{b, c})) \log n) = \Oh(W \cdot k \log n)$, where $\delta(D_{a, b}), \delta(D_{a, c}) = \Oh(W \cdot k)$ due to \cref{lem:bdtocore,lem:bdproduct,obs:bmw_is_bd}.
Therefore, using the idea of binary exponentiation, we can compute $\mds(M_i)$ where $M_i \coloneqq D_{j_i, j_i + 1}^{\otimes (j_{i + 1} - j_i)}$ in time $\Oh(W \cdot k \log^2 n)$.
Doing it for all $i \in \fragmentco{0}{|F|}$ requires time $\Oh(W \cdot k^2 \log^2 n)$.
We then compute $\mds(D_{0, m})$ using the fact that $D_{0, m} = M_0 \otimes M_1 \otimes \cdots \otimes M_{|F| - 1}$ in time $\Oh(|F| \cdot (W \cdot k \log n)) = \Oh(W \cdot k^2 \log n)$, where we compute $M_0 \otimes M_1 \otimes \cdots \otimes M_{|F| - 1}$ in a perfectly balanced binary tree manner.

It remains to reconstruct a $w$-optimal alignment $X\onto Y$; its cost $\wed(X, Y)$ is the only entry of $D_{0, m}$.
For that, we additionally save all the $\mds$ data structures that were computed throughout the time of the algorithm.
We obtained $\mds(D_{0, m})$ by a logarithmic-depth min-plus matrix multiplication sequence from $\mds(D_{j_i, j_i + 1})$s.
When two matrices are multiplied, a specific entry of the answer is equal to the minimum of the sum of $\Oh(k)$ pairs of entries from the input matrices as all matrices we work with have sizes $\Oh(k) \times \Oh(k)$.
Hence, in time $\Oh(k \log n)$ we can reconstruct the entries of the input matrices, from which this value arose by accessing $\mds$ data structures that we saved for all of them.

We use recursive backtracking similar to the one used in the alignment retrieval operation of \cref{prp:hadstwo}.
Whenever the current value we are trying to backtrack is zero, we terminate as there is a single path between any two vertices of $\oAGw(X, Y)$ that has weight zero.
In each internal node of this process, we spend at most $\Oh(k \log n)$ time to backtrack the value to its child nodes.
For each leaf, we apply alignment retrieval of \cref{prp:hadstwo} to reconstruct the $w$-optimal alignment inside $G_i$ at the cost of at most $\Oh(k\log^2 n)$.
As $\wed(X, Y) \le k$ and the process has logarithmic depth, there are $\Oh(k \log n)$ internal nodes and $\Oh(k)$ leaves of this process with non-zero values.
Therefore, the whole procedure takes time $\Oh(k \log n \cdot k \log n + k\cdot k\log^2 n) = \Oh(k^2 \log^2 n)$.

It remains to analyze the time complexity of the complete algorithm.
We spend $\Oh(k^2)$ time and $\Oh(k^2)$ \modelname{} operations for \cref{lem:decomp2}.
The application of \cref{lem:addtoboxds} takes time $\Oh(k^2 \cdot W \log^2 n)$ time and $\Oh(k^2)$ \modelname{} operations.
Computation of $D_{i, i + 1}$ for $i \in F$ takes time $\Oh(|F| \cdot (k \log n + k \cdot W \log n)) = \Oh(k^2 \cdot W \log n)$ due to \cref{obs:bmw_is_bd}.
Computation of $D_{0, m}$ then takes $\Oh(W \cdot k^2 \log^2 n)$ time.
Reconstructing the alignment takes time $\Oh(k^2 \log^2 n)$.
Thus, the whole algorithm works in $\Oh(k^2 \cdot W \log^2  n)$ time plus $\Oh(k^2)$ \modelname{} operations.

We now show how to transform the presented algorithm into the one that takes $\Oh(k^{2.5} \log^2 n)$ time.
The only thing we change is that we now use \cref{lm:matrix_mult_k} instead of \cref{lm:matrix_mult_w} whenever we multiply $\mds$ and use the variant of the algorithm from \cref{lem:addtoboxds} that builds a relaxed hierarchical alignment data structure $\boxdsk$ when we apply it.
As basic operations over matrices preserve $k$-equivalence due to \cref{c-equiv-preserv}, at the end we get some value that is $k$-equivalent to $\wed(X, Y)$.
As we know that $\wed(X, Y) \le k$, we find its exact value.
The alignment reconstruction process works the same way as all values we backtrack to also have values at most $k$, and thus are computed correctly.

The application of \cref{lem:decomp2} still takes $\Oh(k^2)$ time and $\Oh(k^2)$ \modelname{} operations.
\Cref{lem:addtoboxds} now takes $\Oh(k^{2.5} \log^2 n)$ time plus $\Oh(k^2)$ \modelname{} operations.
Finally, we have $\Oh(k \log n)$ min-plus matrix multiplications in \cref{clm:D-formula}, and thus it takes time $\Oh(k^{2.5} \log^2 n)$ to perform all of them using \cref{lm:matrix_mult_k} as all the matrices we work with have cores of size $\Oh(k \sqrt k)$.
At the end, we reconstruct the alignment in time $\Oh(k^2 \log^2 n)$.
Hence, the entire algorithm works in $\Oh(k^{2.5} \log^2 n)$ time plus $\Oh(k^2)$ \modelname{} operations.
\end{proof}

\subsection{The General Case}\label{sec:general}

In the general case, we first design a procedure that ``improves'' an alignment as discussed in \cref{sec:overview}.
That is, given some alignment $\cA : X \onto Y$ (some approximation of a $w$-optimal alignment), finds a $w$-optimal alignment $\cB : X \onto Y$.
For that, we use divide-and-conquer combined with the combinatorial fact of \cref{cor:intersect}: we recursively compute $w$-optimal alignments for the two halves, and then use them and \cref{lm:small-sed-algo} to obtain a $w$-optimal alignment for the whole string.

\begin{lemma} \label{lm:upgrade_alignment}
    There is an algorithm that, given two strings $X, Y \in \Sigma^{\le n}$, oracle access to a weight function $w : \Esigma^2 \to [0 \dd W]$, and the breakpoint representation of an alignment $\cA:X\onto Y$, finds the breakpoint representation of a $w$-optimal alignment $\cB : X\onto Y$. The running time of the algorithm is  $\Oh(\min\{k^2 \cdot W \log^2 n, k^{2.5} \log^2 n\})$ time plus $\Oh(k^2)$ \modelname{} operations, where $k = \wed_{\cA}(X, Y)$.
\end{lemma}

\begin{proof}
    \SetKwFunction{upgr}{ImproveAlignment}

    We develop a recursive procedure \upgr that satisfies the requirements of \cref{lm:upgrade_alignment}; see \cref{alg:upgrade_alignment} and \cref{fig:new-divide-and-conquer}\footnote{Note that \cref{fig:new-divide-and-conquer} depicts a simplified version of the algorithm, in which the right recursive call is run for $X\fragmentco{x_{m}}{|X|}$ and $Y\fragmentco{y_{m}}{|Y|}$.}.
    
    \begin{algorithm}[t]
        \caption{The algorithm from \cref{lm:upgrade_alignment}. Given strings $X$ and $Y$, oracle access to a weight function $w$, and the breakpoint representation of an alignment $\cA$ of $X$ onto $Y$, the algorithm computes the breakpoint representation of a $w$-optimal alignment of $X$ onto $Y$.} \label{alg:upgrade_alignment}
    \upgr{$X, Y, w, \cA$}\Begin{
        $k \gets \wed_{\cA}(X, Y)$\;
        \If{$k = 0$}{
            \Return{$\cA$}\;
        }
        \If{$k > |X|$}{
            \Return{the breakpoint representation of a $w$-optimal alignment computed using \cref{lm:baseline-wed}}\;
        }
        Denote $(x_i, y_i)_{i = 0}^t \coloneqq \cA$\;
        Pick the largest $m\in \fragmentco{0}{t}$ such that $\wed_{\cA}(X\fragmentco{0}{x_m},Y\fragmentco{0}{y_m})\le k/2$\;
        $\ell \gets \min\{i \in \fragmentcc{0}{m} \mid \sed(X\fragmentco{x_i}{x_m}) \le 5k\}$\;
        $r \gets \max\{i \in \fragmentcc{m+1}{t} \mid \sed(X\fragmentco{x_{m+1}}{x_i}) \le 5k\}$\;
        $\cA_{L} \gets (x_i,y_i)_{i=0}^m $\;
        $\cA_{R} \gets (x_i,y_i)_{i={m+1}}^t$\;
        $\cB_{L} \gets \upgr(X\fragmentco{0}{x_m}, Y\fragmentco{0}{y_m}, w, \cA_L)$\;
        $\cB_R \gets \upgr(X\fragmentco{x_{m+1}}{|X|}, Y\fragmentco{y_{m+1}}{|Y|}, w, \cA_R)$\;
        Compute breakpoint representation of a $w$-optimal alignment $\cB_M: X\fragmentco{x_\ell}{x_r}\onto Y\fragmentco{y_\ell}{y_r}$ using \cref{lm:small-sed-algo} for threshold $10k+2$\;
        Compute the leftmost intersection points $(x^*_L,y^*_L)\in \cB_L\cap\cB_M$ and $(x^*_R,y^*_R)\in \cB_R\cap\cB_M$\;
        \Return{the breakpoint representation of an alignment $\cB:X\onto Y$ that follows $\cB_L$ from $(0,0)$ to $(x^*_L,y^*_L)$, follows $\cB_M$ from $(x^*_L,y^*_L)$ to $(x^*_R,y^*_R)$, and follows $\cB_R$ from $(x^*_R,y^*_R)$ to $(|X|,|Y|)$}\;
    }
    \end{algorithm}

    We first find $k \coloneqq \wed_{\cA}(X, Y)$.
    Then, we handle some corner cases: if $k = 0$, we return $\cA$, and if $k > |X|$, we find a $w$-optimal alignment using \cref{lm:baseline-wed}.
    If none of the corner cases are applicable, we proceed as follows.
    We denote $(x_i, y_i)_{i = 0}^t \coloneqq \cA$ and find three indices $\ell,m,r\in \fragmentcc{0}{t}$ such that $m$ is the largest possible such that $\wed_{\cA}(X\fragmentco{0}{x_m},Y\fragmentco{0}{y_m})\le k/2$, $\ell$ is the smallest possible such that $\sed(X\fragmentco{x_\ell}{x_m}) \le 5k$, and $r$ is the largest possible such that $\sed(X\fragmentco{x_{m+1}}{x_r}) \le 5k$.
    We recursively compute the breakpoint representation of a $w$-optimal alignment $\cB_{L} : X\fragmentco{0}{x_m}\onto Y\fragmentco{0}{y_m}$ by improving upon  $\cA_{L} \coloneqq (x_i,y_i)_{i=0}^m$ and recursively compute the breakpoint representation of a $w$-optimal alignment $\cB_R:X\fragmentco{x_{m+1}}{|X|}, Y\fragmentco{y_{m+1}}{|Y|}$ by improving upon $\cA_R \coloneqq (x_i,y_i)_{i=m+1}^t$.
    We then use \cref{lm:small-sed-algo} to find the breakpoint representation of a $w$-optimal alignment $\cB_M: X\fragmentco{x_\ell}{x_r}\onto Y\fragmentco{y_\ell}{y_r}$.
    We have $\sed(X\fragmentco{x_\ell}{x_r})\le \sed(X\fragmentco{x_\ell}{x_m})+\sed(X\fragmentco{x_m}{x_{m+1}})+\sed(X\fragmentco{x_{m+1}}{x_r}) \le 5k+2+5k=10k+2$ and $\wed(X\fragmentco{x_\ell}{x_r},Y\fragmentco{y_\ell}{y_r})\le \wed_{\cA}(X\fragmentco{x_\ell}{x_r},Y\fragmentco{y_\ell}{y_r})\le \wed_\cA(X,Y)\le k$, so it suffices to apply \cref{lm:small-sed-algo} with threshold $10k+2$. 
    We then compute the $w$-optimal alignment $\cB$ by combining $\cB_{L}, \cB_M$, and $\cB_R$.
    We use two pointers to find the leftmost intersection point of $\cB_L$ and $\cB_M$ and the leftmost intersection point of $\cB_M$ and $\cB_R$.
    We then create $\cB$ by first going along $\cB_{L}$ until the first intersection point with $\cB_M$, then switching to $\cB_M$ and going along it until the first intersection point with $\cB_R$, and finally switching to $\cB_R$ and going along it from there.

    Let us prove the correctness of the algorithm. 
    If $k = 0$, then $\cA$ is already a $w$-optimal alignment as no alignment of negative weight can exist.
    If $k > |X|$, we use \cref{lm:baseline-wed}, which returns a $w$-optimal alignment.
    Otherwise, the minimality of $\ell$ implies that $\ell=0$ or $\sed(X\fragmentco{x_\ell}{x_m})=5k>4\wed_{\cA}(X,Y)$, whereas the maximality of $r$ implies that $r=t$ or $\sed(X\fragmentco{x_{m+1}}{x_r})=5k>4\wed_{\cA}(X,Y)$, where equalities are due to the continuity property of $\sed$ (see \cref{fct:selfed-properties}).
    Consequently, we can use \cref{cor:intersect} for the alignments $\cA$, $\cB_L$, $\cB_M$, and $\cB_R$ to conclude that  $\cB_L\cap \cB_M\ne \emptyset \ne \cB_R\cap \cB_M$, and all intersection points $(x^*_L,y^*_L)\in \cB_L\cap \cB_M$ and $(x^*_R,y^*_R)\in \cB_R\cap \cB_M$  satisfy
    \[\wed\!(X,Y)\!=\! \wed_{\cB_L}\!(X\fragmentco{0}{x^*_L},Y\fragmentco{0}{y^*_L})+\wed_{\cB_M}\!(X\fragmentco{x^*_L}{x^*_R}, Y\fragmentco{y^*_L}{y^*_R})+\wed_{\cB_R}\!(X\fragmentco{x^*_R}{|X|},Y\fragmentco{y^*_R}{|Y|}).\]
    Consequently, the constructed alignment $\cB$ is indeed $w$-optimal.

    We now analyze the running time.
    First, focus on a single execution of \cref{alg:upgrade_alignment}, ignoring the time spent in the recursive calls.
    Denote $n \coloneqq |X| + |Y|$ and $k \coloneqq \wed_{\cA}(X, Y)$.
    Retrieving $k$ from the breakpoint representation of $\cA$ takes $\Oh(k + 1)$ time.
    If $k = 0$, we then return $\cA$ in constant time.
    If $k > |X|$, we run \cref{lm:baseline-wed} that works in time $\Oh((|X| + 1)k) = \Oh(k^2)$ since $|X| < k$.
    We henceforth assume that we are not in the corner cases.
    It takes $\Oh(k)$ time to find $m$.
    It then takes $\Oh(k^2)$ time and \modelname{} operations to find $\ell$ and $r$ using \cref{lm:selfed2}.
    Constructing $\cB_M$ using \cref{lm:small-sed-algo} takes $\Oh(\min\{k^2 \cdot W \log^2 n, k^{2.5} \log^2 n\})$ time and $\Oh(k^2)$ \modelname{} operations.
    Combining the three $w$-optimal alignments takes time $\Oh(k)$ since the cost of each of them does not exceed the cost of $\cA$.
    Thus, a single execution of \cref{alg:upgrade_alignment} takes time $\Oh(\min\{k^2 \cdot W \log^2 n, k^{2.5} \log^2 n\})$ plus $\Oh(k^2)$ \modelname{} operations.

    Now, let us consider the overall running time of the algorithm, taking into account the recursive calls.
    We recurse into two subproblems $(X_L,Y_L,w,\cA_L)$ and $(X_R,Y_R,w,\cA_R)$ such that $\wed_{\cA_{L}}(X_L, Y_L)\le \frac12 \wed_{\cA}(X, Y)$ and $\wed_{\cA_R}(X_R, Y_R) < \frac12 \wed_{\cA}(X, Y)$; the latter is true because the maximality of $m$ implies $\wed_{\cA}(X\fragmentco{0}{x_{m+1}},Y\fragmentco{0}{y_{m+1}})>\frac12 \wed_{\cA}(X, Y)$.
    As a single execution of \cref{alg:upgrade_alignment} takes time $\Oh(\min\{k^2 \cdot W \log^2 n, k^{2.5} \log^2 n\})$ plus $\Oh(k^2)$ \modelname{} operations, there exists a constant $\alpha$ such that it takes time at most $\alpha \cdot \min\{k^2 \cdot W \log^2 n, k^{2.5} \log^2 n\}$ plus at most $\alpha \cdot k^2$ \modelname{} operations.
    We prove by induction that the whole algorithm takes time at most $\beta \cdot \min\{k^2 \cdot W \log^2 n, k^{2.5} \log^2 n\}$ for $\beta \coloneqq 2 \alpha$.
    The base case, in which no further recursive calls are made is trivial.
    Otherwise, we spend $\alpha \cdot \min\{k^2 \cdot W \log^2 n, k^{2.5} \log^2 n\}$ time plus two recursive calls, in which the parameter $k$ decreases by at least a factor of two.
    Therefore, using the induction hypothesis, the overall running time is
    \begin{align*}
        &\alpha \cdot \min\{k^2 \cdot W \log^2 n, k^{2.5} \log^2 n\} + 2 \cdot \beta \cdot \min\{(k / 2)^2 \cdot W \log^2 n, (k / 2)^{2.5} \log^2 n\}\\
        &\le (\alpha + \beta / 2) \cdot \min\{k^2 \cdot W \log^2 n, k^{2.5} \log^2 n\}\\
        &= \beta \cdot \min\{k^2 \cdot W \log^2 n, k^{2.5} \log^2 n\}.
    \end{align*}
    Analogously, one can show that the total number of \modelname{} operations is limited by $\beta \cdot k^2$.
\end{proof}

We now use \cref{lm:upgrade_alignment} to design our main static algorithm.
It starts from an optimal unweighted alignment and gradually improves it until we arrive at a $w$-optimal alignment.

\begin{theorem}\label{lm:full_algorithm}
    There is a \modelname{} algorithm that, given two strings $X, Y \in \Sigma^{\le n}$ as well as oracle access to a weight function $w : \Esigma^2 \to [0 \dd W]$, finds $\wed(X, Y)$ as well as the breakpoint representation of a $w$-optimal alignment $\cA : X \onto Y$.
    The running time of the algorithm is $\Oh(\min\{k^2 \cdot W \log^2 n, k^{2.5} \log^3 n\})$ time plus $\Oh(k^2 \log \min \{n, W + 1\})$ \modelname{} operations, where $k = \wed(X, Y)$.
\end{theorem}

\begin{proof}
    We find the breakpoint representation of a $w$-optimal alignment $\cA : X\onto Y$.
    From it, in time $\Oh(k)$, we can compute $\wed(X, Y)$.

    Define $w' : \Esigma^2 \to [0 \dd W']$ as $\myw{w'}{a}{b} \coloneqq \min\{\w{a}{b}, n\}$ for $W' = \min\{W, n\}$.
    Note that $k' \coloneqq \ed^{w'}(X, Y) \le \wed(X, Y) = k$.
    We run the remaining part of the algorithm for the weight function $w'$ instead of $w$.
    At the end, we get some $w'$-optimal alignment $\cA : X \onto Y$ of cost $k'$.
    If $k' \ge n$, we additionally run the algorithm from \cref{lm:baseline-wed} for $X$, $Y$, and $w$ and return the alignment it returns.
    Its running time is $\Oh(k \cdot n) \le \Oh(k^2)$.
    Otherwise, if $k' < n$, we return $\cA$.
    Note that in this case all edges of the alignment $\cA$ have weights at most $k' < n$, and thus $\wed_{\cA}(X, Y) = \ed^{w'}_{\cA}(X, Y) = k' \le k = \wed(X, Y)$.
    Hence, $\cA$ is a $w$-optimal alignment.

    Therefore, it remains to find a $w'$-optimal alignment $\cA$.

    First, we find an optimal unweighted alignment $\cA^* : X \onto Y$ of weight $k^* \le k'$ in time $\Oh(k'^2)$ and $\Oh(k'^2)$ \modelname{} operations using \cref{lm:k2-ed}.
    Let $W^* \le W'$ be the maximum weight $w'$ of the edges in $\cA^*$.

    We now iterate from $t = t_1  \coloneqq \left\lceil \log_2 W^* \right\rceil$ down to $t = 0$, and for each such $t$ compute a $w'_t$-optimal alignment $\cA_t : X \onto Y$, where $\myw{w'_t}{a}{b} \coloneqq \left\lceil \myw{w'}{a}{b} / 2^t \right\rceil$ for any $a, b \in \Esigma$.
    Let $W'_t \coloneqq \left\lceil W' / 2^t \right\rceil$ be an upper bound on the values of $w'_t$.
    Note that $w'_0 = w'$, so at the end we get a $w'$-optimal alignment $\cA = \cA_0$.
    Furthermore, $\myw{w'_t}{a}{b} \le \myw{w'}{a}{b}$ for any $t \in [0 \dd t_1]$ and $a, b \in \Esigma$, and thus $k'_t \coloneqq \ed^{w'_t}(X, Y) \le \ed^{w'}(X, Y) = k' \le k$ for any $t$.

    In the first iteration, we set $\cA_{t_1} \coloneqq \cA^*$.
    Note that all edges in $\cA^*$ have $w'$-weight at most $W^*$ by the definition of $W^*$.
    Hence, they have $w'_{t_1}$-weight at most one as $2^{t_1} \ge W^*$.
    Therefore, $\ed_{\cA^*}^{w'_{t_1}}(X, Y) = \ed_{\cA^*}(X, Y) = \ed(X, Y) \le \ed^{w'_{t_1}}(X, Y)$.
    Thus, $\cA_{t_1} = \cA^*$ is a $w'_{t_1}$-optimal alignment.

    Now, for every other $t$, given a $w'_{t + 1}$-optimal alignment $\cA_{t + 1} : X \onto Y$, we apply \cref{lm:upgrade_alignment} to find a $w'_t$-optimal alignment $\cA_{t} : X \onto Y$.
    It takes $\Oh(\min\{\ell^2 \cdot W'_t \log^2 n, \ell^{2.5} \log^2 n\})$ time plus $\Oh(\ell^2)$ \modelname{} operations, where $\ell \coloneqq \ed_{\cA_{t + 1}}^{w'_t}(X, Y)$.
    Note that $\myw{w'_t}{a}{b} \le 2 \cdot \myw{w'_{t+1}}{a}{b}$ for any $a, b \in \Esigma$, so $\ell \le 2 \ed_{\cA_{t + 1}}^{w'_{t + 1}}(X, Y) = 2 k'_{t + 1} \le 2k'$.
    Hence, the call to \cref{lm:upgrade_alignment} takes time $\Oh(\min\{k'^2 \cdot W'_t \log^2 n, k'^{2.5} \log^2 n\})$ and $\Oh(k'^2)$ \modelname{} operations.

    There are $t_1 + 1 = \Oh(\log (W' + 1)) = \Oh(\log \min\{n, W + 1\})$ iterations of the algorithm, so in total we spend $\Oh(k'^2 \log \min\{n, W + 1\}) \le \Oh(k^2 \log \min \{n, W + 1\})$ \modelname{} operations.
    Furthermore, the overall time complexity is

    \begin{align*}
        \Oh(k'^2 + \sum_{t = 0}^{t_1} \min\{k'^2 \cdot W'_t \log^2 n, k'^{2.5} \log^2 n\}) &\le \Oh(\min\{\sum_{t = 0}^{t_1} k'^2 \cdot W'_t \log^2 n, \sum_{t = 0}^{t_1} k'^{2.5} \log^2 n\}) \\
                                        &= \Oh(\min\{k'^2 \log^2 n \sum_{t = 0}^{t_1} W'_t, (t_1 + 1) \cdot k'^{2.5} \log^2 n\})\\
                                        &\le \Oh(\min\{k'^2 W' \log^2 n, k'^{2.5} \log^3 n\})\\
                                        &\le \Oh(\min\{k^2 W \log^2 n, k^{2.5} \log^3 n\}),
    \end{align*}

    where $\sum_{t = 0}^{t_1} W'_t = \Oh(W' + t_1) = \Oh(W') \le \Oh(W)$ as $W'_t = \left\lceil W' / 2^t \right\rceil$ for all $t \in \fragmentcc{0}{t_1}$.
\end{proof}

As \modelname{} operations in the standard model can be implemented to work in constant time after a linear-time preprocessing \cite[Section 7.1]{CKW20}, this algorithm works in time $\Oh(n + \min\{k^2 \cdot W \log^2 n, k^{2.5} \log^3 n\})$ in the standard model.

\section{Dynamic Algorithm} \label{sec:dynamic}

In this section, we maintain the weighted edit distance of two strings $X$ and $Y$ dynamically under edits.
That is, in one update we insert, delete, or substitute a single character in $X$ or $Y$.

\subsection{Tracking a Sequence of Updates}

Before describing our dynamic algorithms, we first build some tools that we will use to track the updates applied to $X$ and $Y$.
Suppose that the $i$-th update was applied to $X$.
We interpret this update as a pair $(\cB_i, \cC_i)$ of two alignments, where $\cB_i : X \onto X'$ is an alignment of the old version of $X$ onto the new one induced by the edit, and $\cC_i : Y \onto Y'$ where $Y' = Y$ is an identity alignment.
If the $i$-th update is applied to $Y$, we define $(\cB_i, \cC_i)$ symmetrically.
In both situations we have $\ed_{\cB_i}(X, X') + \ed_{\cC_i}(Y, Y') = 1$.
Furthermore, given an update, we can trivially calculate the breakpoint representations of $\cB_i$ and $\cC_i$ in constant time.

To track how a string changes under a sequence of updates, we use the following fact that allows to compose alignments of subsequent versions of the string.

\begin{fact}[{\cite[Fact 2.5]{DGHKS23}}]\label{fct:triangle}
    Consider strings $X,Y,Z\in \Sigma^*$ as well as alignments $\cA : X\onto Y$ and $\cB : Y\onto Z$.
    There exists a \emph{composition alignment} $\cB \circ \cA : X \onto Z$ satisfying the following properties for all $x\in \fragmentco{0}{|X|}$ and $z\in\fragmentco{0}{|Z|}$:
    \begin{itemize}
        \item $\cB\circ \cA$ aligns $X\position{x}$ to $Z\position{z}$ if and only if there exists $y\in \fragmentco{0}{|Y|}$ such that $\cA$ aligns $X\position{x}$ to $Y\position{y}$ and $\cB$ aligns $Y\position{y}$ to $Z\position{z}$.
        \item $\cB\circ \cA$ deletes $X\position{x}$ if and only if $\cA$ deletes $X\position{x}$ or there exists $y\in \fragmentco{0}{|Y|}$ such that $\cA$ aligns $X\position{x}$ to $Y\position{y}$ and $\cB$ deletes $Y\position{y}$.
        \item $\cB\circ \cA$ inserts $Z\position{z}$ if and only if $\cB$ inserts $Z\position{z}$ or there exists $y\in \fragmentco{0}{|Y|}$ such that $\cA$ inserts $Y\position{y}$ and $\cB$ aligns $Y\position{y}$ to $Z\position{z}$.
    \end{itemize}

    If a weight function $w$ satisfies the \emph{triangle inequality}, that is, $\w{a}{b}\le \w{a}{c}+\w{c}{b}$ holds for all $a,b,c\in \Esigma$,  then $\wed_{\cB \circ \cA}(X, Z) \leq \wed_{\cA}(X, Y) + \wed_{\cB}(Y, Z)$.
\lipicsEnd
\end{fact}

The following corollary shows how to efficiently compute the composition alignment.

\begin{corollary}\label{cor:alignment-composition-algorithm}
    There is an algorithm that given strings $X, Y, Z \in \Sigma^*$ as well as the breakpoint representations of alignments $\cA : X \onto Y$ and $\cB : Y \onto Z$, in time $\Oh(\ed_{\cA}(X, Y) + \ed_{\cB}(Y, Z) + 1)$ builds the breakpoint representation of $\cB \circ \cA$.
\end{corollary}

\begin{proof}
    Consider some character $X\position{x}$.
    If $\cA$ matches $X\position{x}$ to some $Y\position{y}$, and $\cB$ matches $Y\position{y}$ to some $Z\position{z}$, then $X\position{x} = Y\position{y} = Z\position{z}$, and $\cC \coloneqq \cB \circ \cA$ matches $X\position{x}$ to $Z\position{z}$.
    Hence, if some vertex $(x, z)$ of $\AG(X, Z)$ is a part of the breakpoint representation of $\cC$, then either $(x, y)$ for some $y$ is a part of the breakpoint representation of $\cA$, or $\cA$ matches $X\position{x}$ to some $Y\position{y}$, and $(y, z)$ is a part of the breakpoint representation of $\cB$.
    Therefore, by scanning the breakpoint representations of $\cA$ and $\cB$ using two pointers, we can construct the breakpoint representation of $\cC$ in time $\Oh(\ed_{\cA}(X, Y) + \ed_{\cB}(Y, Z) + 1)$.
\end{proof}

Even though composition alignment does not necessarily satisfy triangle inequality, we can still formulate a fact similar to the second part of \cref{fct:triangle} in the general case.

\begin{lemma}\label{lem:triw}
    Consider strings $X,Y,Z\in \Sigma^*$, alignments $\cA : X \onto Y$ and $\cB : Y\onto Z$, and a weight function $w:\Esigma^2 \to \RR_{\ge 0}$ such that $\w{a}{b} \le W$ for all $a, b \in \Esigma$ for some value $W$. 
    The composition alignment $\cC=\cB\circ \cA$ satisfies $\wed_{\cC}(X,Z)\le \wed_{\cA}(X,Y)+W\cdot \ed_{\cB}(Y,Z)$ and $\wed_{\cC}(X,Z)\le W \cdot \ed_{\cA}(X,Y) + \wed_{\cB}(Y,Z)$.
\end{lemma}

\begin{proof}
    We first claim that triangle inequality of \cref{fct:triangle} works in a more general case.
    That is, given alignments $\cA : X\onto Y$ and $\cB : Y\onto Z$ and three weight functions $w_1, w_2$, and $w_3$ such that $\wi{1}{a}{b}\le \wi{2}{a}{c}+\wi{3}{c}{b}$ holds for all $a,b,c\in \Esigma$, we have $\ed^{w_1}_{\cB \circ \cA}(X, Z) \leq \ed^{w_2}_{\cA}(X, Y) + \ed^{w_3}_{\cB}(Y, Z)$.
    It follows from the properties of $\circ$ described in \cref{fct:triangle}.
    Any edit $a \mapsto b$ of $\cB \circ \cA$ for $a \in X \cup \{\varepsilon\}$ and $b \in Z \cup \{\varepsilon\}$ can be decomposed into an edit $a \mapsto c$ of $\cA$ and an edit $c \mapsto b$ of $\cB$ for some $c \in Y \cup \{\varepsilon\}$.
    By the triangle inequality, we have $\wi{1}{a}{b}\le \wi{2}{a}{c}+\wi{3}{c}{b}$.
    Furthermore, different edits of $\cB \circ \cA$ correspond to different edits of $\cA$ and $\cB$.
    Therefore, by summing up these inequalities over all edits, we obtain $\ed^{w_1}_{\cB \circ \cA}(X, Z) \leq \ed^{w_2}_{\cA}(X, Y) + \ed^{w_3}_{\cB}(Y, Z)$.

    \medskip

    We now use this fact to prove the lemma.
    Define $w_1 \coloneqq w$, $w_2 \coloneqq w$, and $w_3(a, a) = 0$ for any $a \in \Esigma$ and $w_3(a, b) = W$ for any $a \neq b \in \Esigma$.
    We claim that $\wi{1}{a}{b}\le \wi{2}{a}{c}+\wi{3}{c}{b}$ holds for all $a,b,c\in \Esigma$.
    If $c \neq b$, the claim holds as $\wi{1}{a}{b} = \w{a}{b} \le W = \wi{3}{c}{b} \le \wi{2}{a}{c} + \wi{3}{c}{b}$.
    On the other hand, if $c = b$, the claim holds as $\wi{1}{a}{b} = \w{a}{b} = \wi{2}{a}{b} = \wi{2}{a}{c} = \wi{2}{a}{c}+\wi{3}{c}{b}$.

    By the extension of \cref{fct:triangle} for three different weight functions, we obtain
    $\wed_{\cB \circ \cA}(X, Z) = \ed^{w_1}_{\cB \circ \cA}(X, Z) \leq \ed^{w_2}_{\cA}(X, Y) + \ed^{w_3}_{\cB}(Y, Z) = \wed_{\cA}(X, Y) + W \cdot \ed_{\cB}(Y, Z)$, thus proving the first inequality from the lemma statement.

    The second inequality can be obtained analogously by swapping $w_2$ and $w_3$.
\end{proof}

\begin{corollary}\label{lem:triw-cor}
    Consider strings $X,Y,X',Y' \in \Sigma^*$, alignments $\cA : X \onto Y$, $\cB : X \onto X'$, and $\cC : Y \onto Y'$, and a weight function $w : \Esigma^2 \to \RR_{\ge 0}$ such that $\w{a}{b} \le W$ for all $a, b \in \Esigma$ for some value $W$.
    The composition alignment $\cD = \cC \circ (\cA \circ \cB^{-1})$ satisfies $\wed_{\cD}(X', Y') \le \wed_{\cA}(X, Y) + W \cdot (\ed_{\cB}(X, X') + \ed_{\cC}(Y, Y'))$.
\end{corollary}

\begin{proof}
    We first consider the alignment $\cD' \coloneqq \cA \circ \cB^{-1}$.
    \cref{lem:triw} implies that $\wed_{\cD'}(X', Y) \le W \cdot \ed_{\cB^{-1}}(X', X) + \wed_{\cA}(X, Y) = W \cdot \ed_{\cB}(X, X') + \wed_{\cA}(X, Y)$.
    As $\cD = \cC \circ \cD'$, \cref{lem:triw} implies that $\wed_{\cD}(X', Y') \le \wed_{\cD'}(X', Y) + W \cdot \ed_{\cC}(Y, Y') \le W \cdot \ed_{\cB}(X, X') + \wed_{\cA}(X, Y) + W \cdot \ed_{\cC}(Y, Y')$, thus proving the claim.
\end{proof}

\subsection{The Case of Small Self-Edit Distance}\label{sec:dyn-smallsed}

We now present our dynamic algorithm.
We follow a path similar to the one described in \cref{sec:full-algo} for the static case.
First, we interpret \cref{prp:hadstwo} as a dynamic algorithm that works in time $\tOh((|X| + |Y|) \cdot W)$\footnote{Throughout \cref{sec:dynamic} we assume that the value $W$ is given to the algorithm.} per update to obtain a dynamic algorithm for strings of small self-edit distance that works in time $\tOh(Wk)$ per update for a lifetime of $\floor{k / W}$ updates with $\tOh(W k^2)$-time initialization.
Essentially, to design this algorithm, we have to make \cref{lm:small-sed-algo} dynamic.

\begin{lemma}\label{lem:dyn_sed}
    There exists a dynamic algorithm that, given positive integers $n \ge 2$ and $k \in \fragmentcc{1}{n^2}$, oracle access to a weight function $w:\Esigma^2\to \fragmentcc{0}{W}$, and strings $X,Y\in \Sigma^*$ that initially satisfy $|X|, |Y| \le n$, $\sed(X)\le k$, and $\wed(X,Y)\le k$, maintains the breakpoint representation of a $w$-optimal alignment $\cA:X\onto Y$ subject to $\floor{k/W}$ edits in $X$ and $Y$ in the algorithm lifetime.
The algorithm takes $\Oh(W \cdot k^2\log^2 n)$ time plus $\Oh(k^2)$ \modelname{} operations for initialization and $\Oh(W\cdot k\log^2 n)$ time per update.
\end{lemma}
\begin{proof}
Let $\hX$ and $\hY$ be strings given to the algorithm at the initialization phase.
If $|\hX| \le 3k$ or $|\hY| \le 3k$, then $|\hX|+|\hY|\le 7k$ since $\wed(X,Y)\le k$.
In this case, we simply use \cref{prp:hadstwo} as a dynamic algorithm, which takes $\Oh(W\cdot k^2\log n)$ time for initialization and $\Oh(W \cdot k \log^2 n)$ time per update as at any point in the lifetime of the algorithm we have $|X| + |Y| \le 8k$ as at most $\floor{k / W}$ edits were applied.

Thus, we henceforth assume $3k \le |\hX|, |\hY| \le n$.
Note that in such a case at every point in the lifetime of the algorithm, $X$ and $Y$ are non-empty as at most $\floor{k / W}$ edits were applied.
Furthermore, $|X| + |Y| \le 2n + k / W \le 3n$ throughout the lifetime of the algorithm.
We proceed as in the proof of \cref{lm:small-sed-algo} but with the threshold tripled.
We use the algorithm of \cref{lem:decomp2} for strings $\hX$ and $\hY$ and threshold $\hat{k}=3k$, arriving at a decomposition $\hX=\bigodot_{i=0}^{m-1} \hX_i$ into fragments $\hX_i=\hX\fragmentco{\hx_{i}}{\hx_{i+1}}$ of length $|\hX_i|\in \fragmentco{3k}{6k}$
and a sequence of fragments $\hY_i = \hY\fragmentco{\hy_i}{\hy'_{i+1}}$, where $\hy_i \coloneqq \max\{0, \hx_i - 3k\}$ and $\hy'_{i + 1} \coloneqq \min\{|\hY|, \hx_{i+1} + 9k\}$ for $i\in \fragmentco{0}{m}$.
\Cref{lem:decomp2} also yields a set $F \subseteq \fragmentco{0}{m}$ of size $\Oh(k)$ such that $\hX_i=\hX_{i-1}$ and $\hY_i=\hY_{i-1}$ holds for all $i\in \fragmentco{0}{m} \setminus F$.
To easily handle corner cases, we assume that $\fragmentco{0}{m} \setminus F \subseteq \fragmentcc{2}{m - 5}$;
if the original set $F$ does not satisfy this condition, we add the missing $\Oh(1)$ elements.
Finally, an application of \cref{lem:addtoboxds} lets us build the hierarchical alignment data structures $\boxds(\hX_i, \hY_i)$ for all $i \in F$.

We maintain the breakpoint representations of alignments $\cB : \hX\onto X$ and $\cC:\hY\onto Y$ representing the changes made to the strings since the initialization of the algorithm.
For each $i\in  \fragmentco{0}{m}$, we define fragments $X_i=X\fragmentco{x_i}{x_{i+1}}\coloneqq \cB(\hX_i)$ and $Y_i = Y\fragmentco{y_i}{y'_{i+1}}\coloneqq \cC(\hY_i)$ and consider the subgraph $G_i$ of $\oAGw(X, Y)$ induced by $\fragmentcc{x_i}{x_{i+1}} \times \fragmentcc{y_i}{y'_{i + 1}}$ (see \cref{fig:substitutions-only}).
Denote by $G$ the union of all subgraphs $G_i$.

\begin{claim}
    The distance from $(0,0)$ to $(|X|,|Y|)$ in $G$ equals $\wed(X,Y)$.
\end{claim}
\begin{claimproof}
    Since $G$ is a subgraph of $\oAGw(X,Y)$, it suffices to prove that $G$ contains a path corresponding to a $w$-optimal alignment $\cA : X \onto Y$.
    For this, we need to show that, for every $i\in \fragmentco{0}{m}$, the fragment $Y\fragmentco{\by_i}{\by'_{i+1}}\coloneqq \cA(X_i)$ is contained within $Y_i$.

    Denote $d_X = \ed_{\cB}(\hX,X)$, $d_Y=\ed_{\cC}(\hY,Y)$, and $\cD = \cC \circ (\hcA \circ \cB^{-1})$ for a $w$-optimal alignment $\hcA : \hX \onto \hY$. Since $\wed_{\hcA}(\hX,\hY)=\wed(\hX,\hY)\le k$ and $d_X+d_Y\le \floor{k/W}$, \cref{lem:triw-cor} yields $\wed_{\cA}(X,Y)=\wed(X,Y) \le \wed_{\cD}(X, Y) \le \wed_{\hcA}(\hX,\hY) + W \cdot (d_X + d_Y) \le k + W \cdot \floor{k/W} \le 2k$.
    Consequently, $x_i \le \by_i+2k$ and $x_{i+1}\ge \by'_{i+1}-2k$ due to \cref{lm:paths_dont_deviate_too_much}.
    Moreover, $\hx_i\le x_i+d_X\le \by_i+2k+d_X$ and $\hx_{i+1}\ge x_{i+1}-d_X \ge \by'_{i+1}-2k-d_X$.
    The definitions of $\hy_i$ and $\hy'_{i+1}$ yield $\hy_i = \max\{\hx_i-3k,0\}\le \max\{\by_i-k+d_X,0\}$ and $\hy'_{i+1}\ge \min\{\hx_{i+1}+3k,|\hY|\}\ge \min\{\by'_{i+1}+k-d_X,|\hY|\}$.
    If $\hy_i=0$, then $y_i=0\le \by_i$; otherwise, $y_i \le \hy_i+d_Y \le \by_i-k+d_X+d_Y\le \by_i$.
    If $\hy'_{i+1}=|\hY|$, then $y'_{i+1}=|Y|\ge \by'_{i+1}$; otherwise, $y'_{i+1}\ge \hy_{i+1}-d_Y \ge \by'_{i+1}+k-d_X-d_Y \ge \by'_{i+1}$. 
    In all cases, $y_i \le \by_i$ and $y'_{i+1}\ge \by'_{i+1}$, so $Y\fragmentco{\by_i}{\by'_{i+1}}$ is indeed contained within $Y_i=Y\fragmentco{y_i}{y'_{i+1}}$.
\end{claimproof}

For each $i \in \fragmentcc{0}{m}$, denote $V_i \coloneqq \{(x, y) \in V(G) \mid x = x_i, y \in \fragmentcc{y_i}{y'_i}\}$,
where $y_{m}=|Y|$ and $y'_0=0$.
Moreover, for $i, j \in \fragmentcc{0}{m}$ with $i \le j$, let $D_{i, j}$ denote the matrix of pairwise distances from $V_i$ to $V_j$ in $G$, where rows of $D_{i, j}$ represent vertices of $V_i$ in the decreasing order of the second coordinate, and columns of $D_{i, j}$ represent vertices of $V_j$ in the decreasing order of the second coordinate.
Note that $D_{i, j}$ is a Monge matrix due to \cref{fct:monge}.
Since $V_0 = \{(0,0)\}$ and $V_{m}=\{(|X|,|Y|)\}$, the only entry in $D_{0,m}$ stores $\wed(X,Y)$.

Note that, for all $a,i,b \in \fragmentcc{0}{m}$ with $a\le i \le b$, the set $V_i$ is a separator in $G$ between all vertices of $V_{a}$ and $V_b$.
Thus, every path from $V_{a}$ to $V_b$ in $G$ must cross $V_i$, and hence $D_{a, b} = D_{a, i} \otimes D_{i, b}$.
In particular, $D_{0,m}=\bigotimes_{i=0}^{m-1} D_{i,i+1}$.

We maintain a perfectly balanced binary tree $T$ with leaves indexed by $\fragmentco{0}{m}$ from left to right;
we make sure that, for every node $\nu$, the interval $\fragmentco{\ell_\nu}{r_\nu}$ of leaves in the subtree of $\nu$ is dyadic, 
that is, there exists an integer $j\ge 0$, called the \emph{level} of the node $\nu$, such that $\ell_\nu$ is an integer multiple of $2^j$ and $r_\nu=\min\{\ell_\nu+2^j, m\}$.
For every node $\nu$, we maintain $\mds(D_{\ell_\nu,r_\nu})$, where $\fragmentco{\ell_\nu}{r_\nu}$ is the interval of leaves in the subtree of $\nu$, as well as the values $x_{r_\nu}-x_{\ell_\nu}$, $y_{r_\nu}-y_{\ell_\nu}$, and $y'_{r_\nu}-y'_{\ell_\nu}$.
Additionally, each leaf with $\fragmentco{\ell_\nu}{r_\nu}=\{i\}$ stores a pointer to $\boxds(X_i, Y_i)$.
We emphasize that identical subtrees may share the same memory location.

\subparagraph{Initialization.}
Upon initialization, we use \cref{lem:decomp2} to compute the decomposition $\hX= \bigodot_{i=0}^{m-1} \hX_i$
and a set $F\subseteq\fragmentco{0}{m}$ of size $\Oh(k)$ such $\hX_i=\hX_{i-1}$ and $\hY_i=\hY_{i-1}$ for $i\in \fragmentco{0}{m}\setminus F$.
This set is returned along with the endpoints of $\hX_i$ for each $i\in F$; the endpoints of $\hY_i$ can be retrieved by their definition from the endpoints of $\hX_i$. 
We also apply \cref{lem:addtoboxds} to compute $\boxds(X_i,Y_i)$ for all $i\in F$.
Furthermore, we initialize $\cB : \hX \onto \hX$ and $\cC : \hY \onto \hY$ as identity alignments.

It remains to initialize the tree $T$.
We say that node $\nu$ is \emph{primary} if $F\cap \fragmentoo{2\ell_\nu-r_\nu}{r_\nu}\ne \emptyset$; this condition is equivalent to $F\cap \fragmentoo{\ell_\nu-2^j}{\ell_\nu+2^j}\ne \emptyset$, where $j$ is the level of the node.
Observe that if $\nu$ is not primary, then the subtree rooted at $\nu$ is identical to the subtree rooted at the preceding node of the same level, i.e., the node $\mu$ with $\fragmentco{\ell_{\mu}}{r_{\mu}}=\fragmentco{\ell_{\nu}-2^j}{\ell_\nu}$.
It follows from the fact that for two consecutive elements $i, i' \in F \cup \{m\}$, we have $D_{i, i + 1} = D_{i + 1, i + 2} = \cdots = D_{i' - 1, i'}$ analogously to \cref{clm:D-formula}.
We traverse primary nodes of $T$ in the post-order, maintaining a stack consisting of the most recently visited node at each level. 
Upon entering any node $\nu$, we first check if $\nu$ is primary and, if not, simply return a pointer to the most recently visited node at the same level.
If $\nu$ is a primary leaf with label $i$, we must have $i\in F$.
Thus, we create a new node storing pointers to $\boxds(X_i, Y_i)$ and $\mds(D_{i,i+1})$, retrieved using $\boxds(X_i, Y_i)$ and \cref{lm:ds_tools} as $D_{i, i + 1}$ is a contiguous submatrix of $\BMw(X_i, Y_i)$.
If $\nu$ is a primary internal node, we recurse on the two children $\lambda$ and $\rho$. 
Since $D_{\ell_\nu,r_\nu}=D_{\ell_\lambda,r_\lambda}\otimes D_{\ell_{\rho},r_{\rho}}$, we can construct $\mds(D_{\ell_\nu,r_\nu})$ from $\mds(D_{\ell_\lambda,r_\lambda})$ and $\mds(D_{\ell_\rho,r_\rho})$, retrieved from the children of $\nu$, using \cref{lm:matrix_mult_w}.

Overall, first phase, which uses \cref{lem:decomp2} to build $F$ along with the fragments $\hX_i$ and $\hY_i$ for each $i\in F$, takes $\Oh(k^2)$ time in the \modelname{} model.
An application of \cref{lem:addtoboxds} to build $\boxds(X_i, Y_i)$ for all $i\in F$ costs $\Oh(W\cdot k^2\log^2 n)$ time plus $\Oh(k^2)$ \modelname{} operations.
Finally, while building $T$ we create $\Oh(k\log n)$ primary nodes ($\Oh(k)$ per level). 
Initializing each of them costs $\Oh(W\cdot k\log n)$ time, dominated by the application of \cref{lm:matrix_mult_w} for Monge matrices whose cores are of size $\Oh(W\cdot k)$ by \cref{lem:bdtocore,lem:bdproduct,obs:bmw_is_bd}.
The total initialization time is therefore $\Oh(W\cdot k^2 \log^2 n)$ plus $\Oh(k^2)$ \modelname{} operations.

\subparagraph{Updates.}
Given some edit $(\cB_i, \cC_i)$, we use it to update $\cB$ and $\cC$ using \cref{cor:alignment-composition-algorithm} in time $\Oh(k / W)$.
Define a set $\Delta\subseteq \fragmentco{0}{m}$ consisting of all integers $i$ such that $X_i$ or $Y_i$ is affected by this update. That is, either $\cB_i$ does not match $X_i$ perfectly or $\cC_i$ does not match $Y_i$ perfectly.
The fragments $\hX_i$ are disjoint, so each edit in $X$ affects at most one fragment $X_i=\cB(\hX_i)$.
The fragments $\hY_i,\hY_{i+1}$ for $i \ge 3$ start at least $3k$ positions apart and overlap by at most $12k$ characters, so each position is contained in at most eight fragments $Y_i=\cC(\hY_i)$.
Consequently, each edit in $Y$ affects at most eight fragments $Y_i=\cC(\hY_i)$.
We reconstruct all \emph{affected} nodes $\nu$ in $T$ such that $\fragmentco{\ell_{\nu}}{r_\nu}\cap \Delta \ne \emptyset$. 
For this, we traverse the tree $T$ in a post-order fashion skipping subtrees rooted at unaffected nodes.
Keeping track of the relative location of the edit compared to $x_{\ell_\nu}$, $y_{\ell_\nu}$, and $y'_{\ell_\nu}$, we can check this condition based on the stored values $x_{r_\nu}-x_{\ell_\nu}$, $y_{r_\nu}-y_{\ell_\nu}$, and $y'_{r_\nu}-y'_{\ell_\nu}$.
For each affected leaf $i$, we create an updated leaf using the mega-edit operation of \cref{prp:hadstwo} to obtain $\boxds(X_i, Y_i)$ for the updated versions of $X_i$ and $Y_i$ and \cref{prp:hadstwo,lm:ds_tools} to retrieve $\mds(D_{i,i+1})$; this costs $\Oh(W\cdot k\log^2 k) \le \Oh(W \cdot k \log^2 n)$ time per affected leaf due to \cref{obs:bmw_is_bd}.
For each affected internal node $\nu$, after processing the two children $\lambda$ and $\rho$,  we can construct $\mds(D_{\ell_\nu,r_\nu})$ from $\mds(D_{\ell_\lambda,r_\lambda})$ and $\mds(D_{\ell_\rho,r_\rho})$, retrieved from the children of $\nu$, using \cref{lm:matrix_mult_w}.
Processing each internal node costs $\Oh(W\cdot k\log n)$ time since $|X_i|+|Y_i|\le 25k$ and $|V_i|\le 13k$ hold throughout the lifetime of the algorithm, each processed matrix is of size $\Oh(k)\times \Oh(k)$ and, by \cref{lem:bdtocore,lem:bdproduct,obs:bmw_is_bd}, has core of size $\Oh(W\cdot k)$.
Overall, we update $\Oh(1)$ leaves and $\Oh(\log n)$ internal nodes, so the total update time is $\Oh(W\cdot k \log^2 n)$.

\subparagraph{Alignment retrieval.}
It remains to provide a subroutine that constructs the breakpoint representation of a $w$-optimal alignment $\cA: X\onto Y$.
For this, we develop a recursive procedure that, given a node $\nu$ in $T$, a vertex $p\in V_{\ell_\nu}$, and a vertex $q\in V_{r_\nu}$, construct the breakpoint representation of a shortest path from $p$ to $q$ in $G$.

The distance $d=\dist_G(p,q)$ is stored within $D_{\ell_\nu,r_\nu}$, so we can retrieve it in $\Oh(\log n)$ time from $\mds(D_{\ell_\nu,r_\nu})$.
If $d = 0$, there is a trivial shortest path from $p$ to $q$ always taking diagonal edges of weight zero, and we can return the breakpoint representation of such a path in constant time.
We henceforth assume $d > 0$ and consider two further cases.

If $\nu$ is a leaf with $\fragmentco{\ell_\nu}{r_\nu}=\{i\}$, then $p$ can be interpreted as in input vertex of $G_i$ and $q$ can be interpreted as an output vertex of $G_i$. 
Moreover, by \cref{grid-graph-extension}, every shortest path from $p$ to $q$ is completely contained in $G_i$.
Since $G_i$ is isomorphic to $\oAGw(X_i,Y_i)$, it suffices to use the alignment retrieval operation of $\boxds$.

If $\nu$ is an internal node with children $\lambda$ and $\rho$, then $V_{r_{\lambda}}=V_{\ell_\rho}$ separates $V_{\ell_\nu}=V_{\ell_\lambda}$ from $V_{r_\nu}=V_{r_\rho}$.
Hence, every path from $p$ to $q$ contains a vertex $v\in V_{r_{\lambda}}$ and $d=\min_{v\in V_{r_{\lambda}}}\left(\dist_G(p,v)+\dist_G(v,q)\right)$. 
By \cref{lm:paths_dont_deviate_too_much}, we can identify $\Oh(d)$ vertices $v\in V_{r_{\lambda}}$ for which $\dist_G(p,v)\le \dist_{\oAGw(X,Y)} (p,q) \le d$ \emph{may} hold.
For each such vertex, we can compute $\dist_G(p, v)$ and $\dist_G(v, q)$ in $\Oh(\log n)$ time using the random access provided by $\mds(D_{\ell_\lambda,r_\lambda})$ and $\mds(D_{\ell_\rho,r_\rho})$. 
We pick any vertex $v$ such that $\dist_G(p,v)+\dist_G(v,q)=d$, recurse on $(\lambda,p,v)$ and $(\rho,v,q)$, and concatenate the obtained breakpoint representations.

Let us analyze the running time of this algorithm.
For this, for every node $\nu$, denote with $d_\nu$ the cost that the returned path incurs on the way from $V_{\ell_\nu}$ to $V_{r_\nu}$.
If $d_\nu=0$, then the retrieval algorithm spends $\Oh(\log n)$ time (dominated by the random access to $D_{\ell_\nu,r_\nu}$). 
Unless $\nu$ is the root, this can be charged to the parent $\mu$ of $\nu$ that satisfies $d_\mu > 0$.
If $d_\nu > 0$ and $\nu$ is a leaf, then the running time is $\Oh(d_\nu\log^2 n)$, dominated by the alignment retrieval operation of $\boxds$; see \cref{prp:hadstwo}.
If $d_\nu > 0$ and $\nu$ is an internal node, then the running time (excluding recursive calls) is $\Oh(d_\nu \log n)$, dominated by $\Oh(d_\nu)$ random access queries to $D_{\ell_\lambda,r_\lambda}$ and $D_{\ell_\rho,r_\rho}$.
For each level of the tree $T$, the values $d_\nu$ sum up to $d$. 
Hence, the total cost at the leaf level is at most $\Oh(d\log^2 n)$ and the total cost at every other single level is $\Oh(d \log n)$.
There are $\Oh(\log n)$ levels overall, so the global cost is $\Oh(d \log^2 n)$ if $d > 0$ and $\Oh(\log n)$ if $d = 0$.
Since $d=\wed(X,Y)\le 2k$, this is clearly dominated by the update cost of $\Oh(W\cdot k\log^2 n)$.
\end{proof}

\subsection{The General Case}\label{sec:dyn-general}

For the general case, we use a divide-and-conquer procedure similar to the one used in \cref{lm:upgrade_alignment}.
We first describe a dynamic algorithm that will execute the conquer step in the dynamic divide-and-conquer recursive scheme.
It maintains a data structure that given $w$-optimal alignments for the two halves of the strings, combines them into a $w$-optimal alignment for the complete strings.

\begin{lemma}\label{lem:glue}
    There exists a dynamic algorithm that, given positive integers $n \ge 2$ and $k \in \fragmentcc{1}{n^2}$, oracle access to a weight function $w:\Esigma^2\to \fragmentcc{0}{W}$, and strings $X_L,X_M,X_R,Y_L,Y_M,Y_R\in \Sigma^*$ that initially satisfy $|X_L|, |X_M|, |X_R|, |Y_L|, |Y_M|, |Y_R| \le n$, $\wed(X_L,Y_L)+\wed(X_M,Y_M)+\wed(X_R,Y_R) \le k$, and $|X_M|\le k$, maintains, subject to $\floor{k/W}$ edits in $X_L,X_M,X_R,Y_L,Y_M,Y_R$ in the algorithm lifetime, a data structure supporting the following queries: given $w$-optimal alignments $\cO_L:X_L\onto Y_L$ and $\cO_R:X_R\onto Y_R$, construct a $w$-optimal alignment $\cO:X\onto Y$, where $X=X_LX_MX_R$ and $Y=Y_LY_MY_R$.
    The algorithm takes $\Oh(W \cdot k^2\log^2 n)$ time plus $\Oh(k^2\log (W + 1))$ \modelname{} operations for initialization, $\Oh(W\cdot k\log^2 n)$ time per update, and $\Oh(k)$ time per query.
\end{lemma}
\begin{proof}
Let $\hX_L,\hX_M,\hX_R, \hY_L,\hY_M,\hY_R$ be strings given to the algorithm at the initialization phase.
Consider $w$-optimal alignments $\cA_L : \hX_L\onto \hY_L$, $\cA_M : \hX_M\onto \hY_M$, and $\cA_R: \hX_R\onto \hY_R$.
Moreover, let $\hX^*_L$ be the longest suffix of $\hX_L$ satisfying $\sed(\hX^*_L)\le 11k$ and $\hX^*_R$ be the longest prefix of $\hX_R$ satisfying $\sed(\hX^*_R)\le 11k$. 
Furthermore, let $\hY^*_L = \cA_L(\hX^*_L)$ and $\hY^*_R = \cA_R(\hX^*_R)$.

We maintain the breakpoint representations of alignments $\cB_L : \hX_L\onto X_L$, $\cB_M : \hX_M\onto X_M$, $\cB_R:\hX_R\onto X_R$, $\cC_L:\hY_L\onto Y_L$, $\cC_M:\hY_M\onto Y_M$, and $\cC_R:\hY_R\onto Y_R$ representing the changes made throughout the lifetime of the algorithm, and let $X^*_L = \cB_L(\hX^*_L)$, $X^*_R = \cB_R(\hX^*_R)$, $Y^*_L=\cC_L(\hY^*_L)$, and $Y^*_R=\cC_R(\hY^*_R)$.

At the initialization phase, we apply \cref{lm:selfed2} to determine $\hX^*_L$ and $\hX^*_R$.
We also use \cref{lm:full_algorithm} to compute (breakpoint representations of) $\cA_L$ and $\cA_R$; this lets us derive $\hY^*_L$ and $\hY^*_R$, respectively. 
Finally, we initialize an instance of the dynamic algorithm of \cref{lem:dyn_sed} for $(\hX^*_L\hX_M\hX^*_R,\hY^*_L\hY_M\hY^*_R,23k)$.
We prove that such a threshold is sufficient.
Firstly, $\sed(\hX^*_L\hX_M\hX^*_R)\le \sed(\hX^*_L)+|\hX_M|+\sed(\hX^*_R)\le 11k+k+11k=23k$ follows from \cref{fct:selfed-properties}.
Secondly, denote $d_X = \ed_{\cB_L}(\hX_L,X_L)$, $d_Y=\ed_{\cC_L}(\hY_L,Y_L)$, and $\cD_L = \cC_L \circ (\cA_L \circ \cB_L^{-1})$.
Since $\wed_{\cA_L}(\hX_L,\hY_L)=\wed(\hX_L,\hY_L)\le k$ and $d_X+d_Y\le \floor{k/W}$,
\cref{lem:triw-cor} yields
\begin{align*}
    \wed(X^*_L, Y^*_L) &\le \wed_{\cD_L}(X^*_L, Y^*_L)\\
                       &\le \wed_{\cA_L}(\hX^*_L, \hY^*_L) + W \cdot (\ed_{\cB_L}(\hX^*_L, X^*_L) + \ed_{\cC_L}(\hY^*_L, Y^*_L))\\
                       &\le \wed_{\cA_L}(\hX_L, \hY_L) + W \cdot (\ed_{\cB_L}(\hX_L, X_L) + \ed_{\cC_L}(\hY_L, Y_L))\\
                       &\le k + W \cdot \floor{k / W}\\
                       &\le 2k.
\end{align*}
Analogously, $\wed(X_M, Y_M) \le 2k$ and $\wed(X^*_R, Y^*_R) \le 2k$.
Hence, $\wed(X^*_LX_MX^*_R, Y^*_LY_MY^*_R) \le \wed(X^*_L, Y^*_L) + \wed(X_M, Y_M)+ \wed(X^*_R, Y^*_R) \le 2k + 2k + 2k < 23k$ follows from sub-additivity of edit distance.
Therefore, the threshold of $23k$ is indeed sufficient.

Furthermore, we initialize $\cB : \hX \onto \hX$ and $\cC : \hY \onto \hY$ as identity alignments.

The three steps of the initialization algorithm cost $\Oh(k^2)$ \modelname{} operations, $\Oh(W\cdot k^2\log^2 n)$ time plus $\Oh(k^2 \log (W + 1))$ \modelname{} operations, and $\Oh(W\cdot k^2\log^2 n)$ time plus $\Oh(k^2)$ \modelname{} operations, respectively, for a total of $\Oh(W\cdot k^2\log^2 n)$ time plus $\Oh(k^2 \log (W + 1))$ \modelname{} operations.

\smallskip

Given some edit $(\cB_i, \cC_i)$, we use it to update $\cB$ and $\cC$ using \cref{cor:alignment-composition-algorithm} in time $\Oh(k / W)$.
Furthermore, we reflect such an edit to $X_L,X_R,Y_L,Y_R$ in the instance $(X^*_LX^*_R,Y^*_LY^*_R,22k)$ of \cref{lem:dyn_sed}.
It can support at least $\floor{22k/W}\ge \floor{k/W}$ updates and its update time is $\Oh(W \cdot k \log^2 n)$.
After each update, we store the breakpoint representation of a $w$-optimal alignment $\cO_M : X^*_LX_MX^*_R\onto Y^*_LY_MY^*_R$.

\smallskip

The query algorithm constructs a $w$-optimal alignment $\cO : X\onto Y$ based on the following claim:
\begin{claim}\label{clm:glue:equiv}
    Let $\cO_L : X_L\onto Y_L$, $\cO_R:X_R\onto Y_R$, and $\cO_M : X^*_LX_MX^*_R\onto Y^*_LY_MY^*_R$ be $w$-optimal alignments, interpreted as alignments mapping fragments of $X$ to fragments of $Y$.
    Then, $\cO_L\cap \cO_M\ne \emptyset \ne \cO_R\cap \cO_M$, and all points $(x^*_L,y^*_L)\in \cO_L\cap \cO_M$ and $(x^*_R,y^*_R)\in  \cO_R\cap \cO_M$  satisfy
    \[\wed\!(X,Y)\!=\!\wed_{\cO_L}\!(X\fragmentco{0}{x^*_L},\!Y\fragmentco{0}{y^*_L})+\wed_{\cO_M}\!(X\fragmentco{x^*_L}{x^*_R},\! Y\fragmentco{y^*_L}{y^*_R})+\wed_{\cO_R}\!(X\fragmentco{x^*_R}{|X|},\!Y\fragmentco{y^*_R}{|Y|}).\]
\end{claim}
\begin{claimproof}
Consider an alignment $\cA : \hX_L\hX_M\hX_R\onto \hY_L\hY_M\hY_R$ obtained by concatenating $\cA_L$, $\cA_M$, and $\cA_R$,
an alignment $\cB : \hX_L\hX_M\hX_R\onto X_LX_MX_R$ obtained by concatenating $\cB_L$, $\cB_M$, and $\cB_R$,
and an alignment $\cC : \hY_L\hY_M\hY_R \onto Y_LY_MY_R$ obtained by concatenating $\cC_L$, $\cC_M$, and $\cC_R$.
Finally, let $\cD : X\onto Y$ be an alignment obtained as a composition $\cD = \cC \circ (\cA \circ \cB^{-1})$. 
Since $\wed_{\cA}(\hX_L\hX_M\hX_R,\hY_L\hY_M\hY_R) = \wed_{\cA_L}(\hX_L,\hY_L) +\wed_{\cA_M}(\hX_M,\hY_M) + \wed_{\cA_R}(\hX_R,\hY_R) = \wed(\hX_L,\hY_L)+ \wed(\hX_M,\hY_M) + \wed(\hX_R,\hY_R) \le k$ and since $\cB$ and $\cC$ make at most $\floor{k/W}$ edits in total,
\cref{lem:triw-cor} implies $\wed_{\cD}(X,Y)\le k + W\cdot \floor{k/W}\le 2k$.

If we denote $\cD \eqqcolon (x_i,y_i)_{i=0}^t$ then, by the definitions above, there exist integers $0\le \ell \le p \le q \le r \le t$
such that $X\fragmentco{0}{x_p} = X_L$, $Y\fragmentco{0}{y_p} = Y_L$, $X\fragmentco{x_{\ell}}{x_p} = X_L^*$, $Y\fragmentco{y_{\ell}}{y_p} = Y_L^*$, $X\fragmentco{x_p}{x_q} = X_M$, $Y\fragmentco{y_p}{y_q} = Y_M$, $X\fragmentco{x_q}{|X|} = X_R$, $Y\fragmentco{y_q}{|Y|} = Y_R$, $X\fragmentco{x_q}{x_r} = X_R^*$, and $Y\fragmentco{y_q}{y_r} = Y_R^*$.
Thus, we can reinterpret alignments $\cO_L, \cO_M$, and $\cO_R$ in the following way:
$\cO_L : X\fragmentco{0}{x_p}\onto Y\fragmentco{0}{y_p}$, $\cO_M : X\fragmentco{x_\ell}{x_r}\onto Y\fragmentco{y_\ell}{y_r}$, and  $\cO_R : X\fragmentco{x_q}{|X|}\onto Y\fragmentco{y_q}{|Y|}$.
The claim statement is identical to that of \cref{cor:intersect}, so it suffices to check that the assumptions of \cref{cor:intersect} are satisfied.

The maximality of $\hX^*_L$ and continuity of $\sed$ from \cref{fct:selfed-properties} show that $\hX^*_L=\hX_L$ or $\sed(\hX^*_L)= 11k$.
In the former case, we also have $X^*_L=\cB(\hX^*_L)=\cB(\hX_L)=X_L$ and $Y^*_L=\cC(\cA(\hX^*_L))=\cC(\cA(\hX_L))=\cC(\hY_L)=Y_L$, so $\ell=0$.
In the latter case, the triangle inequality property of \cref{fct:selfed-properties} implies $\sed(X^*_L)\ge \sed(\hX^*_L)-2\ed(\hX^*_L,X^*_L) \ge 11k-2\ed_{\cB_L}(\hX^*_L,X^*_L) \ge 11k-2\floor{k/W}\ge 9k > 4\cdot 2k \ge 4\wed_{\cD}(X,Y)$.
A symmetric argument shows that $r=t$ or $\ed(X^*_R)>4\wed_{\cD}(X,Y)$, so the claim indeed follows from \cref{cor:intersect}.
\end{claimproof}

Following \cref{clm:glue:equiv}, it suffices to reinterpret $\cO_L$, $\cO_R$, and $\cO_M$ as alignments mapping fragments of $X$ to fragments of $Y$ (for which, we traverse the breakpoint representations and shift the coordinates accordingly), find intersection points $(x^*_L,y^*_L)\in \cO_L\cap \cO_M$ and $(x^*_R,y^*_R)\in  \cO_R\cap \cO_M$, and construct the resulting alignment $\cO$ by following $\cO_L$ from $(0,0)$ to $(x^*_L,y^*_L)$, following $\cO_M$ from  $(x^*_L,y^*_L)$ to $(x^*_R,y^*_R)$, and following $\cO_R$ from $(x^*_R,y^*_R)$ to $(|X|,|Y|)$.
Each of these steps takes time proportional to $\Oh(1+\wed(X_L,Y_L)+\wed(X_R,Y_R) + \wed(X^*_LX_MX^*_R, Y^*_LY_MY^*_R))=\Oh(k)$ due to \cref{lem:triw-cor}.
\end{proof}

We are now ready to describe the dynamic \modelname algorithm for the general case under a light constraint that the polylogarithmic factors are taken from some universal fixed bound $n$ given to the algorithm at the initialization.
This algorithm is essentially a dynamic version of \cref{lm:upgrade_alignment}.

\begin{lemma}\label{lm:dynamic}
    There exists a dynamic algorithm that, given a positive integer $n \ge 2$ and oracle access to a weight function $w:\Esigma^2\to \fragmentcc{0}{W}$ with $W \le n$ and strings $X,Y\in \Sigma^*$ that throughout the lifetime of the algorithm satisfy $|X| + |Y| \le n$, maintains the breakpoint representation of a $w$-optimal alignment $\cO : X\onto Y$ subject to edits in $X$ and $Y$.
    The algorithm takes $\Oh(W\cdot d^2\log^2 n)$ time plus $\Oh(d^2 \log (W + 1))$ \modelname{} operations for initialization and $\Oh(W^2\cdot d\log^2 n)$ time plus $\Oh(W\cdot d \log (W + 1))$ \modelname{} operations for updates, where $d=\max\{1, \wed(X,Y)\}$ (here distance is taken after the edit).
\end{lemma}
\begin{proof}
    We develop two dynamic algorithms that will use each other recursively.
    The first dynamic algorithm given oracle access to strings $X$ and $Y$ maintains the breakpoint representation of a $w$-optimal alignment $\cO : X \onto Y$ subject to edits in $X$ and $Y$ and is the algorithm required by the lemma statement.
    We call this dynamic algorithm $\Ta$.
    The other algorithm's interface is similar to $\Ta$, but its lifetime is limited to $\floor{k / W}$ updates, where $k$ is $\wed(X, Y)$ at the initialization.
    We call the second dynamic algorithm $\Ra$.
    The first algorithm $\Ta$ maintains two instances of $\Ra$ and switches between them in epochs.
    The second algorithm $\Ra$ initially splits $X$ into $X_L X_M X_R$ and $Y$ into $Y_L Y_M Y_R$, maintains the algorithm of \cref{lem:glue} on top of them, and recursively maintains an instance of $\Ta$ for $X_L$ and $Y_L$ and an instance of $\Ta$ for $X_R$ and $Y_R$.
    We now describe in detail how these two dynamic algorithms work.

    Consider some instance $\nu$ of $\Ta$ that maintains $\wed$ of two strings $X_{\nu}$ and $Y_{\nu}$ subject to edits.
    Let $\hX_{\nu}$ and $\hY_{\nu}$ be these strings at the initialization phase.
    Algorithm $\Ta$ works in \emph{time blocks}.
    The first time block starts at the initialization, and each time block lasts $\ceil{\varepsilon k_{\nu} / W}$ updates, where $\varepsilon \coloneqq 1 / 500$, and $k_{\nu}$ is equal to $\wed(X_{\nu}, Y_{\nu})$ at the start of the current time block.
    If $\varepsilon k_{\nu} / W \le 1$, this time block is a \emph{single update}, and if $\varepsilon k_{\nu} / W > 1$, it is an \emph{epoch} of length $\ceil{\varepsilon k_{\nu} / W} \ge 2$ updates.
    Let $d_{\nu}$ be $\wed(X_{\nu}, Y_{\nu})$ after the current update.
    Instance $\nu$ maintains $k_{\nu}$, $d_{\nu}$, and potentially two instances $\bA$ and $\bB$ of $\Ra$.
    At the initialization stage, we invoke the algorithm of \cref{lm:full_algorithm} to compute $\hk_{\nu} \coloneqq \wed(\hX_{\nu}, \hY_{\nu})$, where $\hX_{\nu}$ and $\hY_{\nu}$ are $X_{\nu}$ and $Y_{\nu}$ at the initialization.
    Note that at the beginning of the first time block we have $k_{\nu} = \hk_{\nu}$.
    If $\varepsilon \hk_{\nu} / W > 1$, we additionally initialize an instance $\bA$ of $\Ra$.

    We now describe how updates are handled.
    Let us first assume that $\varepsilon k_{\nu} / W > 1$ holds for all time blocks.
    In other words, all time blocks are epochs that last $\ceil{\varepsilon k_{\nu} / W} \ge 2$ updates.
    The updates during the epoch are applied immediately to $\bA$, which we use to maintain a $w$-optimal alignment $\cO : X_{\nu} \onto Y_{\nu}$, and also buffered in a queue.
    In the background, during all epochs except for the first one, we initialize $\bB$ with $X'_{\nu}$ and $Y'_{\nu}$, where $X'_{\nu}$ and $Y'_{\nu}$ are strings $X_{\nu}$ and $Y_{\nu}$ at the beginning of the current epoch.
    Furthermore, we replay the updates using the buffer so that $\bB$ is up-to-date by the end of the epoch.
    At the end of each epoch except for the first one, $\bA$ and $\bB$ switch roles.
    Additionally, at the end of each epoch $k_{\nu}$ is assigned the value of $d_{\nu}$.

    It remains to drop the assumption that $\varepsilon k_{\nu} / W > 1$ at the beginning of each time block.
    If $\varepsilon k_{\nu} / W \le 1$, the current time block is a single update.
    In this case, we discard $\bA$ and $\bB$, and as long as $\varepsilon k_{\nu} / W \le 1$, compute a $w$-optimal alignment from scratch using \cref{lm:full_algorithm}.
    As soon as $\varepsilon d_{\nu} / W > 1$, we initialize $\bA$ and start running regular epochs.

    We now describe how the algorithm $\Ra$ works.
    Consider some instance $\bA$ of $\Ra$ that maintains $\wed$ of two strings $X_{\bA}$ and $Y_{\bA}$ subject to at most $\floor{k_{\bA} / W}$ updates, where $k_{\bA}$ is equal to $\wed$ of $X_{\bA}$ and $Y_{\bA}$ at the initialization phase.
    When $\bA$ is initialized, it first invokes \cref{lm:full_algorithm} to find $k_{\bA}$ and the breakpoint representation of an optimal alignment $(x_i, y_i)_{i=0}^t \coloneqq \cO : X_{\bA} \onto Y_{\bA}$.
    If $k_{\bA} = 0$, $\floor{k_{\bA} / W} = 0$ updates are to be handled, and thus we can just terminate.
    Otherwise, we pick $m \in \fragmentcc{0}{t}$ as the largest possible index satisfying $\wed_{\cO}(X\fragmentco{0}{x_m}, Y\fragmentco{0}{y_m}) \le k_{\bA} / 2$.
    We then define $X_L \coloneqq X\fragmentco{0}{x_m}$, $Y_L \coloneqq Y\fragmentco{0}{y_m}$, $X_M \coloneqq X\fragmentco{x_m}{x_{m + 1}}$, $Y_M \coloneqq Y\fragmentco{y_m}{y_{m + 1}}$, $X_R \coloneqq X\fragmentco{x_{m + 1}}{|X|}$, and $Y_R \coloneqq Y\fragmentco{y_{m + 1}}{|Y|}$.
    Note that $|X_M| \le 1$ and $\wed(X_L, Y_L) + \wed(X_M, Y_M) + \wed(X_R, Y_R) = k_{\bA}$.
    We then initialize the dynamic algorithm of \cref{lem:glue} for these strings with parameters $n$ and $k_{\bA}$.
    Furthermore, we initialize an instance $\lambda_{\bA}$ of $\Ta$ for $X_L$ and $Y_L$ and an instance $\rho_{\bA}$ of $\Ta$ for $X_R$ and $Y_R$.
    Each update to $\bA$ is then applied to the instance of \cref{lem:glue}, and to $\lambda_{\bA}$ or $\rho_{\bA}$ if the update affects the corresponding parts of $X$ or $Y$.
    After each update, $\lambda_{\bA}$ returns a $w$-optimal alignment $\cO_L : X_L \onto Y_L$ and $\rho_{\bA}$ returns a $w$-optimal alignment $\cO_R : X_R \onto Y_R$.
    We then use them to find a $w$-optimal alignment $\cO : X \onto Y$ using \cref{lem:glue}.

    \medskip

    We now prove correctness of the presented algorithms inductively.
    For the correctness of $\Ra$ conditioned on the correctness of the underlying instances of $\Ta$, we just need to show that all conditions for \cref{lem:glue} are satisfied when we create $\bA$ and $\bB$.
    The bounds on the lengths of the input strings are guaranteed by the current lemma statement.
    Furthermore, as $W \le n$, throughout the lifetime of our algorithm we have $\wed(X, Y) \le (|X| + |Y|) \cdot W \le n^2$, and thus the same condition holds for all recursive calls.
    Finally, $|X_m| \le 1 \le k$. Therefore, \cref{lem:glue} is indeed applicable.
    We now prove correctness of $\Ta$ conditioned on the correctness of the underlying instances of $\Ra$.
    For that, we need to prove that at no point in time $\bA$ or $\bB$ are expected to handle more updates than their budget.
    Each instance of $\Ra$ is used for at most two epochs of the algorithm.
    Let $k_{\nu}$ be the value, with which the current instance of $\Ra$ was initialized, and $k'_{\nu}$ be the value of $\wed(X_{\nu}, Y_{\nu})$ at the beginning of the next epoch.
    As after one update, $\wed(X_{\nu}, Y_{\nu})$ changes by at most $W$, we have that $k'_{\nu} \le k_{\nu} + W \cdot \ceil{\varepsilon k_{\nu} / W} \le (1 + 2\varepsilon) \cdot k_{\nu}$, where $\ceil{\varepsilon k_{\nu} / W} \le 2 \varepsilon k_{\nu} / W$ as $\varepsilon k_{\nu} / W > 1$.
    Therefore, these two epochs last at most $\ceil{\varepsilon k_{\nu} / W} + \ceil{\varepsilon k'_{\nu} / W} \le \ceil{\varepsilon k_{\nu} / W} + \ceil{\varepsilon \cdot (1 + 2 \varepsilon) \cdot k_{\nu} / W} \le \varepsilon \cdot (4 + 4 \varepsilon) \cdot k_{\nu} / W \le (k_{\nu} - W) / W \le \floor{k_{\nu} / W}$ updates, where we use $\varepsilon = 1 / 500$ and $\varepsilon k_{\nu} / W > 1$.
    Hence, indeed, the instance of $\Ra$ in question does not run out of updates to handle.

    \medskip

    We now analyze the time complexity of the resulting algorithm.
    The initialization of an instance $\nu$ of $\Ta$ ran on a pair of strings $X_{\nu}$ and $Y_{\nu}$ (potentially with a corresponding instance $\bA$ of $\Ra$) ignoring the further recursive calls takes $\Oh(W \cdot \hk_{\nu}^2 \log^2 n)$ time dominated by the calls to \cref{lm:full_algorithm} and \cref{lem:glue}.
    A single update to $\nu$ inside an epoch ignoring the further recursive calls and the maintenance of $\bB$ takes $\Oh(W \cdot k_{\nu} \log^2 n) = \Oh(W \cdot d_{\nu} \log^2 n)$ time dominated by the update of \cref{lem:glue}.
    Here we use the fact that at most $\ceil{\varepsilon k_{\nu} / W} \le 2 \varepsilon k_{\nu} / W$ updates are performed in the current epoch, and each one of them changes $\wed(X_{\nu}, Y_{\nu})$ by at most $W$, and thus $(1 - 2 \varepsilon) k_{\nu} \le k_{\nu} - W  \cdot \ceil{\varepsilon k_{\nu} / W} \le d_{\nu} \le k_{\nu} + W \cdot \ceil{\varepsilon k_{\nu} / W} \le (1 + 2 \varepsilon) k_{\nu}$.
    Choose a sufficiently large constant $\alpha$ that is not smaller than the maximum of the constants hidden under all $\Oh$ notations in the previous three sentences.
    Denote $\beta = 2 \alpha$, $\gamma = 18924 \alpha$, and $\delta = 37250 \alpha$.
    We use induction on the depth of the recursion of $\Ta$ to prove that run on a pair of strings $X_{\nu}$ and $Y_{\nu}$, an instance $\nu$ of $\Ta$ takes at most $\beta \cdot W \cdot \hk_{\nu}^2 \log^2 n$ time for initialization and at most $\delta \cdot d_{\nu} \cdot W^2 \log^2 n$ time for updates.
    Furthermore, as an induction helper fact, we prove that for any $\bk \ge \hk_{\nu}$, the first $\ceil{\varepsilon \cdot \bk / (2 W)}$ updates take at most $(1 + 2 \varepsilon) \cdot \gamma \cdot \bk \cdot W^2 \log^2 n$ time each.
    In particular, the first two conditions prove the lemma time complexity.

    At the initialization phase, all local work takes at most $\alpha \cdot W \cdot \hk_{\nu}^2 \log^2 n$ time.
    Furthermore, $\bA$ potentially creates two recursive calls of $\Ta$ from $(X_L, Y_L)$ and $(X_R, Y_R)$.
    Moreover, note that $\wed(X_L, Y_L) \le \wed_{\cO}(X_L, Y_L) \le \hk_{\nu} / 2$ and $\wed(X_R, Y_R) \le \wed_{\cO}(X_R, Y_R) = \wed_{\cO}(X, Y) - \wed_{\cO}(X\fragmentco{0}{x_{m + 1}}, Y\fragmentco{0}{y_{m + 1}}) \le \hk_{\nu} - \hk_{\nu} / 2 = \hk_{\nu} / 2$ by the definition of $m$.
    By the inductive hypothesis, each of these recursive calls takes at most $\beta \cdot W \cdot (\hk / 2)^2 \log^2 n$ time for initialization.
    In total, the initialization phase takes time at most $\alpha \cdot W \cdot \hk^2 \log^2 n + 2 \cdot \beta \cdot W \cdot (\hk / 2)^2 \log^2 n = (\alpha + \beta / 2) \cdot W \cdot \hk^2 \log^2 n = \beta \cdot W \cdot \hk^2 \log^2 n$, thus proving the claim.

    We now analyze the updates from the time blocks that consist of single updates.
    In other words, if $\varepsilon k_{\nu} / W \le 1$, and thus $k_{\nu} \le W / \varepsilon$.
    We handle such an update by invoking \cref{lm:full_algorithm}.
    Furthermore, $d_{\nu} \le k_{\nu} + W \le W \cdot (1 + 1 / \varepsilon)$.
    It takes time $\alpha \cdot W \cdot d_{\nu}^2 \log^2 n \le \alpha \cdot W \cdot (W \cdot (1 + 1 / \varepsilon)) \cdot d_{\nu} \log^2 n = ((1 + 1 / \varepsilon) \cdot \alpha) \cdot W^2 \cdot d_{\nu} \log^2 n \le (\gamma / 3) \cdot W^2 \cdot d_{\nu} \log^2 n$, where the last inequality can be checked manually.
    Furthermore, if after this update we have $\varepsilon d_{\nu} / W > 1$, we spend additional $\beta \cdot W \cdot d_{\nu}^2 \log^2 n = 2 \alpha \cdot W \cdot d_{\nu}^2 \log^2 n \le (2 \cdot \gamma / 3) \cdot W^2 \cdot d_{\nu} \log^2 n$ time in total to initialize $\bA$ if the next time block is an epoch.
    In any case, the total time spent is at most $\gamma \cdot W^2 \cdot d_{\nu} \log^2 n$.

    We now analyze the updates from an epoch that is either the first time block or is preceded by a time block consisting of a single update.
    Such an epoch consists of maintaining only $\bA$.
    The local cost of an update is at most $\alpha \cdot W \cdot d_{\nu} \log^2 n \le \alpha \cdot W^2 \cdot d_{\nu} \log^2 n$ time.
    Furthermore, a recursive update may be invoked from either $\lambda_{\bA}$ or $\rho_{\bA}$.
    At the initialization of $\bA$, we have $\hk_{\lambda_{\bA}}, \hk_{\rho_{\bA}} \le k_{\nu} / 2$, where $k_{\nu} \le d_{\nu} / (1 - 2\varepsilon) \le (1 + 3 \varepsilon) \cdot d_{\nu}$ as $d_{\nu} \ge (1 - 2 \varepsilon) k_{\nu}$.
    At most $\ceil{\varepsilon k_{\nu} / W}$ updates were applied, and thus $d_{\lambda_{\bA}}, d_{\rho_{\bA}} \le (1 / 2 + 2 \varepsilon) \cdot k_{\nu} \le (1 / 2 + 2 \varepsilon) \cdot (1 + 3 \varepsilon) \cdot d_{\nu} \le (1 / 2 + 4 \varepsilon) \cdot d_{\nu}$.
    By the induction hypothesis, a recursive update takes time at most $\delta \cdot W^2 \cdot d \log^2 n$, where $d$ is either $d_{\lambda_{\bA}}$ or $d_{\rho_{\bA}}$.
    In total, the whole update takes time at most $\alpha \cdot W^2 \cdot d_{\nu} \log^2 n + \delta \cdot W^2 \cdot d \log^2 n \le \alpha \cdot W^2 \cdot d_{\nu} \log^2 n + \delta \cdot W^2 \cdot ((1 / 2 + 4 \varepsilon) \cdot d_{\nu}) \log^2 n = (\alpha + (1 / 2 + 4 \varepsilon) \cdot \delta) \cdot W^2 \cdot d_{\nu} \log^2 n \le \gamma \cdot W^2 \cdot d_{\nu} \log^2 n$, where the last inequality can be checked manually.
    As $\gamma \le \delta$, we prove the desired upper bound on the time complexity of these updates.

    It remains to analyze the updates from an epoch that is preceded by another epoch.
    We have $\varepsilon k_{\nu} / W > 1$ in this case.
    The local cost of an update ignoring the maintenance of $\bB$ is at most $\alpha \cdot W \cdot d_{\nu} \log^2 n \le \alpha \cdot W^2 \cdot d_{\nu} \log^2 n$ time.
    Furthermore, a recursive update may be invoked from either $\lambda_{\bA}$ or $\rho_{\bA}$ of $\bA$.
    Let $k'_{\nu}$ be the value of $d_{\nu}$ at the start of the previous epoch.
    That is, the value, with which $\bA$ was initialized.
    At the initialization of $\bA$, we have $\hk_{\lambda_{\bA}}, \hk_{\rho_{\bA}} \le k'_{\nu} / 2$, where $k'_{\nu} \le k_{\nu} / (1 - 2 \varepsilon) \le d_{\nu} / (1 - 2 \varepsilon)^2 \le (1 + 5 \varepsilon) \cdot d_{\nu}$.
    At most $\ceil{\varepsilon k_{\nu} / W} + \ceil{\varepsilon k'_{\nu} / W} \le \ceil{\varepsilon \cdot (1 + 3 \varepsilon) \cdot d_{\nu} / W} + \ceil{\varepsilon \cdot (1 + 5 \varepsilon) \cdot d_{\nu} / W} \le \varepsilon (4 + 16 \varepsilon) \cdot d_{\nu} / W$ updates were applied, and thus $d_{\lambda_{\bA}}, d_{\rho_{\bA}} \le (1 + 5 \varepsilon) \cdot d_{\nu} / 2 + \varepsilon (4 + 16 \varepsilon) d_{\nu} = (1 / 2 + 13 \varepsilon / 2 + 16 \varepsilon^2) \cdot d_{\nu} \le (1 / 2 + 8 \varepsilon) \cdot d_{\nu}$.
    By the induction hypothesis, such update takes time at most $\delta \cdot W^2 \cdot d \log^2 n \le (1 / 2 + 8 \varepsilon) \cdot \delta \cdot W^2 \cdot d_{\nu} \log^2 n$, where $d$ is either $d_{\lambda_{\bA}}$ or $d_{\rho_{\bA}}$.
    Furthermore, during the whole epoch, the initialization and replaying of the updates using the buffer is performed for $\bB$.

    The initialization of $\bB$ similarly to the initialization of $\bA$ takes at most $\beta \cdot W \cdot k_{\nu}^2 \log^2 n$ time.
    Distributed equally over $\ceil{\varepsilon k_{\nu} / W}$ updates, it is $\beta W k_{\nu}^2 \log^2 n / \ceil{\varepsilon k_{\nu} / W} \le \beta W k_{\nu}^2 \log^2 n / (\varepsilon k_{\nu} / W) = (\beta / \varepsilon) W^2 k_{\nu} \log^2 n \le (\beta / \varepsilon) W^2 (d_{\nu} / (1 - 2\varepsilon)) \log^2 n$ time.
    Furthermore, at most $\ceil{\varepsilon k_{\nu} / W}$ updates are replayed locally and recursively in $\lambda_{\bB}$ and $\rho_{\bB}$.
    Set $\bk \coloneqq \floor{k_{\nu} / 2}$.
    As $\varepsilon k_{\nu} / W > 1$, we have $\varepsilon k_{\nu} > 1$, and thus $\bk \ge (k_{\nu} - 1) / 2 > k_{\nu} \cdot (1 - \varepsilon) / 2$.
    As $\bk = \floor{k_{\nu} / 2} \ge \hk_{\lambda_{\bB}}$ and $\bk \ge \hk_{\rho_{\bB}}$, by the inductive hypothesis for $\lambda_{\bB}$ and $\rho_{\bB}$, we know that the first $\ceil{\varepsilon \bk / (2 W)} \le \ceil{\varepsilon k_{\nu} / W}$ updates replayed recursively take time at most $(1 + 2 \varepsilon) \cdot \gamma \cdot \bk \cdot W^2 \log^2 n$ each.
    Replaying recursively the remaining $\ceil{\varepsilon k_{\nu} / W} - \ceil{\varepsilon \bk / (2 W)}$ updates takes time at most $\delta \cdot d \cdot W^2 \log^2 n$ each, where $d$ is either $d_{\lambda_{\bB}}$ or $d_{\rho_{\bB}}$.
    We have $d_{\lambda_{\bB}}, d_{\rho_{\bB}} \le \bk + \ceil{\varepsilon k_{\nu} / W} \cdot W \le \bk + 2 \varepsilon k_{\nu} \le \bk + 2 \varepsilon \cdot (2 \bk / (1 - \varepsilon)) \le (1 + 5 \varepsilon) \cdot \bk$.
    Therefore, these updates take time at most $(1 + 5 \varepsilon) \cdot \delta \cdot \bk \cdot W^2 \log^2 n$ each.

    Distributed equally over $\ceil{\varepsilon k_{\nu} / W}$ updates\footnote{In reality, we cannot start replaying the updates before we initialize $\bB$, so actually we distribute equally the combined initialization and total update time of $\bB$, which is equivalent.},
    replaying of the updates on $\bB$ takes time
    \begin{align*}
    &(\ceil{\varepsilon \bk / (2 W)} \cdot (1 + 2 \varepsilon) \gamma \bk W^2 \log^2 n + (\ceil{\varepsilon k_{\nu} / W} - \ceil{\varepsilon \bk / (2 W)}) \cdot (1 + 5 \varepsilon) \delta \bk W^2 \log^2 n) / \ceil{\varepsilon k_{\nu} / W}\\
    &= ((1 + 5 \varepsilon) \cdot \delta - ((1 + 5 \varepsilon) \cdot \delta - (1 + 2 \varepsilon) \cdot \gamma) \cdot \ceil{\varepsilon \bk / (2 W)} / \ceil{\varepsilon k_{\nu} / W}) \cdot \bk \cdot W^2 \log^2 n\\
    &\le ((1 + 5 \varepsilon) \cdot \delta - ((1 + 5 \varepsilon) \cdot \delta - (1 + 2 \varepsilon) \cdot \gamma) \cdot \ceil{\varepsilon \bk / (2 W)} / \ceil{\varepsilon k_{\nu} / W}) \cdot (k_{\nu} / 2) \cdot W^2 \log^2 n\\
    &\le ((1 + 5 \varepsilon) \cdot \delta - ((1 + 5 \varepsilon) \cdot \delta - (1 + 2 \varepsilon) \cdot \gamma) \cdot \ceil{\varepsilon \bk / (2 W)} / (5 \ceil{\varepsilon \bk / (2 W)})) \cdot (k_{\nu} / 2) \cdot W^2 \log^2 n\\
    &= ((1 + 5 \varepsilon) \cdot \delta - ((1 + 5 \varepsilon) \cdot \delta - (1 + 2 \varepsilon) \cdot \gamma) / 5) \cdot (k_{\nu} / 2) \cdot W^2 \log^2 n\\
    &= ((2 / 5 + 2 \varepsilon) \cdot \delta + (1 / 10 + \varepsilon / 5) \cdot \gamma) \cdot k_{\nu} \cdot W^2 \log^2 n\\
    &\le ((2 / 5 + 2 \varepsilon) \cdot \delta + (1 / 10 + \varepsilon / 5) \cdot \gamma) \cdot (d_{\nu} / (1 - 2 \varepsilon)) \cdot W^2 \log^2 n,
\end{align*}
where $(1 + 5 \varepsilon) \cdot \delta - (1 + 2 \varepsilon) \cdot \gamma > 0$ holds as $\delta > \gamma$, and $\ceil{\varepsilon k_{\nu} / W} \le 5 \ceil{\varepsilon \cdot \bk / (2 W)}$ holds as
    \begin{align*}
        \ceil{\varepsilon k_{\nu} / W} &= \ceil{4 \cdot (\varepsilon \cdot (k_{\nu} - 1) / (4 W)) + \varepsilon / W} \le \ceil{4 \cdot (\varepsilon \cdot \bk / (2 W)) + \varepsilon / W}\\
                                       &\le 4 \ceil{\varepsilon \cdot \bk / (2 W)} + \ceil{\varepsilon / W} = 4 \ceil{\varepsilon \cdot \bk / (2 W)} + 1 \le 5 \ceil{\varepsilon \cdot \bk / (2 W)}.
    \end{align*}

    The total update time is thus bounded by
    \begin{align*}
        & (\alpha + (1 / 2 + 8 \varepsilon) \cdot \delta + (\beta / \varepsilon + (2 / 5 + 2 \varepsilon) \cdot \delta + (1 / 10 + \varepsilon / 5) \cdot \gamma) / (1 - 2 \varepsilon)) \cdot d_{\nu} \cdot W^2 \log^2 n\\
        & \le \delta \cdot d_{\nu} \cdot W^2 \log^2 n,
    \end{align*}
    where the inequality can be checked manually.

    It remains to prove that for any $\bk \ge \hk_{\nu}$, the first $\ceil{\varepsilon \cdot \bk / (2 W)}$ updates take at most $(1 + 2 \varepsilon) \cdot \gamma \cdot \bk \cdot W^2 \log^2 n$ time each.
    We consider two cases.
    If $\bk \le 2 \hk_{\nu}$, we have $\ceil{\varepsilon \cdot \bk / (2 W)} \le \ceil{\varepsilon \hk_{\nu} / W}$.
    We either have $\varepsilon \hk_{\nu} / W \le 1$, and thus $\ceil{\varepsilon \cdot \bk / (2 W)} = 1$, and this first update is a part of a single update time block, or we have $\varepsilon \hk_{\nu} / W > 1$, and thus the first epoch lasts for $\ceil{\varepsilon \hk_{\nu} / W}$ updates.
    Therefore, as discussed before, the first $\ceil{\varepsilon \cdot \bk / (2 W)}$ updates take time at most $\gamma \cdot d_{\nu} \cdot W^2 \log^2 n \le (1 + 2 \varepsilon) \cdot \gamma \cdot \bk \cdot W^2 \log^2 n$ each, where $d_{\nu} \le \hk_{\nu} + \ceil{\varepsilon \hk_{\nu} / W} \cdot W \le (1 + 2 \varepsilon) \cdot \hk_{\nu} \le (1 + 2 \varepsilon) \cdot \bk$, thus proving the claim.

    We now consider the case $\bk > 2 \hk_{\nu}$.
    In this case, any update takes time at most $\delta \cdot d_{\nu} \cdot W^2 \log^2 n \le \delta \cdot (1 + 2 \varepsilon) \cdot \hk_{\nu} \cdot W^2 \log^2 n < \delta \cdot (1 + 2 \varepsilon) \cdot (\bk / 2) \cdot W^2 \log^2 n = ((1 + 2 \varepsilon) \cdot \delta / 2) \cdot \bk \cdot W^2 \log^2 n \le \gamma \cdot \bk \cdot W^2 \log^2 n$, where the last inequality can be checked manually.

    Therefore, we proved all the time bounds required by the induction.
    It remains to bound the number of \modelname{} operations.
    All \modelname{} operations we have in our algorithm arise from the calls to \cref{lm:full_algorithm} and \cref{lem:glue}.
    In both of these cases, the number of \modelname{} operations is smaller by a fixed factor of $\Theta(W \log^2 n / \log (W + 1))$ than the time complexity.
    Therefore, the resulting number of \modelname{} operations can be bounded as $\Oh(T \cdot \log (W + 1) / (W \cdot \log^2 n))$, where $T$ is the bound on the time complexity.
\end{proof}

Before proceeding to our main theorem of the section, let us show a simple variation of the \modelname{} implementation in the dynamic setting whose time complexity does not degrade over time.

\begin{lemma}\label{lm:dynamic-pillar}
    Let $X_1, \ldots, X_{\ell}$ be a collection of strings.
    There is an algorithm that dynamically maintains \modelname{} over $X_1, \ldots, X_{\ell}$ while the strings undergo edits.
    The algorithm takes $\Oh(n \log^{\Oh(1)}\log n)$ time for the initialization and $\Oh(\log n \log^{\Oh(1)}\log n)$ time for every edit and \modelname{} operation where $n = \max\{2, \sum_i |X_i|\}$.
\end{lemma}

\begin{proof}
    We concatenate all the input strings into a single string $X \coloneqq \bigodot_{i=1}^{\ell} X_i$ of length at most $n$ and maintain \modelname{} over $X$.
We use a deterministic implementation of the \modelname{} over $X$ in the dynamic setting~\cite[Section 8]{KK22}, \cite[Section 7.2]{CKW22} in epochs.
    It takes $\Oh(n \log^{\Oh(1)}\log n)$ time for the initialization and $\Oh(\log (n + m) \log^{\Oh(1)}\log (n + m))$ time for every edit and \modelname{} operation where $m$ is the number of edits applied to $X$.

    We maintain two instances of the described dynamic algorithm, denoted $\bA$ and $\bB$, and partition the lifetime of the algorithm into epochs.
    We denote by $\hX$ the string $X$ at the beginning of the epoch.
    Each epoch lasts for $\floor{\hn / 4}$ edits, where $\hn \coloneqq \max\{2, |\hX|\}$.
    
    The updates during the epoch are applied immediately to the algorithm $\bA$, which we use to answer \modelname{} queries, and also buffered in a queue.
    In the background, we initialize $\bB$ with $\hX$ and replay the updates using the buffer so that $\bB$ is up-to-date by the end of the epoch.
    In the next epoch, the algorithms $\mathbf{A}$ and $\mathbf{B}$ switch roles.

    As far as the initialization phase is concerned, we initialize $\bA$, which takes $\Oh(\hn \log^{\Oh(1)}\log \hn)$ time.
    We now analyze the time complexity of the updates.
    Note that every algorithm we maintain exists for at most two epochs.
    In particular, $m \le \floor{\hn / 4} + \floor{(\hn + \hn / 4) / 4} \le 3 \hn / 4$ where the algorithm is initialized with the string of length $\hn$.
    Furthermore, we have $n = \Theta(\hn)$ throughout the lifetime of this algorithm.
    Hence, edits and queries to $\bA$ take $\Oh(\log n \log^{\Oh(1)}\log n)$ time each.
    The initialization and replaying of the updates on $\bB$ takes $\Oh(\hn \log^{\Oh(1)} \log \hn + \floor{\hn / 4} \cdot \log \hn \log^{\Oh(1)}\log \hn)$ which costs $\Oh(\log n \log^{\Oh(1)}\log n)$ time per update if distributed equally among the updates.
\end{proof}

We now drop the artificial constraints imposed by \cref{lm:dynamic} and show how to implement the \modelname{} operations with little overhead.

\begin{theorem}\label{thm:dynamic}
    There exists a dynamic algorithm that, given oracle access to a weight function $w:\Esigma^2\to \fragmentcc{0}{W}$ and strings $X,Y\in \Sigma^*$, maintains the breakpoint representation of a $w$-optimal alignment $\cO : X\onto Y$ subject to edits in $X$ and $Y$.
    The algorithm takes $\Oh(n \log^{\Oh(1)}\log n + W\cdot d^2\log^2 n)$ time for initialization and $\Oh(W^2\cdot d\log^2 n)$ time for updates, where $n = \max\{2, |X| + |Y|\}$ and $d=\max\{1, \wed(X,Y)\}$ (here distance and string lengths are taken after the edit).
\end{theorem}
\begin{proof}
    We use \cref{lm:dynamic} in epochs, along with a deterministic implementation of the \modelname{} model over $X$ and $Y$ in the dynamic setting from \cref{lm:dynamic-pillar}. 
    The extra overhead of the latter is $\Oh(n \log^{\Oh(1)}\log n)$ for the initialization and $\Oh(\log n \log^{\Oh(1)}\log n)$ for each update and \modelname{} operation.
    As there is a $\Theta(W \cdot \log^2 n / \log (W + 1)) = \Omega(\log^2 n) = \omega(\log n \log^{\Oh(1)}\log n)$ gap in \cref{lm:dynamic} between the time complexity and the number of \modelname{} operations, the implementation of \modelname{} does not increase the time complexity of the updates.

    We maintain two instances of the dynamic algorithm of \cref{lm:dynamic}, denoted $\mathbf{A}$ and $\mathbf{B}$, and partition the lifetime of the algorithm into epochs. 
    We denote by $\hX$ and $\hY$ the strings $X$ and $Y$ at the beginning of the epoch.
    Each epoch lasts for $\floor{\hn / 5}$ edits, where $\hn \coloneqq 2 \cdot (|\hX| + |\hY|)$.
    Let us first assume that $\hn \ge 5$ and $W \le \hn$ hold at the beginning of each epoch.
    In this case, we have $\floor{\hn / 5} \ge \hn / 10$ as $\hn / 5 \ge 1$.

    The updates during the epoch are applied immediately to the algorithm $\mathbf{A}$, which we use to maintain a $w$-optimal alignment $\cO:X\onto Y$, and also buffered in a queue.
    In the background, we initialize $\mathbf{B}$ with a threshold $\hn$ and replay the updates using the buffer so that $\mathbf{B}$ is up-to-date by the end of the epoch. 
    In the next epoch, the algorithms $\mathbf{A}$ and $\mathbf{B}$ switch roles.

    As far as the initialization phase is concerned, we build \modelname{} over $X$ and $Y$ and initialize $\mathbf{A}$, which takes $\Oh(\hn \log^{\Oh(1)}\log \hn + W\cdot d^2\log^2 \hn)$ time. 

    We now analyze the time complexity of the updates.
    Pick sufficiently large constants $\alpha$ and $c$ such that the initialization of \cref{lm:dynamic} takes time at most $\alpha \cdot W\cdot k^2\log^2 \hn$, where $k$ is the value of $d$ at the start of the epoch, and an update in \cref{lm:dynamic} takes time at most $\alpha \cdot W^2\cdot d\log^2 \hn$.
    Then, while initializing $\bB$ and replaying the updates on it, we perform $90 \alpha W^2 \cdot d \log^2 \hn + \alpha W^2 \cdot d \log^2 \hn = \Oh(W^2 \cdot d \log^2 \hn)$ atomic operations from the buffer per update.
    If at some point the buffer is empty, the algorithm $\bB$ is already up-to-date, and we start performing updates on it simultaneously with $\bA$.
    This way, any update during the epoch costs $\Oh(W^2 \cdot d \log^2 \hn)$ time in total.
    On the other hand, we show that at the end of the epoch, $\bB$ is up-to-date.
    If at some point the buffer becomes empty, this is clear.
    Otherwise, we perform $\sum_t (90 \alpha W^2 \cdot d_t \log^2 \hn + \alpha W^2 \cdot d_t \log^2 \hn)$ operations in total, where the summation is over all points in time during the epoch.
    We first show that the values $90 \alpha W^2 \cdot d_t \log^2 \hn$ sum up to at least $\alpha \cdot W \cdot k^2 \log^2 \hn$.
    If $k \le 3 W$, we have $\sum_t 90 \alpha W^2 \cdot d_t \log^2 \hn \ge \sum_t 90 \alpha W^2 \log^2 \hn = \floor{\hn / 5} \cdot 90 \alpha W^2 \log^2 \hn \ge (\hn / 10) \cdot 90 \alpha W^2 \log^2 \hn = 9 \alpha \cdot \hn \cdot W^2 \log^2 \hn \ge 9 \alpha W^3 \log^2 \hn \ge \alpha \cdot W \cdot k^2 \log^2 \hn$.
    Otherwise, if $k > 3W$, we have $\ceil{k / (6 W)} \le k / (6 W) + 1 < k / (6 W) + k / (3 W) = k / (2 W)$.
    Furthermore, $k / (6 W) \le ((|\hX| + |\hY|) \cdot W) / (6 W) = (\hn \cdot W / 2) / (6 W) = \hn / 12 < \hn / 10 \le \floor{\hn / 5}$.
    Hence, the first $\ceil{k / (6 W)}$ updates are a part of the current epoch, and during them, we have $d \ge k - \ceil{k / (6 W)} \cdot W \ge k / 2$.
    Therefore, we have $\sum_t 90 \alpha W^2 \cdot d_t \log^2 \hn \ge (k / (6 W)) \cdot 90 \alpha W^2 \cdot (k / 2) \log^2 \hn \ge \alpha \cdot W \cdot k^2 \log^2 \hn$.
    Finally, the values $\alpha W^2 \cdot d_t \log^2 \hn$ over all updates sum up to the time required to replay the updates on $\bB$.
    Hence, $\bB$ is indeed up-to-date at the end of this epoch.

    We now show that $\hn$, with which $\bA$ is initialized is equal to $\Theta(\max\{2, |X| + |Y|\})$ throughout the lifetime of $\bA$ (and similarly for $\bB$), and thus we can replace $\hn$ with $n$ in the time complexities.
    The algorithm $\bA$ lives for two epochs.
    The first epoch lasts $\floor{\hn / 5} \le \hn / 5$ updates.
    At the end of this epoch the new value of $\hn' = 2 \cdot (|\hX'| + |\hY'|) \le 2 \cdot (\hn / 2 + \hn / 5) = (1 + 2 / 5) \cdot \hn$.
    Therefore, the two epochs last at most $\hn / 5 + (1 + 2 / 5) \cdot \hn / 5 = 12 \hn / 25$ updates.
    Hence, throughout the lifetime of $\bA$, we have $|X| + |Y| \ge \hn / 2 - 12 \hn / 25 = \hn / 50$.
    Furthermore, $|X| + |Y| \le \hn / 2 + 12 \hn / 25 \le \hn$.
    Hence, indeed the application of \cref{lm:dynamic} is valid and gives the desired time complexity.
    
    It remains to drop the assumption that $\hn \ge 5$ and $W \le \hn$ hold at the beginning of each epoch.
    We first deal with the second inequality.
    If $5 \le \hn < W$ holds at the beginning of the current epoch, we do not create an instance of $\bB$ during the current epoch, and in the next epoch instead of using $\bB$ to answer the queries, each time we simply compute a $w$-optimal alignment from scratch using \cref{lm:baseline-wed}.
    As $|X| + |Y| \le \hn$ during the upcoming two epochs as shown above, this algorithm takes at most $\Oh(\hn^2) \le \Oh(W^2) \le \Oh(W^2 \cdot d \log^2 n)$ time.
    Therefore, the updates in the next epoch are answered in the required time complexity.
    Note that during the current epoch, $\bA$ still can answer the updates as analyzed above.

    We finally drop the assumption that $n \ge 5$ holds at the beginning of every epoch.
    Whenever we have $\hn < 5$, we discard the algorithms $\mathbf{A}$ and $\mathbf{B}$ and, as long as $\hn < 5$, we simply compute a $w$-optimal alignment from scratch using \cref{lm:baseline-wed}.
    It takes time $\Oh(n^2) = \Oh(1) \le \Oh(W^2\cdot d\log^2 n)$.
    As soon as $\hn \ge 5$ holds again, we start the usual dynamic algorithm from scratch, which costs $\Oh(\hn \log^{\Oh(1)}\log \hn + W\cdot d^2\log^2 \hn) \le \Oh(W^2 \cdot d \log^2 n)$ time to initialize $\bA$, where we use that $\hn = \Oh(1)$ and $d \le W \cdot (|\hX| + |\hY|) = \Oh(W)$.
\end{proof}

As a corollary, we show the following variation of our dynamic algorithm.

\begin{corollary}\label{lm:dynamic-sleep}
    There exists a dynamic algorithm that, given a positive integer $\kt$ and oracle access to a weight function $w:\Esigma^2\to \fragmentcc{0}{W}$ and strings $X,Y\in \Sigma^*$, maintains $\wed_{\le \kt}(X, Y)$ subject to edits in $X$ and $Y$.
    Furthermore, at every point in time when $\wed(X, Y) \le \kt$, the algorithm also outputs the breakpoint representation of a $w$-optimal alignment $\cO : X\onto Y$.
    The algorithm takes $\Oh(n \log^{\Oh(1)}\log n + W\cdot d^2\log^2 n)$ time for initialization and $\Oh(W^2\cdot d\log^2 n)$ time for updates, where $n = \max\{2, |X| + |Y|\}$ and $d=\max\{1, \min\{\kt, \wed(X,Y)\}\}$ (here distance and string lengths are taken after the edit).
\end{corollary}

\begin{proof}
    On the high level, we maintain an algorithm of \cref{thm:dynamic} while $\wed(X, Y) = \Oh(\kt)$, and when we have $\wed(X, Y) \gg \kt$, we wait until the distance becomes small again to start maintaining the algorithm of \cref{thm:dynamic} from scratch.
    We now elaborate on the details.
    Note that as we are interested in computing $\wed_{\le \kt}$, we may replace the function $w$ with its capped version where we replace $w(a, b)$ with $\min\{w(a, b), \kt + 1\}$ for all $a, b \in \Esigma$.
    Hence, we may assume $W \le \kt + 1$.
    Moreover, note that by forcing the algorithm of \cref{lm:full_algorithm} to terminate after $\Theta(\kt^2 \cdot W \log^2 n)$ steps, we can compute $\wed_{\le 10 \kt}(X, Y)$ in time $\Oh(d^2 \cdot W \log^2 n)$ in the \modelname{} model.
    Furthermore, we may maintain \modelname{} over $X$ and $Y$ using \cref{lm:dynamic-pillar}.
    At the initialization, we compute $\wed_{\le 10 \kt}(X, Y)$, and if $\wed(X, Y) \le 10\kt$, we initialize the algorithm of \cref{thm:dynamic}.
    While $\wed(X, Y) \le 10\kt$, we proceed as in \cref{thm:dynamic}.
    When we have $\wed(X, Y) > 10 \kt$, we drop the algorithm of \cref{thm:dynamic} and start working in time blocks of $\floor{2\kt / W}$ updates called \emph{sleeping epochs}.
    Here $\floor{2\kt / W} \ge 1$ as $W \le \kt + 1$.
    We will later ensure that a sleeping epoch is never started with $\wed(X, Y) \le 4 \kt$.
    During such a sleeping epoch, we use the described above variation of \cref{lm:full_algorithm} to compute $\wed_{\le 10 \kt}(X, Y)$.
    Distributed over $\floor{2\kt / W}$ updates, it takes time $\Oh(\kt \cdot W^2 \log^2 n) = \Oh(d \cdot W^2 \log^2 n)$ per update.
    Here we use that during the sleeping epoch we have $d \ge 4 \kt - \floor{2\kt / W} \cdot W \ge 2 \kt$.
    If at the end of this sleeping epoch we still have $\wed(X, Y) > 6 \kt$, we continue the same way in the next sleeping epoch.
    Otherwise, during the next sleeping epoch, instead of running a variation of \cref{lm:full_algorithm}, we initialize the algorithm of \cref{thm:dynamic}.
    As \modelname{} is already initialized, distributed over $\floor{2\kt / W}$ updates, it takes time $\Oh(\kt \cdot W^2 \log^2 n) = \Oh(d \cdot W^2 \log^2 n)$ per update.
    Using \cref{thm:dynamic}, we can compute $\wed(X, Y)$ at the end of this sleeping epoch.
    If we have $\wed(X, Y) \le 6k$ we proceed as in \cref{thm:dynamic}.
    Otherwise, if $\wed(X, Y) > 6k$, we continue the sleeping epochs described above.
    As $d \ge 2 \kt$ holds during any sleeping epoch, it can never happen that $\wed(X, Y)$ becomes at most $\kt$ during some sleeping epoch.
    Furthermore, while we maintain \cref{thm:dynamic}, we always have $\wed(X, Y) \le \Oh(\kt)$, and thus we get the desired time complexity.
\end{proof}

\bibliography{refs}

\appendix

\section{Deferred Proofs From \cref{sec:fast-monge-matrix-multiplication}} \label{sec:monge-appendix}

\subsection{Matrix Capping} \label{sec:monge-appendix-cap}

This section is dedicated to proving the following lemma.

\lmreducematrix*

To achieve it, we first prove the following combinatorial claim.

\begin{lemma} \label{lm:limit-core}
    Consider some positive integers $p, q$, and $s$, and a set $P \subseteq [1 \dd p] \times [1 \dd q]$, such that if $(i, j) \in P$, then there are at most $s$ elements $(k, \ell) \in P$ (excluding $(i, j)$), such that $k \ge i$ and $\ell \le j$. Then, $|P| = \Oh(\sqrt{pqs})$.
\end{lemma}

\begin{proof}
    We call all elements of $[1 \dd p] \times [1 \dd q]$ cells.
    Say that a cell $(i, j) \in [1 \dd p] \times [1 \dd q]$ dominates another cell $(k, \ell) \in [1 \dd p] \times [1 \dd q]$ if $k \ge i$ and $\ell \le j$. Denote by $D$ the set of all cells that are either in $P$ or dominated by some cell from $P$ (in particular, if some cell $(i, j)$ is in $D$, then all cells dominated by $(i, j)$ are also in $D$). For each row $r$ of $[1 \dd p] \times [1 \dd q]$, denote by $x_r$ the maximum number, such that $(r, x_r)$ is in $D$ ($x_r = 0$ if no cells in the $r$-th row are in $D$). Note that $(x_r)_{r \in [1\dd p]}$ is a non-decreasing sequence. Denote $R_r \coloneqq [r \dd p] \times [1 \dd x_r]$ for each $r \in [1 \dd p]$. Since $(r, x_r)$ is in $D$, this cell is either in $P$ or is dominated by some cell $(i, j)$ from $P$. Thus, such a cell $(i, j)$ from $P$ dominates all cells in $R_r$, so there are at most $s + 1$ cells from $P$ in each $R_r$.
    Hence, $\sum_r |R_r \cap P| \le p \cdot (s + 1)$.

    Denote $\#_c \coloneqq |\{i \in [1 \dd p] \mid (i, c) \in P\}|$ for every $c \in [1 \dd q]$. Say that if we order the elements of $\{(i, j) \in P \mid j = c\}$ in the increasing order of $i$, we get a sequence $(r_1, c)$, $(r_2, c)$, $\ldots$, $(r_{\#_c}, c)$. As $(r_1, c) \in P$, we get that $x_{r_1} \ge c$. As $(x_r)_r$ is a non-decreasing sequence, $x_j \ge c$ for every $j \in [r_1 \dd p]$ follows. Thus, for any $i \in [1 \dd \#_c]$, the cell $(r_i, c)$ lies in at least $r_i - r_1 + 1 \ge i$ rectangles $R_j$ for $j \in [r_1 \dd r_i]$. Consequently, all cells of $P$ from the $c$-th column lie in at least $1 + 2 + \dots + \#_c = {\#_c + 1 \choose 2}$ rectangles $R_j$ in total (if two cells lie in the same rectangle, it is counted twice). Note that $\sum_r |R_r \cap P|$ is equal to the sum over all elements of $P$ of the number of rectangles they lie in. Thus, we get that \[\sum_{c \in [1 \dd q]} {\#_c + 1 \choose 2} \le \sum_r |R_r \cap P| \le p \cdot (s + 1).\]
    As ${\#_c + 1 \choose 2} \ge \#_c^2 / 2$, we get that $\sum_{c \in [1 \dd q]} \#_c^2 \le 2 p \cdot (s + 1)$.
    By the inequality between the arithmetic mean and the quadratic mean, we get $\sum_{c \in [1 \dd q]} \#_c^2 \ge (\sum_{c \in [1 \dd q]} \#_c)^2 / q$. Thus, $|P| = \sum_{c \in [1 \dd q]} \#_c \le \sqrt{q \cdot 2 p \cdot (s + 1)} = \Oh(\sqrt{pqs})$.
\end{proof}

We now use this combinatorial claim to prove \cref{lm:reduce_matrix} for two restricted classes of matrices.

\begin{lemma} \label{lm:reduce_simple_matrix}
    There is an algorithm that, given some positive integer $k$ and the condensed representation of some Monge matrix $C \in \mathbb{Z}_{\ge 0}^{p \times q}$, where all entries in the downmost row and the leftmost column of $C$ are zeros, in time $\Oh(N + \delta(C) \log N)$ builds the condensed representation of some Monge matrix $C' \in \mathbb{Z}_{\ge 0}^{p \times q}$ such that $\delta(C') \le \Oh(N \sqrt k)$ and $C' \meq{k} C$.
\end{lemma}

\begin{proof}
    The matrix $C$ is Monge, thus, for any $i \in [1 \dd p]$ and $j \in [1 \dd q]$, we have $C_{i, j} = C_{i, 1} + C_{p, j} - C_{p, 1} + \ssum(\dens{C}[i \dd p)[1 \dd j)) = 0 + 0 - 0 + \ssum(\dens{C}[i \dd p)[1 \dd j)) = \ssum(\dens{C}[i \dd p)[1 \dd j))$ by the definition of $\dens{C}$.

    Define $C'$ as a Monge matrix, in which values in the downmost row and the leftmost column are zeros, and $\core{C'}$ is a subset of $\core{C}$, such that an element $(i', j', v) \in \core{C}$ is in $\core{C'}$ if and only if $\ssum(\dens{C}[i' \dd p)[1 \dd j']) - v \le k$. Hence, $C'_{i, j} = \ssum(\dens{C'}[i \dd p)[1 \dd j))$ for all $i \in [1 \dd p]$ and $j \in [1 \dd q]$, where $\dens{C'}$ is defined via $\core{C'}$.

    Let us prove that $C' \meq{k} C$. If $C'_{i, j} \neq C_{i, j}$, we must have deleted some core entry $(i',j',v)\in \core{C}$ with $(i',j')\in [i \dd p) \times [1 \dd j)$.
    Consider the leftmost (and the downmost out of leftmost ones) such deleted core entry.
    Since we deleted the entry, it must satisfy $\ssum(\dens{C}[i' \dd p)[1 \dd j'])-v> k$.
    Moreover, no other core entry in $[i'\dd p)\times [1\dd j']$ was deleted, so $\ssum(\dens{C'}[i' \dd p)[1 \dd j'])= \ssum(\dens{C}[i' \dd p)[1 \dd j']) - v > k$. 
    Finally, we note that $C_{i,j}=\ssum(\dens{C}[i \dd p)[1 \dd j)) \ge \ssum(\dens{C}[i' \dd p)[1 \dd j']) > k+v>k$ and
    $C'_{i,j}=\ssum(\dens{C'}[i \dd p)[1 \dd j)) \ge \ssum(\dens{C'}[i' \dd p)[1 \dd j']) > k$.
    Hence, $\min\{C_{i, j},C'_{i, j}\} >k$.
    We conclude that $C'_{i, j} \meq{k} C_{i, j}$, thus proving $C' \meq{k} C$.

    Furthermore, any core entry of $C'$ has at most $k$ other core entries of $C'$ to the left and down of it as otherwise it would be deleted. Hence, the set $\{(i, j) \mid (i, j, v) \in \core{C'}\}$ satisfies the conditions of \cref{lm:limit-core} for $s = k$, which lets us conclude that $\delta(C') = \Oh(N \sqrt k)$.

    In time $\Oh(\delta(C) \log N)$, we build a data structure of \cref{fct:WORC} over the elements of $\core{C}$, after which, for each element $(i, j, v) \in \core{C}$, we can check in time $\Oh(\log N)$ whether it is a part of $\core{C'}$ or not.
    Hence, we calculate $\core{C'}$.
    The values in the leftmost column of $C'$ are zeros.
    As $C'_{i, j} = \ssum(\dens{C'}[i \dd p)[1 \dd j))$ for all $i \in [1 \dd p]$ and $j \in [1 \dd q]$, the values in the topmost row of $C'$ can be computed in $\Oh(\delta(C') + N)$ time by computing $\ssum(\dens{C'}[1 \dd p)[j \dd j])$ for all $j \in [1 \dd q)$ by scanning $\core{C'}$.
    Therefore, we obtain the condensed representation of $C'$ in time $\Oh(N + \delta(C) \log N)$.
\end{proof}

\begin{lemma} \label{lm:reduce_matrix-prime}
    There is an algorithm that, given some positive integer $k$ and the condensed representation of some Monge matrix $C \in \mathbb{Z}_{\ge 0}^{p \times q}$ such that row and column minima in $C$ are equal to zero, in time $\Oh((N + \delta(C)) \log N)$builds the condensed representation of some Monge matrix $C' \in \mathbb{Z}_{\ge 0}^{p \times q}$ such that $\delta(C') \le \Oh(N \sqrt k)$ and $C' \meq{k} C$.
\end{lemma}

\begin{proof}
    Define $z_i$ for $i \in [1 \dd p]$ as the leftmost zero in the $i$-th row of $C$. Define a matrix $F \in \mathbb{Z}_{\ge 0}^{p \times q}$, where $F_{i, j} = C_{i, j}$ if $j > z_i$, and $F_{i, j} = 0$ otherwise. Define $G \in \mathbb{Z}_{\ge 0}^{p \times q}$, where $G_{i, j} = C_{i, j}$ if $j < z_i$, and $G_{i, j} = 0$ otherwise. Note that for any $i \in [1 \dd p]$ and $j \in [1 \dd q]$, $\min\{F_{i, j}, G_{i, j}\} = 0$ and $F_{i, j} + G_{i, j} = \max\{F_{i, j}, G_{i, j}\} = C_{i, j}$.

    \begin{claim}\label{clm:split-monge-mtx}
         All entries in the leftmost column and the downmost row of $F$ are zeros, and all entries in the topmost row and the rightmost column of $G$ are zeros. Furthermore, matrices $F$ and $G$ are Monge.
    \end{claim}

    \begin{claimproof}
        We prove the lemma for matrix $F$. The proof for matrix $G$ is analogous.

        By the definition of matrix $F$, we have $F_{i, j} = C_{i, j}$ if $j > z_i$, and $F_{i, j} = 0$ otherwise. As $z_i \ge 1$ for any $i$, for $F_{i, 1}$, the second defining equality applies. Thus, all entries in the leftmost column of $F$ are zeros.

        Define $z'_i$ for $i \in [1 \dd p]$ as the rightmost zero in the $i$-th row of $C$. We claim that all entries between $z_i$-th and $z'_i$-th in the $i$-th row of $C$ are zeros. Assume for the sake of contradiction that there is some $i \in [1 \dd p]$ and $j \in (z_i \dd z'_i)$ such that $C_{i, j} > 0$. Matrix $C$ is Monge, thus the topmost column minima positions are non-decreasing. $C_{i, z'_i} = 0$, so it is a column minimum for the $z'_i$-th column. Thus, the topmost minimum position in this column is $\le i$, and thus the topmost minimum position $(i', j)$ in the $j$-th column satisfies $i' \le i$ (and $C_{i', j} = 0$ as all column minima are zeros). As $C_{i, j} \neq 0$, we have $i' < i$. Therefore, we arrive at a contradiction as $C_{i', j} + C_{i, z_i} = 0 < C_{i, j} \le C_{i, j} + C_{i', z_i}$ violates the fact that $C$ is Monge. Hence, indeed all entries between $z_i$-th and $z'_i$-th in the $i$-th row of $C$ are zeros.

        Furthermore, we claim that $z_{i + 1} \le z'_i + 1$. Assume for the sake of contradiction that there is some $i \in [1 \dd p)$, such that $z_{i + 1} > z'_i + 1$. We have $C_{i + 1, z_{i + 1}} = 0$, thus the topmost column minimum in the $z_{i + 1}$-st column is in a row with index at most $i+1$. The topmost minima positions are non-decreasing, thus topmost minimum in the $z'_i+1$-st column appears in row $i' \le i+1$ ($C_{i', z'_i+1}=0$ as all column minima are zeros). As both $i$-th and $i+1$-st rows do not contain zero in the $z'_i+1$-st column, we have $i' \le i-1$.
        Therefore, we arrive at a contradiction as $C_{i', z'_i+1} + C_{i, z_{i}} = 0 < C_{i, z'_i+1} \le C_{i, z'_i+1} + C_{i', z_{i}}$ violates the fact that $C$ is Monge. Hence, indeed $z_{i+1} \le z'_i + 1$ for all $i$.

        Column minimum in the rightmost column is zero, thus there is some $i$, for which $z'_i = q$. The sequence $z'$ is non-decreasing as $C$ is Monge.
        Therefore, $z'_p=q$ follows. By the definition of matrix $F$, $F_{i, j} = C_{i, j}$ if $j > z_i$, and $F_{i, j} = 0$ otherwise. As all entries in the $i$-th row between $z_i$-th and $z'_i$-th are zeros, we may write an alternative definition: $F_{i, j} = C_{i, j}$ if $j > z'_i$, and $F_{i, j} = 0$ otherwise. We have $z'_p=q$, so for any $F_{p, j}$ the second defining equality applies. Thus, all entries in the downmost row of $F$ are zeros.

        It remains to show that $F$ is Monge. It is sufficient to check Monge property for all contiguous $2 \times 2$ submatrices of $F$, so $F_{i, j} + F_{i + 1, j + 1} \le F_{i, j + 1} + F_{i + 1, j}$ for all $i \in [1 \dd p)$ and $j \in [1 \dd q)$. If all the summands in the inequality are zeros or all are equal to the entries from $C$, the Monge property clearly holds. Thus, it remains to check it for submatrices, in which some but not all non-zero values from $C$ are preserved in $F$. That is, there is at least one non-zero entry in this submatrix in $C$ that is preserved and at least one that is not preserved. These two entries can not be neighbors in the same row as between preserved and not preserved positive values in a row there should be at least one zero element. Analogously, they can not be neighbors in the same column. So these two elements lie on the same diagonal. Then the remaining two elements should be zeros in $C$, as otherwise they would be either preserved or not preserved, but, in any case, two adjacent elements, where one is preserved and the other is not, would exist, which is impossible, as we have already discussed. If entries $C_{i, j}$ and $C_{i + 1, j + 1}$ are zeros, then in any case Monge property is still valid in $F$. If entries $C_{i, j + 1}$ and $C_{i + 1, j}$ are zeros, then Monge property would be violated in $C$. Thus, all contiguous $2 \times 2$ submatrices of $F$ satisfy Monge property, and so $F$ is Monge.
    \end{claimproof}

    In time $\Oh(N + \delta(C) \log N)$ we build $\mds(C)$ that provides $\Oh(\log N)$-time random access to the values of $C$.
    In time $\Oh(N \log N)$ we find values $z_i$ for all $i \in [1 \dd p]$ using \cref{thm:smawk-row-maxima} and $\mds(C)$.
    Given the condensed representation of $C$ and the values $z_i$, in time $\Oh(N + \delta(C))$ we may build the condensed representations of $F$ and $G$.
    We have $\delta(F), \delta(G) \le \Oh(\delta(C) + N)$.
    We now solve the problem for matrices $F$ (using Lemma \ref{lm:reduce_simple_matrix}) and $G$ (using Lemma \ref{lm:reduce_simple_matrix} for $G^T$ and Lemma \ref{lm:ds_tools} after that) to obtain the condensed representations of $F'$ and $G'$ for some appropriate $F' \meq{k} F$ and $G' \meq{k} G$ in time $\Oh((N + \delta(C)) \log N)$. As $C = F + G$, we get that $C' \coloneqq F' + G' \meq{k} C$ due to \cref{c-equiv-preserv}.
    Furthermore, matrix $C'$ is Monge as the sum of two Monge matrices $F'$ and $G'$.
    The condensed representation of $C'$ can be obtained in time $\Oh(\delta(F') + \delta(G') + N) \le \Oh(\delta(C) + N)$.

    The overall time complexity is $\Oh((\delta(C) + N) \log N)$.
    Furthermore, $\delta(C') \le \Oh(\delta(F') + \delta(G') + N) \le \Oh(N \sqrt k)$.
\end{proof}

We now finally prove the same claim in the general case.

\begin{proof}[Proof of \cref{lm:reduce_matrix}]
    In time $\Oh(N + \delta(C) \log N)$ we build $\mds(C)$ that provides $\Oh(\log N)$-time random access to the entries of $C$.
    We apply \cref{thm:smawk-row-maxima} to compute row-minima $r_i$ of $C$ in time $\Oh(N \log N)$.
    Define a Monge matrix $D$, such that $D_{i, j} \coloneqq C_{i, j} - r_i$ for any $i \in [1 \dd p]$ and $j \in [1 \dd q]$.
    Note that the values $r_i$ and $\mds(C)$ provide $\Oh(\log N)$-time random access to the values of $D$.
    We now use this random access and \cref{thm:smawk-row-maxima} to compute column-minima $c_j$ of $D$ in time $\Oh(N \log N)$.
    Define a Monge matrix $E$, such that $E_{i, j} \coloneqq D_{i, j} - c_j = C_{i, j} - r_i - c_j$.
    Note that $\core{E} = \core{D} = \core{C}$, and thus the condensed representation of $E$ can be obtained in time $\Oh(N + \delta(C))$.
    We use \cref{lm:reduce_matrix-prime} to build the condensed representation of $E'$ for some Monge matrix $E' \in \mathbb{Z}_{\ge 0}^{p \times q}$, such that $\coresiz(E') \le \Oh(N \sqrt k)$ and $E' \meq{k} E$ in time $\Oh((\delta(C) + N) \log N)$.
    Define matrix $C'$, such that $C'_{i, j} \coloneqq E'_{i, j} + r_i + c_j$.
    Note that $\core{C'} = \core{E'}$, and thus $\delta(C') \le \Oh(N \sqrt k)$.
    Furthermore, the condensed representation of $C'$ can be obtained in time $\Oh(N + \delta(E')) \le \Oh(N + \delta(C))$.
    Finally, as $E' \meq{k} E$, we have $C' \meq{k} C$ due to \cref{c-equiv-preserv} ($C = E + M$ and $C' = E' + M$ for matrix $M$ with $M_{i, j} = r_i + c_j$).
    The overall time complexity is $\Oh((\delta(C) + N) \log N)$.
\end{proof}

\subsection{Monge Stitching} \label{sec:monge-appendix-capped-stitching}

In this section we prove \cref{capped-matrix-stitching}.

\cappedmatrixstitching*

\begin{proof}
    Consider $\bC = (A^T \mid B'^T)^T$.
    Note that $\bC \meq{k} C$.
    Furthermore, all entries of $\dens{\bC}$ in all rows except for $n_1$-st correspond to some entries of $\dens{A}$ and $\dens{B'}$, and thus are non-negative.
    Moreover, we can find all entries in the $n_1$-st row of $\dens{\bC}$ in time $\Oh(m \log N)$ by accessing all elements in the downmost row of $A$ and the topmost row of $B'$ through $\mds(A)$ and $\mds(B')$.
    If all these entries are non-negative, we define $C' \coloneqq \bC$ and use \cref{lm:ds_tools} to build $\mds(\bC)$ in time $\Oh((N + \delta(A) + \delta(B')) \log N)$ as $\bC = (A^T \mid B'^T)^T$.
    It is clear that such a matrix $C'$ satisfies all the required conditions.

    Otherwise, let $j^*$ be the smallest index in $[1 \dd m)$ such that $\dens{\bC}_{n_1, j^*} < 0$.
    Note that $j^* \ge \bj$ as $B_{1, j} = B'_{1, j}$ for all $j \le \bj$.
    As we have access to $\core{A}$ and $\core{B'}$ via $\mds(A)$ and $\mds(B')$, we can build $\core{\bC}$ in time $\Oh(\delta(A) + \delta(B') + N)$.
    
    Furthermore, in time $\Oh(\delta(A) + \delta(B') + N)$ we can build
    \[S \coloneqq \{(i, j, v) \in \core{\bC} \mid i \neq n_1 \text{ or } j < j^*\} \cup \{(n_1, j^*, k + 1 + \bC_{n_1 + n_2, 1})\}.\]
    We claim that for any $(i, j, v) \in S$, we have $v \ge 0$.
    Note that it holds for $(n_1, j^*, k + 1 + \bC_{n_1 + n_2, 1})$.
    It remains to show it for all entries of $\core{\bC}$ that are in $S$.
    As all entries of $\dens{\bC}$ in all rows except for $n_1$-st correspond to some entries of $\dens{A}$ and $\dens{B'}$, they are non-negative as $A$ and $B'$ are Monge.
    Furthermore, by the definition of $j^*$, all entries $(n_1, j, v) \in \core{\bC}$ with $j < j^*$ satisfy $v \ge 0$ as $j^*$ is the smallest index, for which $\dens{\bC}_{n_1, j^*}$ is negative.
    Thus, indeed all elements $(i, j, v) \in S$ satisfy $v \ge 0$.

    We now consider a Monge matrix $C' \in \mathbb{Z}^{(n_1 + n_2) \times m}$ such that $C'_{i, 1} = \bC_{i, 1}$ for $i \in [1 \dd n_1 + n_2]$, $C'_{n_1 + n_2, j} = \bC_{n_1 + n_2, j}$ for $j \in [1 \dd m]$, and $\core{C'} = S$.
    That is, we have
    \[C'_{i, j} = \bC_{i, 1} + \bC_{n_1 + n_2, j} - \bC_{n_1 + n_2, 1} + \sum_{(a, b, v) \in S\text{, s.t. } a \ge i, b < j} v\]
    for any $i \in [1 \dd n_1 + n_2]$ and $j \in [1 \dd m]$.
    Such a matrix is Monge as $v \ge 0$ for all $(i, j, v) \in S$.
    Note that we can calculate $C'_{n_1 + n_2, j} = \bC_{n_1 + n_2, j} = B'_{n_2, j}$ for all $j \in [1 \dd m]$ in time $\Oh(m \log N)$ using $\mds(B')$.
    Thus, using the formula from the definition of $C'$, we can calculate $C'_{1, j}$ for all $j \in [1 \dd m]$ in time $\Oh((m + |S|) \log N) = \Oh((\delta(A) + \delta(B') + N) \log N)$.
    Furthermore, we have $\core{C'} = S$ explicitly (in particular, $\delta(C') = \Oh(N + \delta(A) + \delta(B'))$).
    Therefore, we can use \cref{lm:build_mtx_ds} to build $\mds(C')$ in time $\Oh(N + (\delta(A) + \delta(B') + N) \log N)$.
    The whole algorithm takes time $\Oh((\delta(A) + \delta(B') + N) \log N)$.

    \smallskip

    It remains to show that $C'_{i, 1} = C_{i, 1}$ for $i \in [1 \dd n_1]$, all entries of $C'$ are non-negative, and $C' \meq{k} C$.
    The first claim obviously follows as $C'_{i, 1} = \bC_{i, 1} = A_{i, 1} = C_{i, 1}$ for all $i \in [1 \dd n_1]$.
    As the leftmost column and the downmost row of $C'$ coincide with the leftmost column and the downmost row of $\bC$, and as core entries of $C'$ and $\bC$ outside of $[1 \dd n_1] \times [j^* \dd m]$ coincide, we have that $C'_{i, j} = \bC_{i, j} \ge 0$ for all $(i, j) \notin [1 \dd n_1] \times [j^* + 1 \dd m]$.
    Furthermore, we have $C'_{i, j} = \bC_{i, j} \meq{k} C_{i, j}$.

    Therefore, it remains to prove the claim for $C'_{i, j}$ for $(i, j) \in [1 \dd n_1] \times [j^* + 1 \dd m]$.
    Note that in this case $n_1 \ge i$, $j^* < j$, and $(n_1, j^*, k + 1 + \bC_{n_1 + n_2, 1}) \in S$.
    Therefore, by the definition of $C'$, we have
    \begin{align*}
        C'_{i, j} &= \bC_{i, 1} + \bC_{n_1 + n_2, j} - \bC_{n_1 + n_2, 1} + \sum_{(a, b, v) \in S\text{, s.t. } a \ge i, b < j} v\\
                  &\ge \bC_{i, 1} + \bC_{n_1 + n_2, j} - \bC_{n_1 + n_2, 1} + (k + 1 + \bC_{n_1 + n_2, 1})\\
                  &= \bC_{i, 1} + \bC_{n_1 + n_2, j} + k + 1\\
                  &> k.
    \end{align*}

    Furthermore, we claim that $C_{i, j} > k$ holds for all $(i, j) \in [1 \dd n_1] \times [j^* + 1 \dd m]$.
    Note that $\dens{\bC}_{n_1, j^*} < 0$ by the definition of $j^*$.
    Hence, $\dens{\bC}_{n_1, j^*} \neq \dens{C}_{n_1, j^*}$ as $C$ is Monge.
    Therefore, one of the entries of $\bC$ in $[n_1 \dd n_1 + 1] \times [j^* \dd j^* + 1]$ is not equal to the corresponding entry of $C$.
    Let this entry be $(i', j')$.
    As $\bC \meq{k} C$, it follows that $\bC_{i', j'}, C_{i', j'} > k$.
    By the definition of $\bj$, we have $C_{i, j} \le C_{i - 1, j}$ for all $(i, j) \in (1 \dd n_1 + 1] \times [\bj \dd m]$ and $C_{i, j} \le C_{i, j + 1}$ for all $(i, j) \in [1 \dd n_1 + 1] \times [\bj \dd m)$.
    We have $i' \le n_1 + 1$ and $j' \ge j^* \ge \bj$.
    Therefore, we obtain $C_{i, j} \ge C_{i', j'} > k$ for all $(i, j) \in [1 \dd i'] \times [j' \dd m]$.
    In particular, it holds for all $(i, j) \in [1 \dd n_1] \times [j^* + 1 \dd m]$ as $n_1 \le i'$ and $j^* + 1 \ge j'$.

    As $C'_{i, j}> k$ and $C_{i, j} > k$, we obtain $C'_{i, j} \meq{k} C_{i, j}$, thus proving the claim.
\end{proof}

\section{Deferred Proofs From \cref{sec:hierarchical}} \label{sec:hierarchical-appendix}

\subsection{Properties of the Augmented Alignment Graph} \label{subsec:hierarchical-appendix1}

In this section, we prove \cref{grid-graph-extension}.

\gridgraphextension*

\begin{proof}

    \textbf{Monotonicity.}
    Consider any shortest path $(x, y) = (x_0, y_0) \to (x_1, y_1) \to \dots \to (x_m, y_m) = (x', y')$ from $(x, y)$ to $(x', y')$ in $\oAGw(X, Y)$. If its first and second coordinates are monotone, we are done. Otherwise, without loss of generality assume that $(x_i)_{i \in [0 \dd m]}$ is not monotone. Every two adjacent elements in this sequence differ by at most one, thus if such a sequence is not monotone, there exist two indices $i$ and $j$ such that $i < j - 1$, $x_i = x_j$, and $x_k \neq x_i$ for all $k \in (i \dd j)$. Here $j$ is the first index that breaks monotonicity, and $i$ is the last index before $j$ with $x_i = x_j$. We claim that by replacing the part of the path between $(x_i, y_i)$ and $(x_j, y_j)$ with a straight segment between these two points, we strictly decrease the weight of the path, thus the initial path could not be a shortest one. We consider two cases.

    \begin{enumerate}
        \item $y_i \ge y_j$. In this case, the straight segment $(x_i, y_i) \to (x_i, y_i - 1) \to \dots \to (x_i, y_j + 1) \to (x_i, y_j) = (x_j, y_j)$ uses only back edges and has weight $(W + 1) \cdot (y_i - y_j)$. On the other hand, any path from $(x_i, y_i)$ to $(x_j, y_j)$ should decrease the second coordinate of the current vertex at least $y_i - y_j$ times, and decreasing a coordinate is possible only using a back edge, thus any such path has at least $y_i - y_j$ back edges. Consequently, its weight is at least $(W + 1) \cdot (y_i - y_j)$, hence, the straight-line path is not longer. Furthermore, the old path should contain more than $y_i - y_j$ edges, because the only path between $(x_i, y_i)$ and $(x_i, y_j)$ of that length is a straight segment.
            All edges except for forward diagonal ones have strictly positive weights, and thus if there are any more of them, the weight of the old path is strictly larger than $(W + 1) \cdot (y_i - y_j)$, thus proving the claim. On the other hand, if there is a forward diagonal edge, it increases the second coordinate of the current vertex by one, thus there should be at least $y_i - y_j + 1$ back edges of total weight at least $(W + 1) \cdot (y_i - y_j + 1)$, thus proving the claim. Hence, indeed the old path is strictly more expensive than the straight segment path.
        \item $y_i < y_j$.
            The path $(x_i, y_i) \to (x_{i + 1}, y_{i + 1}) \to \dots \to (x_j, y_j)$ goes from a vertex with the second coordinate equal to $y_i$ to a vertex with the second coordinate equal to $y_j$. Thus, for every $\bar y \in [y_i \dd y_j)$, there exists some $k_{\bar y} \in [i \dd j)$ such that $y_{k_{\bar y}} = \bar y$ and $y_{k_{\bar y} + 1} = \bar y + 1$. Furthermore, let $c$ be the number of back edges on the path $(x_i, y_i) \to (x_{i + 1}, y_{i + 1}) \to \dots \to (x_j, y_j)$. The cost of this path is at least $(W + 1) \cdot c + \sum_{\bar y \in [y_i \dd y_j)} \w{(x_{k_{\bar y}}, y_{k_{\bar y}})}{(x_{k_{\bar y} + 1}, y_{k_{\bar y} + 1})}$ because all edges $(x_{k_{\bar y}}, y_{k_{\bar y}}) \to (x_{k_{\bar y} + 1}, y_{k_{\bar y} + 1})$ are distinct and forward. Furthermore, edges of form $(x_{k_{\bar y}}, y_{k_{\bar y}}) \to (x_{k_{\bar y} + 1}, y_{k_{\bar y} + 1})$ are either diagonal if $x_{k_{\bar y} + 1} = x_{k_{\bar y}} + 1$ or vertical if $x_{k_{\bar y} + 1} = x_{k_{\bar y}}$. For a vertical edge, its weight is equal to the weight of the edge $(x_i, y_{k_{\bar y}}) \to (x_i, y_{k_{\bar y}} + 1)$ from the straight segment path. For a diagonal edge, it increases the first coordinate of the current vertex on the path by one, and as $x_i = x_j$, the number of such increases is bounded by the number of back edges $c$ on this path because only they decrease the first coordinate of the current vertex on the path. Thus, there are $\le c$ edges $(x_i, y_{k_{\bar y}}) \to (x_i, y_{k_{\bar y}} + 1)$ of the straight-line path that correspond to diagonal edges in the old path.
    Each such edge has weight $\le W$. From this, it follows that the straight-line path is not longer than the old path.

    Furthermore, $c \ge 1$ because the only path between $(x_i, y_i)$ and $(x_i, y_j)$ without back edges is a straight line segment. Thus, the sum of all diagonal edges and back edges in the old path is at least $c \cdot (W + 1)$, while there are at most $c$ corresponding edges in the new path, each having length $\le W$, hence having total length $\le c \cdot W < c \cdot (W + 1)$. Thus, indeed the old path is strictly more expensive than the straight segment path.
\end{enumerate}

    As no path that is not monotone in both coordinates can be a shortest one, we prove the first claim of the lemma.

    \medskip

    \subparagraph*{Distance preservation.} Distance preservation follows from monotonicity as if $x \le x'$ and $y \le y'$, both coordinates on any shortest path are non-decreasing, and thus there are no back edges on this path. Hence, the same path exists in $\AGw(X, Y)$. Furthermore, clearly any path from $\AGw(X, Y)$ exists in $\oAGw(X, Y)$, thus $\dist_{\AG^w(X, Y)}((x, y), (x', y')) = \dist_{\overline{\AG}^w(X, Y)}((x, y), (x', y'))$.

    \medskip

    \subparagraph*{Path irrelevance.} Without loss of generality assume that $x \le x'$ and $y \ge y'$. Consider some path $(x, y) = (x_0, y_0) \to (x_1, y_1) \to \dots \to (x_m, y_m) = (x', y')$ from $(x, y)$ to $(x', y')$ in $\oAGw(X, Y)$ that is monotone in both coordinates. Such a path has $(x_i)_{i \in [0 \dd m]}$ non-decreasing and $(y_i)_{i \in [0 \dd m]}$ non-increasing. All edges that are monotone in this direction are horizontal forward edges and vertical back edges. Every such path consists of $y - y'$ vertical back edges of weight $W + 1$ each and a single horizontal forward edge of form $(\bar x, \bar y) \to (\bar x + 1, \bar y)$ for each $\bar x \in [x \dd x')$ for some $\bar y$. All horizontal edges of this form have equal weights for a fixed $\bar x$, thus the weight of the whole path is independent of the specific path we choose.
\end{proof}

\subsection{Boundary Distance Matrix Composition} \label{subsec:hierarchical-appendix2}

This section is dedicated to proving \cref{boundary-distance-matrix-combination}.
Before proving it, we show a helper lemma.

\begin{lemma} \label{lm:degenerate-distance-matrix}
    Consider strings $X, Y \in \Sigma^{\le n}$ and a weight function $w : \Esigma^2 \to [0 \dd W]$.
    Furthermore, consider a sequence $(r_i)_{i \in [1 \dd m]} \eqqcolon (x_i, y_i)_{i \in [1 \dd m]}$ of distinct vertices of $\oAGw(X, Y)$, where $(x_i)_{i \in [1 \dd m]}$ is non-decreasing and $(y_i)_{i \in [1 \dd m]}$ is non-increasing.
    Let $(p_i)_{i \in [1 \dd m_p]}$ and $(q_i)_{i \in [1 \dd m_q]}$ be two contiguous subsequences of $(r_i)_{i \in [1 \dd m]}$.
    Let $M$ be a matrix of size $m_p \times m_q$ such that $M_{i, j} = \dist_{\oAGw(X, Y)}(p_i, q_j)$ for all $i \in [1 \dd m_p]$ and $j \in [1 \dd m_q]$.

    Matrix $M$ is Monge with $\delta(M) = \Oh((x_m - x_1) + (y_1 - y_m))$.
    Furthermore, there is an algorithm that in time $\Oh(((x_m - x_1) + (y_1 - y_m) + 1) \log n)$ builds $\mds(M)$.
\end{lemma}

\begin{proof}
    Note that every pair $(p_i, q_j)$ satisfies the path irrelevance condition of \cref{grid-graph-extension}.
    Therefore, all paths between them that are monotone in both dimensions are the shortest ones.
    Fix some such path $P$ from $r_1$ to $r_m$ going through all $r_i$.
    We calculate the distances $\dist_{\oAGw(X, Y)}(p_i, q_j)$ along this path (and its inverse).
    The fact that $M$ is Monge follows directly from \cref{fct:monge} as all nodes $p_i$ and $q_j$ lie in the right order on the outer face of the subgraph of $\oAGw(X, Y)$ induced by $P$.
    
    Consider node $p_1$. Distances from it to all nodes $q_j$ can be found in time $\Oh((x_m - x_1) + (y_1 - y_m) + 1)$ by going along $P$ in both directions from $p_1$.
    Analogously, we can find distances from all nodes $p_i$ to $q_1$.
    Therefore, we find all entries of $M$ in the first row and the first column.

    We now compute the core of $M$.
    Consider the $i$-th row of $M$.
    We find all entries $(i', j', v) \in \core{M}$ with $i' = i$.
    Let $p_i$ be the $k$-th element of the sequence $r$.
    Therefore, $p_{i + 1} = r_{k + 1}$.
    Consider the difference $M_{i + 1, j} - M_{i, j} = \dist_{\oAGw(X, Y)}(p_{i + 1}, q_j) - \dist_{\oAGw(X, Y)}(p_{i}, q_j)$ for some fixed $j \in [1 \dd m_q]$.
    Let $q_j$ be the $\ell_j$-th element of the sequence $r$.
    If $\ell_j \le k$, we have $M_{i + 1, j} - M_{i, j} = \distag(p_{i + 1}, {p_i})$ as the path along $P$ becomes longer by this fragment.
    Otherwise, if $\ell_j \ge k + 1$, we have $M_{i + 1, j} - M_{i, j} = -\distag(p_i, p_{i + 1})$ as the path along $P$ becomes shorter by this fragment.
    Note that these differences do not depend on $j$.
    Therefore, if $\ell_j \ge k + 1$ for all $j \in [1 \dd m_q]$ or $\ell_j \le k$ for all $j \in [1 \dd m_q]$, then $\dens{A}_{i, j} = 0$ for all $j \in [1 \dd m_q]$, and thus there are no entries $(i', j', v) \in \core{M}$ with $i' = i$.
    Otherwise, let $j^*$ be the last index in $[1 \dd m_q)$ such that $\ell_{j^*} \le k$.
    We have that $\dens{M}_{i, j} = 0$ for all $j \neq j^*$ and $\dens{M}_{i, j^*} = \distag(p_{i + 1}, {p_i}) + \distag(p_i, p_{i + 1})$.
    Hence, $(i, j^*, \distag(p_{i + 1}, {p_i}) + \distag(p_i, p_{i + 1}))$ is the only entry $(i', j', v) \in \core{M}$ with $i' = i$.

    By going through rows of $M$ in the increasing order we can maintain $j^*$ on the go using two pointers.
    Hence, we can compute $\core{M}$ in time $\Oh((x_m - x_1) + (y_1 - y_m) + 1)$.
    Furthermore, $\delta(M) \le m_p - 1 \le m - 1 \le (x_m - x_1) + (y_1 - y_m)$.
    Moreover, as we have all entries of the first row and the first column of $M$ and $\core{M}$, we may apply \cref{lm:build_mtx_ds} to build $\mds(M)$ in time $\Oh(((x_m - x_1) + (y_1 - y_m) + 1) \log n)$.
\end{proof}

We are now ready to prove \cref{boundary-distance-matrix-combination}.

\boundarydistancematrixcombination*

\begin{proof}
    We prove the lemma for rectangles adjacent vertically.
    The claim for rectangles adjacent horizontally is analogous.

    We first prove the first statement of the lemma.
    Consider matrix $D$ of distances in $\oAGw(X, Y)$ between the left input vertices of $[a\dd b] \times [d+1\dd e]$ and the output vertices of $[a\dd b] \times [c\dd d]$.
    Due to \cref{lm:degenerate-distance-matrix}, matrix $D$ is Monge, $\delta(D) = \Oh(m)$, and we can build $\mds(D)$ in time $\Oh(m \log m)$.

    By applying \cref{lm:ds_tools}, we can build $\mds(E)$ for matrix $E \coloneqq (D^T \mid A^T)^T$ of distances in $\oAGw(X, Y)$ from the input vertices of $[a\dd b] \times [c\dd e]$ to the output vertices of $[a\dd b] \times [c\dd d]$ in time $\Oh(m \log m + \delta(A) \log m)$, where $(D^T \mid A^T)$ is Monge due to \cref{fct:monge}. The size of the core of $E$ is $\Oh(\delta(A) + m)$.

    Consider matrix $F$ of distances in $\oAGw(X, Y)$ from the input vertices of $[a\dd b] \times [c\dd e]$ to the left input vertices of $[a\dd b] \times [d + 1\dd e]$.
    Due to \cref{lm:degenerate-distance-matrix}, matrix $F$ is Monge, $\delta(F) = \Oh(m)$, and we can build $\mds(F)$ in time $\Oh(m \log m)$.

    By applying \cref{lm:ds_tools} we can build $\mds(G)$ for matrix $G \coloneqq (F \mid E)$ of distances in $\oAGw(X, Y)$ from the input vertices of $[a\dd b] \times [c\dd e]$ to the union of left input vertices of $[a\dd b] \times [d+1\dd e]$ and the output vertices of $[a\dd b] \times [c\dd d]$ in time $\Oh(m \log m + (\delta(A) + m) \log m)$, where $(F \mid E)$ is Monge due to \cref{fct:monge}. The size of the core of $G$ is $\Oh(\delta(A) + m)$.

    By an analogous procedure starting from matrix $B$, in time $\Oh(m \log m + \delta(B) \log m)$ we can build $\mds(H)$ for matrix $H$ of distances in $\oAGw(X, Y)$ from the union of left input vertices of $[a\dd b] \times [d+1\dd e]$ and the output vertices of $[a\dd b] \times [c\dd d]$ to the output vertices of $[a\dd b] \times [c\dd e]$, where $\coresiz(H) = \Oh(\delta(B) + m)$.

    We claim that $G \otimes H = C$. It follows from the fact that vertices of the union of left input vertices of $[a\dd b] \times [d+1\dd e]$ and the output vertices of $[a\dd b] \times [c\dd d]$ are a separator between the input and the output vertices of $[a\dd b] \times [c\dd e]$ in the subgraph of $\oAGw(X, Y)$ induced by $[a\dd b] \times [c\dd e]$, and thus, every path between them inside $[a\dd b] \times [c\dd e]$ goes through this separator. \cref{grid-graph-extension} implies that every shortest path between boundary vertices of $[a\dd b] \times [c\dd e]$ indeed lies inside $[a\dd b] \times [c\dd e]$. As $G$ contains distances from the input vertices of $[a\dd b] \times [c\dd e]$ to this separator in $\oAGw(X, Y)$ and $H$ contains distances from this separator to the output vertices of $[a\dd b] \times [c\dd e]$ in $\oAGw(X, Y)$, we get that $G \otimes H$ is a matrix of distances from the input to the output vertices of $[a\dd b] \times [c\dd e]$.
    Given $\mds(G)$ and $\mds(H)$, we apply \cref{lm:matrix_mult_w} to get $\mds(C)$ in time $\Oh((m + \delta(G) + \delta(H)) \log m) = \Oh((m + \delta(A) + \delta(B)) \log m)$. Thus, the total time complexity is $\Oh((m + \delta(A) + \delta(B)) \log m)$.

    \medskip

    We now prove the second statement of the lemma.
    We proceed the same way as in the first statement but emphasize the differences.
    First, we build $\mds(D)$ using \cref{lm:degenerate-distance-matrix} in time $\Oh(m \log m)$.
    We note that all entries of $E[1 \dd e - d + 1][1 \dd \perim([a \dd b] \times [c \dd d])]$ are monotonically decreasing by rows and monotonically increasing by columns due to the definition of $E$ as all left input vertices of $[a \dd b] \times [d \dd e]$ are to the left and down of all output vertices of $[a \dd b] \times [c \dd d]$.
    Furthermore, $A_{1, 1} = 0$ as it represents the distance in $\oAGw(X, Y)$ from $(a, d)$ to $(a, d)$.
    As $A' \meq{k} A$, we have $A'_{1, 1} = A_{1, 1} = 0$.
    Therefore, as $E = (D^T \mid A^T)^T$, by applying \cref{capped-matrix-stitching} for $k$, $\mds(D)$, $\mds(A')$, and $\bj=1$, in time $\Oh((m + \delta(A')) \log m)$ we can obtain $\mds(E')$ for some $E'$ such that $E' \meq{k} E$, $\delta(E') = \Oh(\delta(A') + m)$, and $E'_{i, 1} = E_{i, 1}$ for all $i \in [1 \dd e - d]$.

    We now build $\mds(F)$ in time $\Oh(m \log m)$ using \cref{lm:degenerate-distance-matrix}.
    We note that all entries of $G[e - d + 1 \dd \perim([a \dd b] \times [c \dd e])][1 \dd e - d + 1]$ are monotonically increasing by rows and monotonically decreasing by columns due to the definition of $G$ as all input vertices of $[a \dd b] \times [c \dd d]$ are to the right and up of all left input vertices of $[a \dd b] \times [d \dd e]$.
    Furthermore, $E'_{i, 1} = E_{i, 1}$ for all $i \in [1 \dd e - d]$.
    Moreover, $E_{e - d + 1, 1} = 0$ as it represents the distance in $\oAGw(X, Y)$ from $(a, d)$ to $(a, d)$.
    As $E' \meq{k} E$, we have $E'_{e - d + 1, 1} = E_{e - d + 1, 1} = 0$.
    Finally, we can obtain $\mds(F^T)$ and $\mds(E'^T)$ using  \cref{lm:ds_tools} in time $\Oh((m + \delta(A')) \log m)$.
    Therefore, as $G = (F \mid E)$, by applying \cref{capped-matrix-stitching} for $k$, $\mds(F^T)$, $\mds(E'^T)$, and $\bj = e - d + 1$, in time $\Oh((m + \delta(A')) \log m)$ we can obtain $\mds(G'')$ for some $G''$ such that $G'' \meq{k} G^T$ and $\delta(G'') = \Oh(\delta(A') + m)$.
    Using \cref{lm:ds_tools}, we can build $\mds(G')$ for $G' \coloneqq G''^T$ in time $\Oh((\delta(A') + m) \log m)$.
    We have $G' \meq{k} G$.
    The whole procedure of building $\mds(G')$ takes $\Oh((m + \delta(A')) \log m)$ time.

    By an analogous procedure starting from matrix $B'$, in time $\Oh((m+ \delta(B')) \log m)$ we can build $\mds(H')$ for some $H'$ such that $H' \meq{k} H$ and $\delta(H') = \Oh(\delta(B') + m)$.
    We now use \cref{lm:matrix_mult_k} to calculate $\mds(C')$ for some $C'$ such that $C' \meq{k} G' \otimes H' \meq{k} G \otimes H = C$ and $\delta(C') = \Oh(m \sqrt k)$ in time $\Oh((m + \delta(A) + \delta(B)) \log m)$.
    The overall time complexity is $\Oh((m + \delta(A) + \delta(B)) \log m)$.
\end{proof}

\subsection{Hierarchical Alignment Data Structure Implementation} \label{subsec:hierarchical-appendix3}

This section is dedicated to proving \cref{lm:boxdsslp-query,lm:boxdsslp-update}, which implement the alignment retrieval and the construction of the recursive data structures $\boxdsslp(A,B)$ and $\boxdskslp(A,B)$ of \cref{def:boxdsslp}.

\lmboxdsslpquery*

\begin{proof}
We focus on the case of $\boxdskslp$. 
The query algorithm for the non-relaxed version $\boxdsslp$ of the data structure can be obtained by setting $k=(W+1)\cdot (\len(A) + \len(B))$ below so that $k\ge d$. 

We have a pointer to $\mds(M)$ for some $M$ such that $M \meq{k} \BMw(A,B)$.
Note that $d = \dist_{\oAGw(A, B)}(p, q)$ is one of the entries of $\BMw(A,B)$. 
We read the corresponding entry of $M$ using the random access functionality of $\mds(M)$; if the value exceeds $k$, we conclude that $d>k$ and report a failure.
Otherwise, the corresponding entry of $M$ contains exactly $d\le k$.
Overall, we can retrieve $d$ in $\Oh(\log n)$ time.

If $d = 0$, there is a trivial shortest path from $p$ to $q$ always taking diagonal edges of weight zero.
We can return the breakpoint representation of such a path in constant time.
We henceforth assume $d > 0$ and, furthermore, $\len(A)\ge \len(B)$; the case $\len(A) < \len(B)$ is analogous.
If $\len(A) = \len(B)=1$, then the path can be constructed in $\Oh(1)$ time using oracle access to $w$.
Otherwise, we have pointers to $\boxdskslp(A_L, B)$ and $\boxdskslp(A_R, B)$, where $\rhs(A)\eqqcolon A_LA_R$.
Let $R_L \coloneqq [0 \dd \len(A_L)] \times [0 \dd \len(B)]$ and $R_R \coloneqq [\len(A_L) \dd \len(A)] \times [0 \dd \len(B)]$ be rectangles in $\oAGw(A, B)$.

If $p,q\in R_L$, then we interpret $p,q$ as vertices of $\oAGw(A_L, B)$ and recursively compute the breakpoint representation of a shortest path in $\oAGw(A_L, B)$ via $\boxdskslp(A_L, B)$; by \cref{grid-graph-extension}, the shortest path stays within $R_L$.
Similarly, if $p,q\in R_R$, then we recursively compute the breakpoint representation of a shortest path in $\oAGw(A_R, B)$.
Otherwise, if $p\in R_R$ and $q\in R_L$, then, by \cref{grid-graph-extension}, any monotone path from $p$ to $q$ is a shortest one. 
Thus, we pick any and construct its breakpoint representation.
This takes $\Oh(d)$ time because the path does not contain any cost-$0$ edges.

In the remaining case when $p\in R_L$ and $q\in R_R$, we consider a set $S \coloneqq \outp(R_L) \cap \inp(R_R)$ of vertices of $\oAGw(A, B)$.
Observe that $S$ is a separator between $R_L$ and $R_R$, so the shortest path from $p$ to $q$ goes through some vertex $v \in S$, and we have $\dist_{\oAGw(A, B)}(p, q) = \dist_{\oAGw(A, B)}(p, v) + \dist_{\oAGw(A, B)}(v, q)$. 
In particular, $\dist_{\oAGw(A, B)}(p, v), \dist_{\oAGw(A, B)}(v, q) \le \dist_{\oAGw(A, B)}(p, q) = d\le k$.

By \cref{lm:paths_dont_deviate_too_much}, we can identify $\Oh(d)$ vertices $v\in S$ for which $\dist_{\oAGw(A, B)}(p,v)\le d$ \emph{may} hold.
For each such vertex, we can compute $\min\{k+1,\dist_{\oAGw(A, B)}(p, v)\}$. 
Due to \cref{grid-graph-extension}, an optimal path from $p$ to $v$ lies completely inside $R_L$, and thus the distance from $p$ to $v$ is an entry of $\BMw(A_L, B)$.
From $\boxdskslp(A_L, B)$, we have a pointer to $\mds(M_L)$ for some $M_L \meq{k} \BMw(A_L, B)$.
Hence, by accessing $M_L$ using $\mds(M_L)$ in time $\Oh(\log n)$, we can find some value that is $k$-equivalent to $\dist_{\oAGw(A, B)}(p, v)$. If such a value is larger than $k$, we cap it with $k+1$, and if such a value is at most $k$, it is equal to $\dist_{\oAGw(A, B)}(p, v)$.
Analogously, in time $\Oh(\log n)$ we can compute $\min\{k+1,\dist_{\oAGw(A, B)}(v, q)\}$.
Overall, in $\Oh(d\log n)$ time we are able to identify a vertex $v\in S$ such that $\dist_{\oAGw(A, B)}(p, q) = \dist_{\oAGw(A, B)}(p, v) + \dist_{\oAGw(A, B)}(v, q)$.

We recursively compute the breakpoint representation of an optimal path from the input vertex of $\oAGw(A_L, B)$ corresponding to $p$ to the output vertex of $\oAGw(A_L, B)$  corresponding to $v$, 
and the breakpoint representation of an optimal path from the input vertex of $\oAGw(A_R, B)$ corresponding to $v$ to the output vertex of $\oAGw(A_R, B)$ corresponding to $q$.
We then reinterpret these breakpoint representations such that they represent a path from $p$ to $v$ and a path from $v$ to $q$ in $\oAGw(A, B)$.
We combine them to obtain the breakpoint representation of an optimal path from $p$ to $q$.

The overall running time can be expressed using recurrence relation $T(n,0)=\Oh(\log n)$ and $T(n,d)=\Oh(d\log n)+T(n_L,d_L)+T(n_R,d_R)$ for $d\ge 1$, where $n_L,n_R\le 0.9n$ (because the SLP is balanced) and $d_L+d_R=d$, so the total time complexity of the recursive procedure is $\Oh((d+1)\log^2 n)$.
\end{proof}

Before we proceed with a proof of \cref{lm:boxdsslp-update}, which constructs $\boxdsslp(A',B')$ based on $\boxdsslp(A,B)$, let us characterize the set of pairs $(C,D)$ for which $\boxdsslp(C,D)$ is a part of the recursive structure of $\boxdsslp(A, B)$.
Formally, this set can be defined as follows:

\begin{definition}
For symbols $A, B \in \S_{\slp}$, define the \emph{closure} $\closure{A, B}$ of $A$ and $B$ recursively:
\begin{itemize}
    \item If $\len(A) = \len(B) = 1$, let $\closure{A, B} = \set{(A, B)}$.
    \item If $\len(A) \ge \len(B)$ and $\len(A) > 1$, let $\closure{A, B} = \set{(A, B)} \cup \closure{A_L, B} \cup \closure{A_R, B}$, where $\rhs(A)\eqqcolon A_LA_R$.
    \item If $\len(A) < \len(B)$, let $\closure{A, B} = \set{(A, B)} \cup \closure{A, B_L} \cup \closure{A, B_R}$, where $\rhs(B)\eqqcolon B_LB_R$.\lipicsEnd
\end{itemize}    
\end{definition}

While constructing $\boxdsslp(A,B)$, we would like to re-use $\boxdsslp(C,D)$ whenever $(C,D)\in \closure{A,B}\cap \closure{A',B'}$, but it would be too expensive to identify such pairs by exploring the entire recursive structure of $\boxdsslp(A,B)$.
Thus, we provide a characterization of $\closure{A,B}$ based on how the weights (expansion lengths) of $C$ and $D$ compare to the weights of their ancestor symbols.

\begin{definition}\label{def:plen}
    Consider an SLP $\slp$ and a symbol $A \in \S_{\slp}$.
    For a node $\nu$ of $\Tr(A)$, define
    \[\plen(\nu) \coloneqq
    \begin{cases}
        \len(\mu) & \!\!\!\text{if $\nu$ has parent $\mu$},\\
        \infty & \!\!\!\text{if $\nu$ is the root}.
    \end{cases}
    \]
    Furthermore, for a symbol $C \in \dep(A)$, let $\plen_A(C) \coloneqq \setmax{\plen(\nu) \mid \nu \in \Tr(A), \symb(\nu) = C}$.
    If $\plen_A(C) < \infty$, let $\pre_A(C) \coloneqq \symb(\mu)$, where $\mu$ is the parent of $\nu^* = \argmax\set{\plen(\nu) \mid \nu \in \Tr(A), \symb(\nu) = C}$.
    \lipicsEnd
\end{definition}

\begin{lemma}\label{lem:lenplen}
Consider an SLP $\slp$.
For every pair of symbols $(A,B)\in \S_{\slp}^2$, the set $\closure{A,B}$ consists precisely of pairs $(C, D) \in \dep(A) \times \dep(B)$ such that $\len(D)\le \plen_A(C)$ and $\len(C)<\plen_B(D)$.
Furthermore, for a pair $(C, D) \in \closure{A, B}$, we have $(\pre_A(C), D) \in \closure{A, B}$ if $\plen_A(C) < \plen_B(D)$, and $(C, \pre_B(D)) \in \closure{A, B}$ if $\plen_A(C) \ge \plen_B(D) \neq \infty$.
\end{lemma}

\begin{proof}
Let us first prove that $(C,D)\in \closure{A,B}$ holds for every pair $(C, D) \in \dep(A) \times \dep(B)$ of symbols satisfying $\len(D)\le \plen_A(C)$ and $\len(C)<\plen_B(D)$, with $(\pre_A(C), D) \in \closure{A, B}$ if $\plen_A(C) < \plen_B(D)$, and $(C, \pre_B(D)) \in \closure{A, B}$ if $\plen_A(C) \ge \plen_B(D) \neq \infty$.
We proceed by induction from larger to smaller values of $\len(C)+\len(D)$:
\begin{enumerate}
    \item If $\plen_A(C)=\plen_B(D)=\infty$, then $(C,D)=(A,B)$ belongs to $\closure{A,B}$ by definition. 
    \item If $\plen_A(C)\ge \plen_B(D) \ne \infty$, then the parent $D'\coloneqq \pre_B(D)$ of $D$ by definition satisfies $\len(D')=\plen_B(D)\le \plen_A(C)$.
    Moreover, $\len(C)<\plen_B(D)=\len(D')<\plen_B(D')$ holds the by assumption on $(C,D)$.
    Thus, the inductive assumption implies $(C,D')\in \closure{A,B}$.
    Since $\len(D')>\len(C)$ and $D$ is a child of $D'$, the definition of closed sets implies $(C,D)\in \closure{A,B}$. 
    \item If $\plen_A(C)<\plen_B(D)$, then the parent $C'\coloneqq \pre_A(C)$ of $C$ by definition satisfies $\len(C')=\plen_A(C)<\plen_B(D)$.
    Moreover, $\len(D)\le \plen_A(C)=\len(C')<\plen_A(C')$. 
    Thus, the inductive assumption implies $(C',D)\in \closure{A,B}$.
    Since $\len(C')\ge\len(D)$ and $C$ is a child of $C'$, the definition of closed sets implies $(C,D)\in \closure{A,B}$.  
\end{enumerate}
Consequently, $(C, D) \in \closure{A, B}$ holds for every pair $(C, D) \in \dep(A) \times \dep(B)$ satisfying $\len(D)\le \plen_A(C)$ and $\len(C)<\plen_B(D)$.

Furthermore, it is clear that $\closure{A, B} \subseteq \dep(A) \times \dep(B)$ by definition.
It remains to prove that every element $(C, D)\in \closure{A,B}$ satisfies $\len(D)\le \plen_A(C)$ and $\len(C)<\plen_B(D)$.
Again, we proceed by induction from larger to smaller values of $\len(C)+\len(D)$ and observe that every element $(C, D)\in \closure{A,B}$ other than $(A,B)$
has a \emph{parent} element $(C',D')\in \closure{A,B}$ such that either $D=D'$ and $C$ is a child of $C'$, or $C=C'$ and $D$ is a child of $D'$.
\begin{enumerate}
    \item If $(C,D)=(A,B)$ then $plen_A(C)=\plen_B(D)=\infty$, so the condition is satisfied trivially.
    \item If $(C,D)$ has a parent $(C',D)\in \closure{A,B}$ such that $C$ is a child of $C'$,  then $\len(D)\le \len(C')$, and thus $\len(D)\le \len(C')\le \plen_A(C)$. The inductive assumption further guarantees that $\len(C)<\len(C')<\plen_B(D)$.
    \item If $(C,D)$ has a parent $(C,D')\in \closure{A,B}$ such that $D$ is a child of $D'$, then $\len(C)<\len(D')$, and thus $\len(C)<\len(D')\le \plen_B(D)$. The inductive assumption further guarantees that $\len(D)<\len(D')\le \plen_A(C)$.\qedhere
\end{enumerate}
\end{proof}

Another application of \cref{lem:lenplen} is to bound the time needed to construct $\mds(\BMw(C,D))$ for all $(C,D)\in \closure{A',B'}\setminus \closure{A,B}$.
For this, we use the following helper lemma.

\begin{lemma}\label{cor:newincl}
    Let $\slp$ be a weight-balanced SLP and let $A,A',B\in \S_\slp$. 
    Write $u \coloneqq |\dep(A') \setminus \dep(A)|$.
    We have
    \[\sum_{(C,D)\in \closure{A',B}\setminus \closure{A,B}}\left(\len(C)+\len(D)\right) = \Oh\left((u+1)(\len(A')+\len(B))+\len(B)\log\left(1+\tfrac{\len(B)}{\len(A')}\right)\right).\]
\end{lemma}

\begin{proof}
    Let $\cU \coloneqq \dep(A') \setminus \dep(A)$.
    Moreover, let $\cV\subseteq \S_\slp$ consist of $A'$ as well as all symbols that occur in $\rhs(U)$ for $U\in \cU\cap \N_\slp$.
    Observe that $|\cV|\le 2u+1$ since $|\rhs(U)|\le 2$ holds for every $U\in \N_\slp$.

    \begin{claim} 
        Every $(C,D)\in \closure{A',B}\setminus \closure{A,B}$ satisfies $C\in \cV$.
    \end{claim}
    \begin{claimproof}
        For a proof by contradiction, suppose that $C\notin \cV$.
        As $A' \in \cV$, we have $C \neq A'$.
        Let $C'\coloneqq \pre_{A'}(C)$.
        The assumption $C\notin \cV$ implies $C'\notin \cU$, so $C' \in \dep(A)$ holds.
        Since $C$ is a child of $C'$, we have $C \in \dep(A)$ and $\plen_{A'}(C)=\len(C') \le \plen_A(C)$.

        By \cref{lem:lenplen}, we have $\len(D)\le \plen_{A'}(C)$ and $\len(C)<\plen_B(D)$.
        Due to $\plen_{A'}(C)\le \plen_A(C)$, we conclude that $\len(D)\le \plen_A(C)$ and $\len(C)<\plen_B(D)$; thus, \cref{lem:lenplen} implies $(C,D)\in \closure{A,B}$, contradicting the choice of $(C,D)$.
    \end{claimproof}

    Next, we bound the contribution of every $C\in \cV$.
    \begin{claim}
        Every $C\in \S_\slp$ satisfies \[\sum_{D : (C,D)\in\closure{A',B}} \left(\len(C)+\len(D)\right)=\begin{cases} \Oh\left(\len(A')+\len(B)\cdot \log\left(1+\tfrac{\len(B)}{\len(A')}\right)\right)& \text{if $C=A'$,}\\
        \Oh(\len(A')+\len(B))& \text{otherwise}.\end{cases}\]
    \end{claim}
    \begin{claimproof}
        Consider an arbitrary $D \in \dep(B)$.
        By \cref{lem:lenplen}, we have $(C,D)\in \closure{A',B}$ if and only if there is a node $\beta$ of $\Tr(B)$ that satisfies $\symb(\beta)=D$, $\len(\beta)\le \plen_{A'}(C)$, and $\len(C)< \plen(\beta)$. 

        Let us first consider nodes $\beta$ that additionally satisfy $\len(\beta)\le \len(C)$.        
        Each root-to-leaf path in $\Tr(B)$ contributes at most one node $\beta$ satisfying $\plen(\beta)>\len(C)\ge \len(\beta)$. Excluding the root of $\Tr(B)$, each such node $\beta$ satisfies $\len(C)+\len(\beta)<\plen(\beta)+\len(\beta)\le 5\cdot \len(\beta)$ because $\slp$ is balanced.
        Consequently, the total contribution of $\len(C)+\len(\beta)$ across all such non-root nodes $\beta$ does not exceed $5\cdot \len(B)$. 
        The contribution of the root of $\Tr(B)$ does not exceed $\len(C)+\len(B)$. 

        It remains to focus on the nodes $\beta$ satisfying $\len(C)<\len(\beta) \le \plen_{A'}(C)$.
        Let us first assume that $C\ne A'$ so $\plen_{A'}(C)\le 4\cdot \len(C)$ holds because $\slp$ is balanced.
        Since the lengths of subsequent nodes on a single root-to-leaf path are smaller by a factor of at least $\frac43$, the number of such nodes $\beta$ on the path does not exceed $\ceil{\log_{4/3}(\plen_{A'}(C)/\len(C))}\le \ceil{\log_{4/3} 4}=5$.
        Each such node $\beta$ satisfies $\len(C)+\len(\beta)<2\cdot \len(\beta)$, so the total contribution of $\len(C)+\len(\beta)$ across all such nodes does not exceed $10\cdot \len(B)$. 
        Overall, if $C\ne A'$, then the total contribution of all nodes $\beta$ is at most $\len(C)+15\cdot\len(B)=\Oh(\len(A')+\len(B))$.

        Next, suppose that $C=A'$. In this case, the condition $\plen_{A'}(C)\ge \len(\beta)> \len(C)$ is equivalent to $\len(B)\ge \len(\beta)>\len(A')$.
        Since the lengths of subsequent nodes on a single root-to-leaf path are smaller by a factor of at least $\frac43$, the number of such nodes $\beta$ on the path does not exceed $\ceil{\log_{4/3}(\len(B)/\len(A'))}$.
        Each such node $\beta$ satisfies $\len(C)+\len(\beta)<2\cdot \len(\beta)$, so the total contribution of $\len(C)+\len(\beta)$ across all such nodes does not exceed $2\cdot \len(B)\cdot \ceil{\log_{4/3}(\len(B)/\len(A'))}$. 
        Overall, if $C=A'$, then the total contribution of all nodes $\beta$ is at most $\Oh(\len(A')+\len(B) \log(1+\len(B)/\len(A')))$.
    \end{claimproof}
    
    Combining the two claims with $|\cV|=\Oh(u+1)$, we conclude that \[\sum_{(C,D)\in \closure{A',B}\setminus \closure{A,B}}\!\!\!\left(\len(C)+\len(D)\right) = \Oh\left((u+1)(\len(A')+\len(B))+\len(B)\log\left(1+\tfrac{\len(B)}{\len(A')}\right)\right)\qedhere.\]
\end{proof}

We are now ready to prove \cref{lm:boxdsslp-update}.

\lmboxdsslpupdate*
\begin{proof}
We show how to obtain $\boxdsslp(A', B)$ from $\boxdsslp(A, B)$.
Then $\boxdsslp(A', B')$ can be obtained from $\boxdsslp(A', B)$ analogously.

We first describe a na\"ive procedure that builds $\boxdsslp(A', B)$.
We assume that $\len(A')\ge \len(B)$; the case $\len(A') < \len(B)$ is analogous.
If $\len(A') = \len(B) = 1$, the $\mds(\BMw(A, B))$ data structure can be built trivially in constant time.
Otherwise, we recursively build $\boxdsslp(A'_L, B)$ and $\boxdsslp(A'_R, B)$, where $\rhs(A')\eqqcolon A'_LA'_R$ and store pointers to them.
It remains to build $\mds(\BMw(A', B))$.
For that, we use the references to $\mds(\BMw(A'_L, B))$ and $\mds(\BMw(A'_R, B))$ and \cref{boundary-distance-matrix-combination}.

The na\"ive procedure for building $\boxdsslp(A', B)$ described above may be inefficient.
To improve it, note that if $\Tr(A')$ and $\Tr(A)$ share a lot of symbols, many of the recursive calls $(C, D)$ we compute are already computed in the recursive structure of $\boxdsslp(A, B)$.
If this is the case, we may just find a reference to $\boxdsslp(C, D)$ and terminate immediately.
To build $\boxdsslp(A', B)$, we have to build $\boxdsslp(C, D)$ for all $(C, D) \in \closure{A', B}$.
However, if $(C, D) \in \closure{A, B}$, we already computed $\boxdsslp(C, D)$ and can reuse it.

Using \cref{lem:lenplen}, we can determine whether a given recursive call $\boxdsslp(C, D)$ was already computed, and if so, find it.
For that, we traverse $\Tr(A)$ and store $\plen_A(C)$ for every $C \in \dep(A)$ in a dictionary implemented as a binary search tree in $\Oh(n \log n)$ time.
Furthermore, for every $C \in \dep(A) \setminus \set{A}$, we store in the dictionary a reference to the dictionary entry corresponding to $\pre_A(C)$.
We build an analogous data structure for $B$.
Now given some $(C, D) \in \S^2_\slp$, we can decide whether $(C, D) \in \closure{A, B}$ in $\Oh(\log n)$ time using the first claim of \cref{lem:lenplen}.
If so, we further have to find a pointer to $\boxdsslp(C, D)$.
For that, we iteratively use the second claim of \cref{lem:lenplen} to compute a sequence $((C_i, D_i))_{i \in \fragmentcc{1}{\ell}}$, where $(C_1, D_1) = (C, D)$, $(C_{\ell}, D_{\ell}) = (A, B)$, $(C_i, D_i) \in \closure{A, B}$ for all $i \in \fragmentcc{1}{\ell}$, and $\boxdsslp(C_{i}, D_{i})$ is a child recursive call of $\boxdsslp(C_{i+1}, D_{i+1})$.\footnote{The last property can be verified by analyzing the second claim of \cref{lem:lenplen} as $\plen_A(C_i) \ge \len(D_i)$ and $\plen_B(D_i) > \len(C_i)$.}
Furthermore, as $\slp$ is weight balanced, we have $\log \len(C_{i+1}) + \log \len(D_{i+1}) \ge \log \len(C_i) + \log \len(D_i) + \log (4 / 3)$, and thus $\ell = \Oh(\log n)$.
Consequently, the sequence $((C_i, D_i))_{i \in \fragmentcc{1}{\ell}}$ can be built in $\Oh(\log n + \ell) = \Oh(\log n)$ time.
After building the sequence, we can start from $\boxdsslp(C_{\ell}, D_{\ell}) = \boxdsslp(A, B)$ and descend down the recursive structure to $\boxdsslp(C_1, D_1) = \boxdsslp(C, D)$ in $\Oh(\ell) = \Oh(\log n)$ time.

We are now ready to alternate our algorithm for building $\boxdsslp(A', B)$.
For every recursive call $(C, D)$, in $\Oh(\log n)$ time we check whether $(C, D) \in \closure{A, B}$, and if so, in $\Oh(\log n)$ time find a pointer to $\boxdsslp(C, D)$.
Otherwise, we recurse and compute $\boxdsslp(C, D)$.

We now analyze the time complexity of the described algorithm.
If $(C, D) \notin \closure{A, B}$, the local work in the recursive call for $(C, D)$ is $\Oh(W (\len(C) + \len(D)) \log n)$ dominated by the application of \cref{boundary-distance-matrix-combination}.
Otherwise, we spend $\Oh(\log n)$ time to find a reference to $\boxdsslp(C, D)$, which can be charged to the parent recursive call if $(C, D) \neq (A', B)$.
In total, the time complexity of building $\boxdsslp(A', B)$ is $\Oh(L W \log n + (K + 1) \log n) = \Oh((L + 1) W \log n)$, where $L \coloneqq \sum_{(C, D) \in \closure{A',B} \setminus \closure{A, B}} (\len(C) + \len(D))$ and $K \coloneqq |\closure{A',B} \setminus \closure{A, B}| \le L$.

Due to \cref{cor:newincl}, the whole recursive procedure of building $\boxdsslp(A', B)$ takes time $\Oh((u + \log n) W n \log n)$.
We build $A'$ in $\Oh(\log n)$ time and the dictionaries storing $\plen_A(C)$ and $\plen_B(D)$ in $\Oh(n \log n)$ time.
Therefore, the total time complexity of building $\boxdsslp(A', B)$ is $\Oh((u + \log n) W n \log n)$.

For the relaxed version of the hierarchical alignment data structure, we replace the applications of the first part of \cref{boundary-distance-matrix-combination} with the applications of the second part of \cref{boundary-distance-matrix-combination}.
Hence, the total time complexity is $\Oh((L+1) \sqrt k \log n) = \Oh((u + \log n) n \sqrt k \log n)$.
Due to \cref{boundary-distance-matrix-combination}, the matrix $M$ for which we store a pointer to $\mds(M)$ satisfies $M \meq{k} \BMw(A', B)$ and $\coresiz(M) = \Oh((\len(A') + \len(B)) \sqrt k)$.
\end{proof}

\end{document}